\documentclass[aps,twocolumn,showpacs,longbibliography,superscriptaddress]{revtex4-1}
\pdfoutput=1
\usepackage[utf8]{inputenc}
\usepackage[english]{babel}
\usepackage[T1]{fontenc}
\usepackage{amsmath}
\usepackage{xcolor}
\colorlet{myPurple}{blue!40!red}
\colorlet{myCyan}{cyan!60!gray}
\colorlet{myRed}{blue!55!gray}
\usepackage{tikz}
\usepackage{pgfplots}
\pgfplotsset{compat=1.14}
\usepackage[colorlinks=true,citecolor=myRed,urlcolor=myRed,linkcolor=myRed]{hyperref}
\usepackage{exscale}
\usepackage{bbm}
\usepackage{graphicx}
\usepackage{amsmath}
\usepackage{latexsym}
\usepackage{amsfonts}
\usepackage{amssymb}
\usepackage{times}
\usepackage[T1]{fontenc}
\usepackage{lipsum}
\usepackage{amsthm}
\usepackage{enumerate}
\usepackage{bbold}
\usepackage{color}
\usepackage{nicefrac}
\usepackage{changes}
\usepackage{bm}
\usepackage{dsfont}
\usepackage{amsthm}
\usepackage{array}
\usepackage{cancel}
\usepackage[toc,page]{appendix}
\usepackage{multirow}
\usepackage{color}
\usepackage{calrsfs}
\usetikzlibrary{backgrounds,decorations.pathreplacing,calc}

\usepackage{tkz-euclide}
\usepackage{tcolorbox}
\usepackage{mathtools}

\newcommand{\sket}[1]{{\ensuremath{\lvert#1\rangle}}}
\newcommand{\lket}[1]{{\ensuremath{\left\lvert#1\right\rangle}}}
\newcommand{\ket}[1]{\if@display\lket{#1}\else\sket{#1}\fi}

\newcommand{\sbra}[1]{{\ensuremath{\langle#1\rvert}}}
\newcommand{\lbra}[1]{{\ensuremath{\left\langle#1\right\rvert}}}
\newcommand{\bra}[1]{\if@display\lbra{#1}\else\sbra{#1}\fi}

\newcommand{\sbraket}[2]{{\ensuremath{\langle#1\rvert#2\rangle}}}
\newcommand{\lbraket}[2]{{\ensuremath{\left\langle#1\!\left\rvert\vphantom{#1}#2\right.\!\right\rangle}}}
\newcommand{\braket}[2]{\if@display\lbraket{#1}{#2}\else\sbraket{#1}{#2}\fi}

\newcommand{\sketbra}[2]{{\ensuremath{\lvert #1\rangle\!\langle #2\rvert}}}
\newcommand{\lketbra}[2]{{\ensuremath{\left\lvert #1\right\rangle\!\!\left\langle #2\right\rvert}}}
\newcommand{\ketbra}[2]{\if@display\lketbra{#1}{#2}\else\sketbra{#1}{#2}\fi}

\newcommand{\proj}[1]{\ketbra{#1}{#1}}

\newcommand{\tp}{\otimes}
\newcommand{\tr}{\textrm{Tr}}

\newcommand{\rA}{\text{A}}
\newcommand{\rB}{\text{B}}
\newcommand{\rC}{\text{C}}
\newcommand{\rP}{\text{P}}

\newcommand{\M}{\mathsf{M}}
\newcommand{\N}{\mathsf{N}}

\theoremstyle{plain}
\newtheorem{thm}{Theorem}
\newtheorem{lem}{Lemma}
\newtheorem{lemma}{Lemma}

\newtheorem{definition}{Definition}

\newtheorem{prop}[thm]{Proposition}
\newtheorem{cor}[thm]{Corollary}

\DeclarePairedDelimiter{\ceil}{\lceil}{\rceil}

\newcommand{\idd}{\mathds{1}}

\newcommand{\W}{\mathsf{W}}
\newcommand{\X}{\mathsf{X}}
\newcommand{\Y}{\mathsf{Y}}
\newcommand{\Z}{\mathsf{Z}}

\newcommand{\ab}{\mathbf{a}}
\newcommand{\xb}{\mathbf{x}}
\newcommand{\bb}{\mathbf{b}}
\newcommand{\yb}{\mathbf{y}}

\newcommand{\CHSH}{\mathsf{CHSH}}
\newcommand{\Ss}{\mathsf{P}}
\newcommand{\T}{\mathsf{Q}}

\begin{document}

\title{Quantum networks self-test all entangled states}

\author{Ivan \v{S}upi\'{c}}
\affiliation{CNRS, LIP6, Sorbonne Universit\'{e}, 4 place Jussieu, 75005 Paris, France}
\author{Joseph Bowles}
\affiliation{ICFO-Institut de Ciencies Fotoniques, The Barcelona Institute of Science and Technology, 08860 Castelldefels (Barcelona), Spain}
\author{Marc-Olivier Renou}
\affiliation{ICFO-Institut de Ciencies Fotoniques, The Barcelona Institute of Science and Technology, 08860 Castelldefels (Barcelona), Spain}
\author{Antonio Ac\'{i}n}
\affiliation{ICFO-Institut de Ciencies Fotoniques, The Barcelona Institute of Science and Technology, 08860 Castelldefels (Barcelona), Spain}
\affiliation{ICREA - Instituci\'{o} Catalana de Recerca i Estudis Avan\c cats, 08010 Barcelona, Spain}
\author{Matty J. Hoban}
\affiliation{Department of Computing, Goldsmiths, University of London, New Cross, London SE14 6NW, United Kingdom}
\affiliation{Cambridge Quantum Computing Ltd}

\date{\today}

\begin{abstract}
Certifying quantum properties with minimal assumptions is a fundamental problem in quantum information science. 
Self-testing is a method to infer the underlying physics of a quantum experiment only from the measured statistics.  
While all bipartite pure entangled states can be self-tested, little is known about how to self-test quantum states of an arbitrary number of systems. Here, we introduce a framework for network-assisted self-testing and use it to self-test any pure entangled quantum state of an arbitrary number of systems. 
The scheme requires the preparation of a number of singlets that scales linearly with the number of systems, and the implementation of standard projective and Bell measurements, all feasible with current technology.
When all the network constraints are exploited, the obtained self-testing  certification is stronger than what is achievable in any Bell-type scenario. Our work does not only solve an open question in the field, but also shows how properly designed networks offer new opportunities for the certification of quantum phenomena. 
\end{abstract}

\maketitle

The difficulty in classically simulating quantum systems offers radically new avenues for information processing; this difficulty even becomes an impossibility when causality constraints are imposed, such as in a Bell test~\cite{Bell1964}. These new possibilities bring challenges, as experimental demonstrations require a very precise control of the quantum devices involved. It is therefore crucial for the development of quantum information technologies to design tools to certify the correct functioning of complex quantum devices using our limited classical information processing capabilities~\cite{CertReview2020}.

A ubiquitous form of certification is to determine that a system is in a particular quantum state. Standard state tomography~\cite{Fano,Vogel} achieves this by performing measurements on the system to certify and compare the obtained results with the predictions from the Born rule. This method is described as \textit{device-dependent}, as it assumes that measurements are perfectly characterised (an unrealistic assumption in many setups).
Measurements can also be certified device-dependently, through the preparation of, in turn, perfectly characterized quantum states, introducing a form of circularity in the procedure. 
The strongest form of device certification should then minimise the assumptions made: it should be based solely on experimental data and make very few assumptions about the devices involved, without requiring any detailed characterization of them. 
To attain this form of certification the \textit{device-independent} framework~\cite{acin2007device,Pironio09,colbeck}, in which quantum devices are modelled as uncharacterised `black boxes' with only classical interaction (inputs and outputs) with these boxes, offers a solution. 
Being a data-driven framework, in order to certify genuine quantum properties in the device-independent approach it is necessary to observe statistics without any classical analogue (for example correlations violating a Bell inequality). 

Self-testing pushes device-independent quantum certification to its strongest form: modulo some symmetries inherent to the device-independent framework, self-testing protocols certify the precise form of the quantum state and/or quantum measurements only from the statistics they generate \cite{upi2019selftesting}. The concept of self-testing was introduced in~\cite{Mayers2004} (see also~\cite{Mayers98}, and~\cite{summerswerner} for a precursor result) and relies on a standard Bell test in which local measurements are performed on an entangled state. The authors of~\cite{Mayers2004} showed the existence of statistics (correlations) such that black boxes reproducing them must essentially prepare the maximally entangled state $\ket{\phi^+}=\frac{1}{\sqrt{2}}(\ket{00}+\ket{11})$. These Bell correlations, therefore, self-test the state $\ket{\phi^+}$. 
These seminal works opened a new research program: what is the ultimate power of self-testing? In particular, can one design self-testing protocols for any quantum state? This question remains open since more than twenty years.
In a seminal result, Coladangelo \emph{et al.} showed that all pure bipartite states can be self tested~\cite{Coladangelo2017b}. In the general multipartite case however, and despite partial progress for some restricted families of states~\cite{graphstates,Wstate,Wstate2,SCAA,Flavio}, little is known. This should not come as a surprise, as entanglement is much richer and harder to characterize in the multipartite case. 

In this work, we answer this question by providing a self-testing protocol for any pure state of an arbitrary number of systems with arbitrary local dimension. To do so, we introduce the concept of \emph{network assisted self-testing}, in which a network made of non-trusted states and measurements is employed to self-test a target quantum device. We first show how to self-state a generic pure state of $N$ particles using a network structure that utilises $N$ additional maximally entangled states, standard single-particle projection and two-particle Bell measurements; all of which are within experimental reach. This result is obtained under standard-self testing assumptions, i.e.\ without assuming the causal constraints associated to the network geometry. Considering these additional constraints, we also show how our results imply a type of certification that is strictly stronger than any possible in standard self-testing.

\begin{figure}
    \centering
    \includegraphics[scale=0.9]{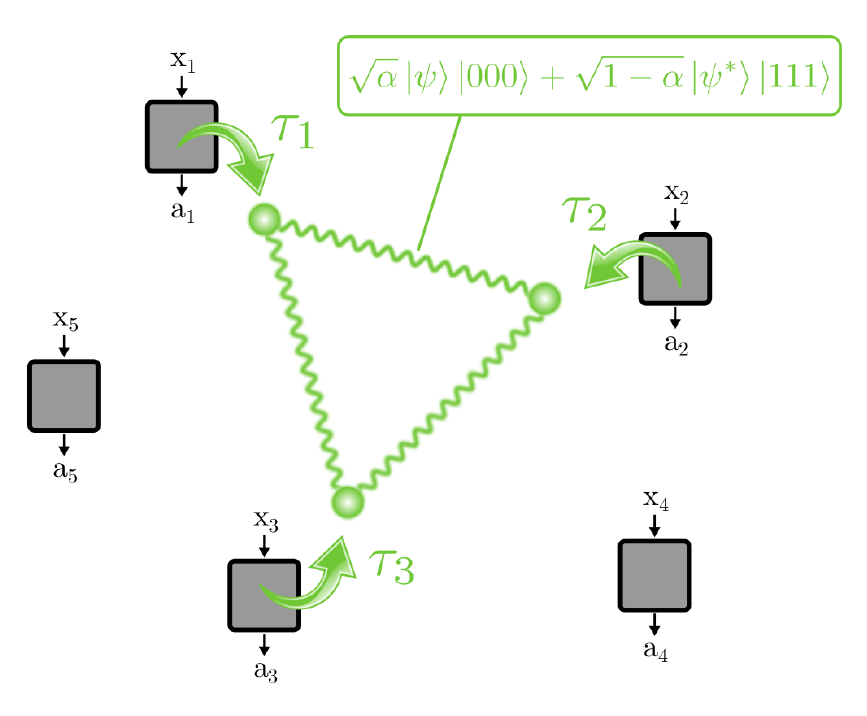}
    \caption{(a) Network-assisted Self-testing in a scenario for $M=5$ devices where the aim is to self-test a three-party state $\ket{\psi}$. We consider that the state $\ket{\psi}$ can be self-tested if there exist quantum correlations $P(a_1...a_5|x_1...x_5)$ between all devices which imply the existence of local channels $\tau_1, ..., \tau_3$ which extract a superposition of $\ket{\psi}$ and $\ket{\psi^*}$ from a subset of three devices. This superposition is ``flagged", with local ancilla qubits which indicate whether the state is conjugated or not, see Equation~\ref{extractedstate}. See the appendices for a comparison with other existing definitions of self-testing.}
    \label{fig:selftesting_definition}
\end{figure}

\emph{Standard Self-testing}---
Self-testing aims at characterising the informational content of quantum devices.  It is defined in a standard Bell test in which an $N$-party state is distributed among $N$ observers who can run $m$ possible measurements of $r$ possible results. We label the choice of measurement by each party by $x_i=1,\ldots,m$, and the obtained result by $a_i=1,\ldots,r$, with $i=1,\ldots,N$. The resulting statistics is described by the conditioned probability distribution $P(a_1...a_N|x_1...x_N)$, which we simply dub as correlations. In a quantum realisation, these observed correlations read
\begin{equation}
\label{eq:qcorr}
    P(a_1...a_N|x_1...x_N)=\tr\left(\rho_N M_{a_1|x_1}\otimes\ldots\otimes M_{a_N|x_N}\right),
\end{equation}
where $\rho_N$ denotes the states shared by the $N$ observers and $M_{a_1|x_1}$ the positive operators defining their local measurements. 
The goal of self-testing is to certify the state and/or measurements in Eq.~\eqref{eq:qcorr} only from the observed correlations. This is done in the device-independent scenario, assuming the validity of Eq.~\eqref{eq:qcorr} but without invoking any assumptions on the state and measurements appearing in it. 
In particular the dimension is not fixed, and as a result the state in Eq.~\eqref{eq:qcorr} is taken to be pure, $\rho_N=\proj{\psi_N}$. It is also standard to assume that measurements are projective as correlations from arbitrary measurements can be simulated by projective measurements. 
In order to arrive at the correlations in Eq.~\eqref{eq:qcorr} there is an assumption that in each run of the experiment the state and measurement operators are the same. 

The device-independent formulation of self-testing has important consequences. First, self-testing is limited to pure states. Second, the state can only be specified modulo some unavoidable symmetries. 
For instance, a local rotation (or unitary) of one part of the state, compensated by a rotation of the corresponding measurement device, results in the same statistics. 
Also, a source creating extra unmeasured degrees of freedom produces the same statistics. Moreover, one cannot discriminate between a given quantum realisation, and its complex conjugate: they also lead to the same statistics. The formal definition of state self-testing we use in this work takes into account all these facts as follows. 

\begin{definition}[State Self-testing]\label{def:SelfTest}
A state $\ket{\psi_N}$ can be self-tested if there exist correlations $P(a_1...a_N|x_1...x_N)$ such that for any quantum realisation of them, there exists a set of $N$ local quantum maps $\tau_1,\cdots,\tau_N$ that when applied to the unknown state of the system extract the state 
\begin{align}\label{extractedstate}
\ket{\psi^\alpha_N}=\sqrt{\alpha}\ket{\psi_N}\ket{0}^{\otimes N}+\sqrt{1-\alpha}\ket{\psi_N^*}\ket{1}^{\otimes N}
\end{align}
for some unknown $\alpha\in[0,1]$. 
\end{definition}

Operationally, this definition implies that from the measurement statistics one can conclude that the target state $\ket{\psi_N}$ can be extracted from the underlying state in the experiment. 
This extraction procedure is defined by a set of $N$ local quantum channels, i.e. deterministic quantum operations that can be physically realised.
Importantly, these channels describing the extraction do not need to be implemented. 
All that is required is a proof that such channels exist, implying the target state can be extracted. 
The extraction procedure takes care of the freedom inherent to self-testing in unmeasured additional degrees of freedom or local rotations. 
The symmetry with respect to complex conjugation is reflected by the unknown linear combination of the target state and its complex conjugate with auxiliary systems appearing in the definition of the extracted state $\ket{\psi^\alpha_N}$. 
In fact, the correlations derived when making local measurements $M_{a_i\vert x_i}$ on $\ket{\psi_N}$ can also be obtained by making local measurements $M_{a_i\vert x_i}\otimes\ket{0}\bra{0}+M^*_{a_i\vert x_i}\otimes\ket{1}\bra{1}$ on $\ket{\psi^\alpha_N}$ independently of $\alpha$. 
Extracting state $\ket{\psi^\alpha_N}$ is in principle the most one can hope for when using only the observed correlations in the Bell test. 
Note that in the particular case of bipartite systems, the issue with complex conjugation does not appear because every pure state is real when rotated to its Schmidt decomposition, i.e. one can always take $\alpha=1$ as in~\cite{Coladangelo2017a}.

\begin{figure*}
    \centering
    \includegraphics[width=\textwidth]{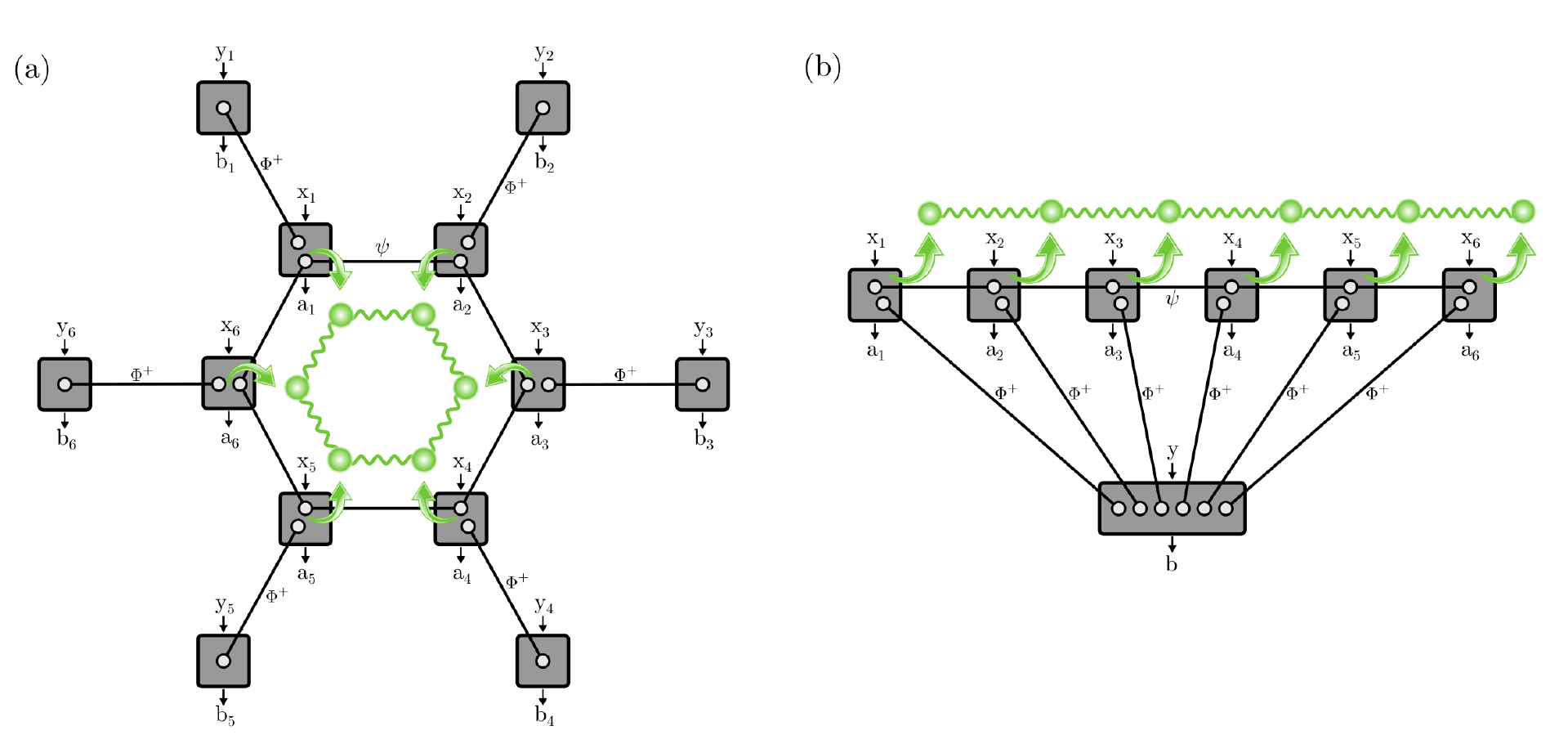}
    \caption{
    The network scenario we use to prove that any pure state $\ket{\psi_N}$ can be self-tested in a network of $M=2N$ parties composed of $N$ main parties $A_1, ..., A_N$, each being assisted by an auxiliary party $B_1, ..., B_N$ (here $N=6$). In order to establish the self-testing correlations, the main parties (holding the state $\ket{\psi}$ to be self-tested) each share an additional maximally entangled state with their corresponding auxiliary party. By making joint measurements at the main parties, one can show that the resulting correlations imply the existence of local channels that extract the state $\ket{\psi_N^\alpha}$ (for some $\alpha$) from these devices. (b) The network scenario we use to prove fully network-assisted self-testing (See Theorem \ref{thm:MainResult2}). Here, we have only a single auxiliary party that receives input $y=(y_1,\cdots, y_N)$ and outputs $b=(b_1,\cdots,b_N)$. The measurement strategy used to establish the self-testing correlations is the same as in (a), where now all the auxiliary parties are grouped together. Assuming the independence of the sources defined by this network structure, we can prove that the extracted state must be $\ket{\psi_N}$ or its complex conjugate.}
    \label{fig:scenarios}
\end{figure*}

\begin{figure*}[hbt!]
    \centering
    \includegraphics[width=\textwidth]{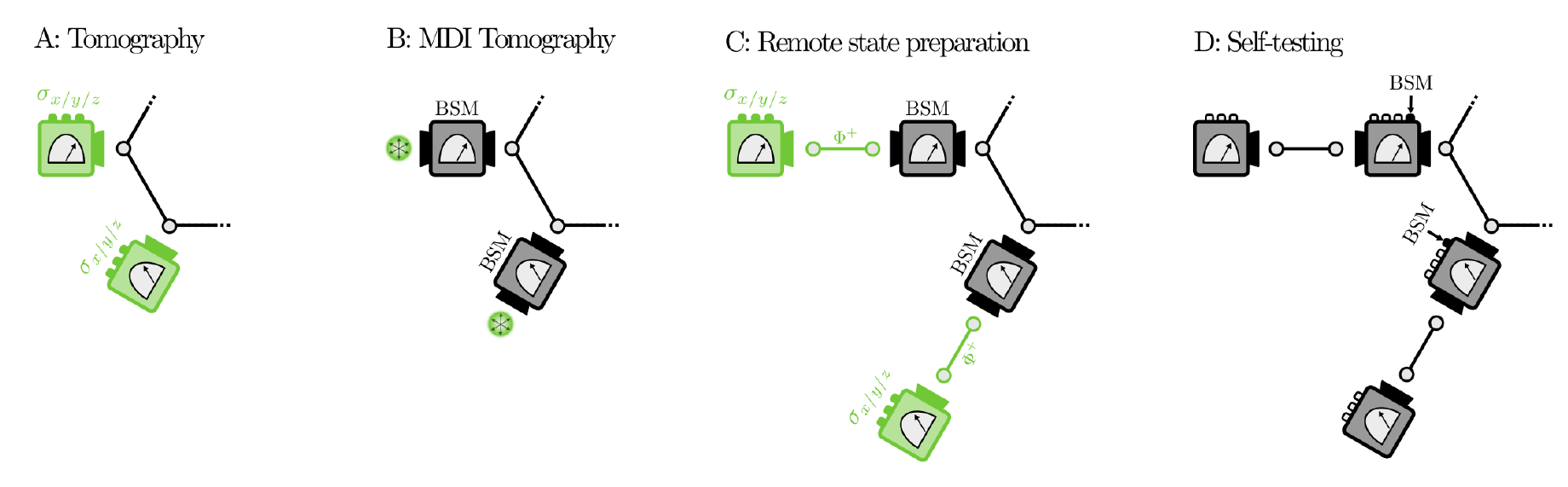}
    \caption{The ingredients required for our self-testing protocols. We start with quantum state tomography (A), and through adapting the set-up we end up in the self-testing scenario (D), where we wish to perform the same task as tomography but through an \emph{uncharacterised} setup. Characterised, or trusted, devices are (in green) are eventually replaced with uncharacterised, or untrusted devices (in black). For illustrative purposes we focus on qubits.\\
(A) Local Tomography: by measuring the three (characterised) Pauli operators $\sigma_x,\sigma_y, \sigma_z$, we characterise a state $\ket{\psi_N}$ shared among $N$ parties (of which two are shown here).\\
(B) Measurement-Device-Independent (MDI) Tomography:
 this modification of tomography removes the assumption that measurement devices are perfectly characterised, and introduce auxiliary systems in informationally-complete single-qubit states, e.g. the eigenstates of the three Pauli operators. If Bell-state measurements are performed, this is equivalent to tomography, but a priori we do not need to assume this \cite{MDIST}.\\
(C) MDI Tomography assisted with Remote State Preparation (RSP):
the auxiliary states used in MDI Tomography can be obtained by measuring one half a maximally-entangled system. The reduced post-measurement state is thus informationally-complete.\\
(D) Network Self-Testing:
the characterised devices in MDI Tomography with RSP are replaced with an initially uncharacterised measurement device and bipartite system shared between an original node and an auxiliary node. The maximally-entangled two-qubit state and Pauli measurements on one qubit are then self-tested; for this we use maximal violations of three different Clauser-Horne-Shimony-Holt inequalities~\cite{Clauser_1969} as shown in the appendices. Thus a priori characterisation is now unnecessary.}
    \label{fig:DITomography}
\end{figure*}

\emph{Network-assisted Self-testing}---
Our work is based on the idea that networks provide new avenues for quantum certification.
For that, we introduce the framework of network-assisted self-testing, and of \textit{fully} network-assisted self-testing, where in the latter case additional network constraints are assumed. 

The general idea of the framework is to construct a network of $M\geq N$ parties and use it to generate statistics that guarantees the presence of the target state $\ket{\psi_N}$ in the network. 
The first variant of the network-assisted framework is very similar to the standard Bell scenario and assumes that the generated statistics has the same form as in Eq.~\eqref{eq:qcorr}, now for a system made of $M$ particles.  

%
\begin{definition}[Network-assisted Self-testing]\label{def:NetSelfTest}
A state $\ket{\psi_N}$ can be self-tested in a network of $M\geq N$ parties if there exist quantum correlations $P(a_1...a_M|x_1...x_M)$ such that, for any $M$-partite quantum realisation of them, there exists a set of $N$ local quantum maps $\tau_1,\cdots,\tau_N$ that when applied to a subset of $N$ parties of the unknown state of the system extract the state 
\begin{align}\label{extractedstatetwo}
\ket{\psi^\alpha_N}=\sqrt{\alpha}\ket{\psi_N}\ket{0}^{\otimes N}+\sqrt{1-\alpha}\ket{\psi_N^*}\ket{1}^{\otimes N}
\end{align}
for some unknown $\alpha\in[0,1]$. 
\end{definition}
This definition is illustrated in Fig~\ref{fig:selftesting_definition}. Note that network-assisted self-testing includes standard self-testing by taking $M=N$. Again, a self-testing statement in terms of $\ket{\psi^\alpha_N}$ is the most one can hope if the only assumption is that the observed correlations have the form of Eq.~\eqref{eq:qcorr}. We can now present the first main result of our work.

\begin{thm}\label{thm:MainResult1}
Any $N$-party pure state $\ket{\psi_N}$ can be self-tested according to Definition \ref{def:NetSelfTest} in a network involving $M=2N$ parties assisted with $N$ maximally entangled states.
\end{thm}

The network used to prove the theorem is shown in Figure~\ref{fig:scenarios}a. It consists of $M=2N$ parties, where each local, main party of $\ket{\psi_N}$ shares a maximally entangled state of dimension $2^k$ with an additional auxiliary party. Here $k$ is the smallest integer such that $2^k$ is larger than the local Hilbert space dimension of $\ket{\psi_N}$. 

While the proof of the theorem is given in the appendices, here we just provide the basic intuition behind the construction of the network-assisted self-testing protocol. The ingredients in the proof are presented in Fig.~\ref{fig:DITomography}. The idea is to turn a standard device-dependent tomography process, where measurements are characterised and trusted, into a self-testing protocol in which the measurements are not characterised. To do so, we introduce $N$ auxiliary parties that each share a maximally entangled state with one of the original $N$ main parties; the main parties will now have two systems. The auxiliary and main parties perform local measurements that self-test their shared maximally entangled state.  The maximally entangled state can now be seen as a resource for essentially teleporting the initial local state of each main party to an auxiliary system. To perform this teleportation,  a Bell measurement can be performed at each of the main parties. In addition to the measurements self-testing the maximally entangled state, the auxiliary parties perform tomographically complete measurements, thus tomography is performed on these teleported states. Recall that tomography defines a one-to-one map between probabilities and states using trusted measurements. Thus, once we remove the trust on the measurements, we now have a one-to-one map between the observed correlations and the state $\ket{\psi_N}$, always up to the relevant symmetries.


Remarkably, the use of networks allows for the introduction of another form of self-testing protocols. This second variant exploits the fact that, in a general network scenario, there may be causal constraints enforced by the network geometry that are not covered by Eq.~\eqref{eq:qcorr}. 
In particular, the network may consist of independent preparations of quantum states, which imply that not only the measurements but also the state in the network has a tensor-product form. Such source independence assumptions are common in the study of network nonlocality, see \cite{BranciardBilocality, FritzBeyondBell, RenouGenuineTriangle,nonindependent}, and \cite{Tavakoli2021reviewNetworkNonloc} for a review. 
For instance, if the network consists of two states prepared by two independent sources, as in a standard entanglement-swapping experiment, the network state is the tensor product of the two preparations. 
These constraints are also natural in the network of Fig.~\ref{fig:selftesting_definition}b but were not used in the proof of Theorem~\ref{thm:MainResult1}, which simply assume the validity of Eq.~\eqref{eq:qcorr}. 
In general, when enforcing also the independent-preparation constraints, the correlations in a network have the form 
\begin{equation}
\label{eq:netcorr}
    P(a_1...a_M|x_1...x_M)=\tr(\bigotimes_i\rho_i \bigotimes_j M_{a_j|x_j}),
\end{equation}
where we slightly abuse the notation and simply represent the independence constraints of the states and measurements in the network by the generic tensor products. 
The fact that the constraints on the observed correlations~\eqref{eq:netcorr} are now stronger than those associated to the Bell tests used in standard self-testing, Eq.~\eqref{eq:qcorr}, opens the possibility of stronger self-testing statements. 
In particular, using state-independence constraints it may be possible to exclude all the fully correlated states in $\ket{\psi^\alpha_N}$ and restrict the extraction process to the two extreme non-correlated cases $\alpha=0,1$, that is, the target state or its complex conjugate. This idea was recently used to prove the necessity of complex numbers in Quantum Theory~\cite{RenouRealQT}.
All these considerations lead to the following stronger definition of network-assisted self-testing, which is impossible when using standard Bell tests.

\begin{definition}[Fully Network-assisted Self-testing]\label{def:FullNetSelfTest}
A state $\ket{\psi_N}$ can be self-tested in a network of $M\geq N$ parties if there exist quantum correlations $P(a_1...a_M|x_1...x_M)$ respecting the network independence constraints such that, for any $M$-partite quantum realisation of them, there exists a set of $N$ local quantum maps $\tau_1,\cdots,\tau_N$ that when applied to a subset of $N$ parties of the unknown state of the system extract either the state $\ket{\psi_N}$ or $\ket{\psi_N^*}$.
\end{definition}

Our second main result is to prove that all states can also be self-tested according to this stronger definition. 

\begin{thm}\label{thm:MainResult2}
Any $N$-party pure state $\ket{\psi_N}$ can be self-tested according to definition \ref{def:FullNetSelfTest} in a network involving $M=N+1$ parties assisted with $N$ maximally entangled states.
\end{thm}

The proof is given in the appendices but the idea behind is intuitive. The network is similar to the one used in the proof of Theorem~\ref{thm:MainResult1}, but now the independence constraints are also enforced upon the states. Because this independence is assumed, in some sense, we not need to separate out each system into an individual node. That is, there is no need to have $N$ auxiliary nodes as in Fig.~\ref{fig:selftesting_definition}, and they can now be grouped into a single node sharing the $N$ maximally entangled states with the $N$ systems in $\ket{\psi_N}$. The resulting network is shown in Fig.~\ref{fig:scenarios}b. 

Let us conclude this letter with several remarks. Our protocol is the same for any state, and consists of a fixed structure to which any pure state can be plugged in. 
The number of auxiliary systems required to self-test an $N$-partite state scales linearly in $N$, and the protocol is composed of the same repeated element for each subsystem. 
Importantly, in practice, the devices utilised in our modular scheme, namely maximally-entangled states and Bell measurements, are some of the most elementary building blocks of the quantum internet. Hence, our protocol is particularly appealing for practical purposes. 
Although we have not proven the noise tolerance of this scheme, future work could utilise standard techniques in self-testing to prove its robustness. 
In one particular direction, it would be useful to adapt the well known numerical SWAP technique to our scenario~\cite{swap}, so that numerical bounds on robustness could be found. Finally, our method gives a rather generic way of translating device-dependent schemes to their device-independent counterparts, in our case from tomography to self-testing protocols, that could be generalised to other certification procedures~\cite{badescu,BancalSelfTestEntMeas,RenouSelfTestEntdMeas}. This could then lead to new device-independent protocols for quantum key distribution and verifiable distributed quantum computation \cite{ruv,verifierleash}, especially those requiring multipartite states \cite{mckaguebqp}.

~\\
\textbf{Acknowledgements}
I.\v{S}. acknowledges funding from ERC grant QUSCO. M.-O.R. is supported by the supported by the grant PCI2021-122022-2B financed by MCIN/AEI/10.13039/501100011033 and by the European Union NextGenerationEU/PRTR and acknowledge the Swiss National Fund Early Mobility Grants P2GEP2\_19144. M.J.H. acknowledges the FQXi large grant ``The Emergence of Agents from
Causal Order" (FQXi FFF Grant number FQXi-RFP-1803B). We acknowledge support from the Government of Spain (FIS2020-TRANQI, Retos Quspin and Severo Ochoa CEX2019-000910-S), Fundacio Cellex, Fundacio Mir-Puig, Generalitat de Catalunya (CERCA, AGAUR SGR 1381 and QuantumCAT), the ERC AdG CERQUTE, the AXA Chair in Quantum Information Science and the Swiss NCCR SwissMap.




\bibliographystyle{alphaurl}
\bibliography{multipartite}

\newcommand{\etalchar}[1]{$^{#1}$}
\begin{thebibliography}{TPKLR21}

\bibitem[ABG{\etalchar{+}}07]{acin2007device}
Antonio Ac{\'\i}n, Nicolas Brunner, Nicolas Gisin, Serge Massar, Stefano
  Pironio, and Valerio Scarani.
\newblock Device-independent security of quantum cryptography against
  collective attacks.
\newblock {\em Physical Review Letters}, 98(23):230501, 2007.
\newblock \href {https://doi.org/10.1103/PhysRevLett.98.230501}
  {\path{doi:10.1103/PhysRevLett.98.230501}}.

\bibitem[APVW16]{APVW16}
Antonio Ac\'{\i}n, Stefano Pironio, Tam\'as V\'ertesi, and Peter Wittek.
\newblock Optimal randomness certification from one entangled bit.
\newblock {\em Phys. Rev. A}, 93:040102, Apr 2016.
\newblock \href {https://doi.org/10.1103/PhysRevA.93.040102}
  {\path{doi:10.1103/PhysRevA.93.040102}}.

\bibitem[BA{\v{S}}{\etalchar{+}}20]{Flavio}
Flavio Baccari, Remigiusz Augusiak, Ivan {\v{S}}upi\'{c}, Jordi Tura, and
  Antonio Ac\'{\i}n.
\newblock Scalable {B}ell inequalities for qubit graph states and robust
  self-testing.
\newblock {\em Phys. Rev. Lett.}, 124:020402, Jan 2020.
\newblock \href {https://doi.org/10.1103/PhysRevLett.124.020402}
  {\path{doi:10.1103/PhysRevLett.124.020402}}.

\bibitem[Bel64]{Bell1964}
John~Stewart Bell.
\newblock {On the Einstein Podolsky Rosen paradox}.
\newblock {\em Physics}, 1:195--200, 1964.
\newblock \href {https://doi.org/10.1103/PhysicsPhysiqueFizika.1.195}
  {\path{doi:10.1103/PhysicsPhysiqueFizika.1.195}}.

\bibitem[BGP10]{BranciardBilocality}
Cyril Branciard, Nicolas Gisin, and Stefano Pironio.
\newblock Characterizing the nonlocal correlations created via entanglement
  swapping.
\newblock {\em Phys. Rev. Lett.}, 104:170401, Apr 2010.
\newblock \href {https://doi.org/10.1103/PhysRevLett.104.170401}
  {\path{doi:10.1103/PhysRevLett.104.170401}}.

\bibitem[BOW19]{badescu}
Costin B\u{a}descu, Ryan O'Donnell, and John Wright.
\newblock Quantum state certification.
\newblock In {\em Proceedings of the 51st Annual ACM SIGACT Symposium on Theory
  of Computing}, STOC 2019, page 503–514, New York, NY, USA, 2019.
  Association for Computing Machinery.
\newblock \href {https://doi.org/10.1145/3313276.3316344}
  {\path{doi:10.1145/3313276.3316344}}.

\bibitem[B{\v{S}}CA18a]{Bowles_2018PRL}
Joseph Bowles, Ivan {\v{S}}upi\'{c}, Daniel Cavalcanti, and Antonio Ac\'{i}n.
\newblock Device-independent entanglement certification of all entangled
  states.
\newblock {\em Physical Review Letters}, 121(18), Oct 2018.
\newblock \href {https://doi.org/10.1103/physrevlett.121.180503}
  {\path{doi:10.1103/physrevlett.121.180503}}.

\bibitem[B{\v{S}}CA18b]{Bowles_2018}
Joseph Bowles, Ivan {\v{S}}upi\'{c}, Daniel Cavalcanti, and Antonio Ac\'{i}n.
\newblock Self-testing of pauli observables for device-independent entanglement
  certification.
\newblock {\em Physical Review A}, 98(4), Oct 2018.
\newblock \href {https://doi.org/10.1103/physreva.98.042336}
  {\path{doi:10.1103/physreva.98.042336}}.

\bibitem[BSS18]{BancalSelfTestEntMeas}
Jean-Daniel Bancal, Nicolas Sangouard, and Pavel Sekatski.
\newblock Noise-resistant device-independent certification of bell state
  measurements.
\newblock {\em Phys. Rev. Lett.}, 121:250506, Dec 2018.
\newblock \href {https://doi.org/10.1103/PhysRevLett.121.250506}
  {\path{doi:10.1103/PhysRevLett.121.250506}}.

\bibitem[CGJV19]{verifierleash}
Andrea Coladangelo, Alex~B. Grilo, Stacey Jeffery, and Thomas Vidick.
\newblock Verifier-on-a-leash: New schemes for verifiable delegated quantum
  computation, with quasilinear resources.
\newblock In Yuval Ishai and Vincent Rijmen, editors, {\em Advances in
  Cryptology -- EUROCRYPT 2019}, pages 247--277, Cham, 2019. Springer
  International Publishing.

\bibitem[CGS17]{Coladangelo2017b}
Andrea Coladangelo, Koon~Tong Goh, and Valerio Scarani.
\newblock All pure bipartite entangled states can be self-tested.
\newblock {\em Nature Communications}, 8:15485, may 2017.
\newblock \href {https://doi.org/10.1038/ncomms15485}
  {\path{doi:10.1038/ncomms15485}}.

\bibitem[CHSH69]{Clauser_1969}
John~F. Clauser, Michael~A. Horne, Abner Shimony, and Richard~A. Holt.
\newblock Proposed experiment to test local hidden-variable theories.
\newblock {\em Phys. Rev. Lett.}, 23:880--884, Oct 1969.
\newblock \href {https://doi.org/10.1103/PhysRevLett.23.880}
  {\path{doi:10.1103/PhysRevLett.23.880}}.

\bibitem[Col06]{colbeck}
Roger Colbeck.
\newblock {\em Quantum And Relativistic Protocols For Secure Multi-Party
  Computation}.
\newblock PhD thesis, University of Cambridge, 2006.
\newblock URL: \url{https://arxiv.org/abs/0911.3814}.

\bibitem[Col17]{Coladangelo2017a}
Andrea Coladangelo.
\newblock Parallel self-testing of (tilted) {EPR} pairs via copies of (tilted)
  {CHSH} and the magic square game.
\newblock {\em Quantum Information and Computation}, 17(9-10):831--865, 2017.

\bibitem[EHW{\etalchar{+}}20]{CertReview2020}
Jens Eisert, Dominik Hangleiter, Nathan Walk, Ingo Roth, Damian Markham, Rhea
  Parekh, Ulysse Chabaud, and Elham Kashefi.
\newblock Quantum certification and benchmarking.
\newblock {\em Nature Reviews Physics}, 2(7):382--390, 2020.
\newblock \href {https://doi.org/10.1038/s42254-020-0186-4}
  {\path{doi:10.1038/s42254-020-0186-4}}.

\bibitem[Fan57]{Fano}
Ugo Fano.
\newblock Description of states in quantum mechanics by density matrix and
  operator techniques.
\newblock {\em Rev. Mod. Phys.}, 29:74--93, Jan 1957.
\newblock \href {https://doi.org/10.1103/RevModPhys.29.74}
  {\path{doi:10.1103/RevModPhys.29.74}}.

\bibitem[Fri15]{FritzBeyondBell}
Tobias Fritz.
\newblock Beyond {B}ell’s {T}heorem {II}: Scenarios with arbitrary causal
  structure.
\newblock {\em Communications in Mathematical Physics}, 341(2):391–434, Nov
  2015.
\newblock \href {https://doi.org/10.1007/s00220-015-2495-5}
  {\path{doi:10.1007/s00220-015-2495-5}}.

\bibitem[Hil07]{Hildebrand}
Roland Hildebrand.
\newblock Positive partial transpose from spectra.
\newblock {\em Physical Review A}, 76(052325), 2007.
\newblock \href {https://doi.org/10.1103/PhysRevA.76.052325}
  {\path{doi:10.1103/PhysRevA.76.052325}}.

\bibitem[JP18]{Johnston}
Nathaniel Johnston and Everett Patterson.
\newblock The inverse eigenvalue problem for entanglement witnesses.
\newblock {\em Linear Algebra and its Applications}, 550:1 -- 27, 2018.
\newblock \href {https://doi.org/https://doi.org/10.1016/j.laa.2018.03.043}
  {\path{doi:https://doi.org/10.1016/j.laa.2018.03.043}}.

\bibitem[McK11]{graphstates}
Matthew McKague.
\newblock Self-testing graph states.
\newblock In {\em Revised Selected Papers of the 6th Conference on Theory of
  Quantum Computation, Communication, and Cryptography - Volume 6745}, TQC
  2011, page 104–120, Berlin, Heidelberg, 2011. Springer-Verlag.
\newblock \href {https://doi.org/10.1007/978-3-642-54429-3_7}
  {\path{doi:10.1007/978-3-642-54429-3_7}}.

\bibitem[McK16]{mckaguebqp}
Matthew McKague.
\newblock Interactive proofs for $\mathsf{BQP}$ via self-tested graph states.
\newblock {\em Theory of Computing}, 12(3):1--42, 2016.
\newblock \href {https://doi.org/10.4086/toc.2016.v012a003}
  {\path{doi:10.4086/toc.2016.v012a003}}.

\bibitem[MM10]{sixstate}
Matthew McKague and Michele Mosca.
\newblock Generalized self-testing and the security of the 6-state protocol.
\newblock In {\em Conference on Quantum Computation, Communication, and
  Cryptography}, pages 113--130. Springer, 2010.

\bibitem[MM11]{McKagueMosca}
Matthew McKague and Michele Mosca.
\newblock Generalized self-testing and the security of the 6-state protocol.
\newblock {\em Lecture Notes in Computer Science}, page 113–130, 2011.
\newblock \href {https://doi.org/10.1007/978-3-642-18073-6_10}
  {\path{doi:10.1007/978-3-642-18073-6_10}}.

\bibitem[MY98]{Mayers98}
Dominic Mayers and Andrew Yao.
\newblock Quantum cryptography with imperfect apparatus.
\newblock {\em Proceedings 39th Annual Symposium on Foundations of Computer
  Science (Cat. No.98CB36280)}, 1998.
\newblock \href {https://doi.org/10.1109/sfcs.1998.743501}
  {\path{doi:10.1109/sfcs.1998.743501}}.

\bibitem[MY04]{Mayers2004}
Dominic Mayers and Andrew Yao.
\newblock Self testing quantum apparatus.
\newblock {\em Quantum Info. Comput.}, 4:273, 2004.

\bibitem[MYS12]{McKague2012}
Matthew McKague, Tzyh~Haur Yang, and Valerio Scarani.
\newblock Robust self-testing of the singlet.
\newblock {\em Journal of Mathematical Physics}, 45(45):455304, 2012.

\bibitem[PAB{\etalchar{+}}09]{Pironio09}
Stefano Pironio, Antonio Ac{\'{\i}}n, Nicolas Brunner, Nicolas Gisin, Serge
  Massar, and Valerio Scarani.
\newblock Device-independent quantum key distribution secure against collective
  attacks.
\newblock {\em New Journal of Physics}, 11(4):045021, apr 2009.
\newblock \href {https://doi.org/10.1088/1367-2630/11/4/045021}
  {\path{doi:10.1088/1367-2630/11/4/045021}}.

\bibitem[PVN14]{Wstate2}
K\'aroly~F. P\'al, Tam\'as V\'ertesi, and Miguel Navascu\'es.
\newblock Device-independent tomography of multipartite quantum states.
\newblock {\em Phys. Rev. A}, 90:042340, Oct 2014.
\newblock \href {https://doi.org/10.1103/PhysRevA.90.042340}
  {\path{doi:10.1103/PhysRevA.90.042340}}.

\bibitem[RBB{\etalchar{+}}19]{RenouGenuineTriangle}
Marc-Olivier Renou, Elisa B\"aumer, Sadra Boreiri, Nicolas Brunner, Nicolas
  Gisin, and Salman Beigi.
\newblock Genuine quantum nonlocality in the triangle network.
\newblock {\em Phys. Rev. Lett.}, 123:140401, Sep 2019.
\newblock \href {https://doi.org/10.1103/PhysRevLett.123.140401}
  {\path{doi:10.1103/PhysRevLett.123.140401}}.

\bibitem[RKB18]{RenouSelfTestEntdMeas}
Marc-Olivier Renou, J{\k{e}}drzej Kaniewski, and Nicolas Brunner.
\newblock Self-testing entangled measurements in quantum networks.
\newblock {\em Phys. Rev. Lett.}, 121:250507, Dec 2018.
\newblock \href {https://doi.org/10.1103/PhysRevLett.121.250507}
  {\path{doi:10.1103/PhysRevLett.121.250507}}.

\bibitem[RTW{\etalchar{+}}21]{RenouRealQT}
Marc-Olivier Renou, David Trillo, Mirjam Weilenmann, Thinh~P. Le, Armin
  Tavakoli, Nicolas Gisin, Antonio Acín, and Miguel Navascués.
\newblock Quantum theory based on real numbers can be experimentally falsified.
\newblock {\em Nature}, 600(7890):625–629, Dec 2021.
\newblock \href {https://doi.org/10.1038/s41586-021-04160-4}
  {\path{doi:10.1038/s41586-021-04160-4}}.

\bibitem[RUV13]{ruv}
Ben~W. Reichardt, Falk Unger, and Umesh Vazirani.
\newblock Classical command of quantum systems.
\newblock {\em Nature}, 496:456, 2013.
\newblock \href {https://doi.org/10.1038/nature12035}
  {\path{doi:10.1038/nature12035}}.

\bibitem[{\v{S}}B20]{upi2019selftesting}
Ivan {\v{S}}upi{\'{c}} and Joseph Bowles.
\newblock Self-testing of quantum systems: a review, September 2020.
\newblock \href {https://doi.org/10.22331/q-2020-09-30-337}
  {\path{doi:10.22331/q-2020-09-30-337}}.

\bibitem[{\v{S}}BB20]{nonindependent}
Ivan {\v{S}}upi\'{c}, Jean-Daniel Bancal, and Nicolas Brunner.
\newblock Quantum nonlocality in networks can be demonstrated with an
  arbitrarily small level of independence between the sources.
\newblock {\em Phys. Rev. Lett.}, 125:240403, Dec 2020.
\newblock \href {https://doi.org/10.1103/PhysRevLett.125.240403}
  {\path{doi:10.1103/PhysRevLett.125.240403}}.

\bibitem[{\v{S}}CAA18]{SCAA}
Ivan {\v{S}}upi{\'{c}}, Andrea Coladangelo, Remigiusz Augusiak, and Antonio
  Ac{\'{\i}}n.
\newblock Self-testing multipartite entangled states through projections onto
  two systems.
\newblock {\em New Journal of Physics}, 20(8):083041, aug 2018.
\newblock \href {https://doi.org/10.1088/1367-2630/aad89b}
  {\path{doi:10.1088/1367-2630/aad89b}}.

\bibitem[{\v{S}}HCA19]{MDIST}
Ivan {\v{S}}upi\'{c}, Matty~J. Hoban, Laia~Domingo Colomer, and Antonio
  Ac{\'{i}}n.
\newblock {S}elf-testing and certification using trusted quantum inputs, 2019.
\newblock arXiv:1911.09395v1.
\newblock \href {http://arxiv.org/abs/1911.09395} {\path{arXiv:1911.09395}}.

\bibitem[SU01]{Seevinck2001}
Michael Seevinck and Jos Uffink.
\newblock Sufficient conditions for three-particle entanglement and their tests
  in recent experiments.
\newblock {\em Phys. Rev. A}, 65:012107, Dec 2001.
\newblock \href {https://doi.org/10.1103/PhysRevA.65.012107}
  {\path{doi:10.1103/PhysRevA.65.012107}}.

\bibitem[SW87]{summerswerner}
Stephen~J. Summers and Reinhard Werner.
\newblock Maximal violation of {B}ell's inequalities is generic in quantum
  field theory.
\newblock {\em Communications in Mathematical Physics}, 110(2):247--259, 1987.
\newblock \href {https://doi.org/10.1007/BF01207366}
  {\path{doi:10.1007/BF01207366}}.

\bibitem[TPKLR21]{Tavakoli2021reviewNetworkNonloc}
Armin Tavakoli, Alejandro Pozas-Kerstjens, Ming-Xing Luo, and Marc-Olivier
  Renou.
\newblock Bell nonlocality in networks, 2021.
\newblock \href {http://arxiv.org/abs/2104.10700} {\path{arXiv:2104.10700}}.

\bibitem[VR89]{Vogel}
K.~Vogel and H.~Risken.
\newblock Determination of quasiprobability distributions in terms of
  probability distributions for the rotated quadrature phase.
\newblock {\em Phys. Rev. A}, 40:2847--2849, Sep 1989.
\newblock \href {https://doi.org/10.1103/PhysRevA.40.2847}
  {\path{doi:10.1103/PhysRevA.40.2847}}.

\bibitem[WCY{\etalchar{+}}14]{Wstate}
Xingyao Wu, Yu~Cai, Tzyh~Haur Yang, Huy~Nguyen Le, Jean-Daniel Bancal, and
  Valerio Scarani.
\newblock Robust self-testing of the three-qubit {W} state.
\newblock {\em Phys. Rev. A}, 90:042339, Oct 2014.
\newblock \href {https://doi.org/10.1103/PhysRevA.90.042339}
  {\path{doi:10.1103/PhysRevA.90.042339}}.

\bibitem[YVB{\etalchar{+}}14]{swap}
Tzyh~Haur Yang, Tam\'as V\'ertesi, Jean-Daniel Bancal, Valerio Scarani, and
  Miguel Navascu\'es.
\newblock Robust and versatile black-box certification of quantum devices.
\newblock {\em Phys. Rev. Lett.}, 113:040401, Jul 2014.
\newblock \href {https://doi.org/10.1103/PhysRevLett.113.040401}
  {\path{doi:10.1103/PhysRevLett.113.040401}}.

\end{thebibliography}



\newpage

\onecolumngrid
\appendix

\appendixpage
\addappheadtotoc


In the following, we provide the framework for network-assisted self-testing of all pure multipartite quantum states and explain our proof in appendix~\ref{app:GlobalStrategyNetworkAssisted}. We give a rigorous detailed proof in appendices~\ref{app:Step1NetworkAssisted},~\ref{app:Step2},~\ref{app:Step3NetworkAssisted}. We also explain in appendix~\ref{app:GlobalStrategyFullyNetworkAssisted} how to adapt our ideas to the additional hypothesis of source independence in the network, in order to obtain fully network-assisted self-testing of all pure multipartite quantum states. 
For simplicity, we concentrate on qubit states: extension to qudit states is given in appendix~\ref{app:ExtensionQudits}. 
We also restrict ourselves to pure genuinely multipartite entangled quantum states, as when a pure state is not genuinely multipartite entangled, it can be written as the product of several pure genuinely multipartite entangled states of smaller size, and one can then self test each of these states separately.

\section{Network-assisted self-testing}\label{app:GlobalStrategyNetworkAssisted}

We first introduce in appendix~\ref{app:NotationsNetworkAssisted} the definition of network assisted self-testing, as well as the concepts of (i) ideal reference experiment and (ii) performed physical experiment that produces the self testing correlations.
Then, in appendix~\ref{app:ProofStrategyNetworkAssisted}, we outline the proof of our result, which is decomposed in three steps.

\subsection{Notations, definitions and setup}\label{app:NotationsNetworkAssisted}

~\\
\textbf{Reference and physical experiments.}
The experiment involving the setup  we want to self-test is called the \emph{reference experiment}. It consists of the reference state $\ket{\psi_N}$ to be self-tested (assumed to be pure multipartite genuinely entangled), reference measurements (in our case, they will be Pauli and Bell state measurements), and additional auxiliary states (in our case, they will be several maximally entangled bipartite qubit states). 
The actual experiment performed in the laboratory is named the \emph{physical experiment}, consisting of the physical state $\ket{\Psi}$ (which include all states created in the physical experiment, without assuming independent preparation) and physical measurements. The only assumptions made about the physical experiment is that it can be described by quantum theory, that no signaling is unforced, and that the observed correlations are the same as in the reference experiment.
\\
\\
\textbf{The reference experiment for network assisted self-testing.}
The reference experiment is the ideal experiment which should be done to obtain self-testing, see Figure~\ref{fig:ReferenceExperimentNetworkAssisted}. It involves $2N$ parties, decomposed into $N$ \emph{main parties} labeled $A_1, ..., A_N$ and $N$ \emph{auxiliary parties} labeled $B_1, ..., B_N$, the reference state $\ket{\psi_N}$, and $N$ copies of the Bell state $\ket{\phi^+}=\frac{1}{\sqrt{2}}(\ket{00}+\ket{11})$. In the following, we write $\mathbb{A}\equiv A_1 ... A_N$ and $\mathbb{B}\equiv B_1 ... B_N$. The $N$ main parties are sharing the reference state $\ket{\psi_N}_{\mathbb{A}}$ together, and each main party $A_j$ is associated with her auxiliary party $B_j$, with whom she shares a state $\ket{\phi^+}_{A_jB_j}$.

The main party $A_j$ (resp. auxiliary party $B_j$) receives inputs $x_j\in\{0,...,5,\diamond\}$ (resp. $y_j\in\{0,1,2\}$). On all inputs except $\diamond$, the main party $A_j$ (resp. auxiliary party $B_j$) provides a bit output $a_j\in\{0,1\}$ (resp. $b_j\in\{0,1\}$). 
These outputs are obtained via local Pauli measurements of the pair  maximally violating the 3-CHSH inequalities \cite{Clauser_1969}, composed of the sum of three versions of the same CHSH inequality up to permutations~\cite{Bowles_2018PRL} (see the reference correlations below). More precisely, the reference observables of each auxiliary party $B_j$ are three Pauli observables $\sigma_0:=\sigma_z$, $\sigma_1:=\sigma_x$ and $\sigma_2:=\sigma_y$  and each main party $A_j$ performs six different measurements, corresponding to the reference observable $(\sigma_x\pm \sigma_z)/\sqrt{2}$, $(\sigma_x\pm \sigma_y)/\sqrt{2}$ and $(\sigma_y\pm \sigma_z)/\sqrt{2}$. On the last input $\diamond$, each main party performs a full Bell state measurement in basis $\ket{\phi^\pm}=\frac{1}{\sqrt{2}}(\ket{00}\pm\ket{11}),\ket{\psi^\pm}=\frac{1}{\sqrt{2}}(\ket{01}\pm\ket{10})$, providing two bits of output.
The corresponding measurements projectors of parties $A_j, B_j$ are written $\M'_{a_j|x_j},N'_{b_j|y_j}$. 
~\\
\\
\textbf{The physical experiment for network assisted self-testing.}
The physical experiment consists in the experiments actually performed in the laboratory, seen by an observer, who should be convinced that there exist a channel, defined by a completely-positive trace-preserving (CPTP) map, which can extract the state $\ket{\psi_N}$ from the experimental devices. This physical experiment provides the same correlations as the reference experiment, detailed below. However, the observer does not assume anything about the internal functioning of the physical experiment, not even its network causal structure. They aim to \emph{deduce} the existence of a CPTP map which can extract $\ket{\psi_N}$ from the devices in the physical experiment that give rise to the correlations she observes.

We denote the full state all main and auxiliary parties share as $\ket{\Psi}^{A_1,\cdots,A_N,B_1,\cdots, B_N} \equiv \ket{\Psi}$ and write the measurement operator of $A_j$ on input $x_j$, for output $a_j$ (resp. $B_j$, $y_j$, $b_j$), as a projector $\M_{a_j|x_j}$ (resp. $\N_{b_j|y_j}$). When the output $a_j$ or $b_j$ is a bit (on all inputs except $x_j=\diamond$), we introduce dichotomic observables $\rA_{x_j}^{(j)} = \sum_{a_j}(-1)^{a_j}\M_{a_j|x_j}$ and $\rB_{y_j}^{(j)} = \sum_{b_j}(-1)^{b_j}\N_{b_j|y_j}$. 
\\
\\
\textbf{Reference correlations in the network assisted self-testing experiment.}
The only assumptions made on the physical experiment are that it can be described by quantum mechanics, satisfies the no signalling constraints and achieves the same correlations as in the reference experiment, as follows:
\begin{tcolorbox}[title=Reference correlations in the network-assisted scenario,title filled]
\begin{enumerate}
    \item In between every pair of a main party and its auxiliary party $A_j, B_j$, the following 3-CHSH inequality is maximally violated:
    \begin{equation}\label{eq:3CHSHViolationNetworkAssisted}
        \forall j, \CHSH_j(0,1;0,1)+\CHSH_j(2,3;0,2)+\CHSH_j(4,5;1,2) = 6\sqrt{2},
    \end{equation}
    where $\CHSH_j(x_1,x_2,y_1,y_2)=\langle A^{(j)}_{x_1}B^{(j)}_{y_1}\rangle +\langle A^{(j)}_{x_1}B^{(j)}_{y_2}\rangle+\langle A^{(j)}_{x_2}B^{(j)}_{y_1}\rangle-\langle A^{(j)}_{x_2}B^{(j)}_{x_2}\rangle$.
    \item When all main parties perform the measurement $\diamond$, the statistics correspond to the ideal one obtained via the teleportation of the state $\ket{\psi_N}$ to the auxiliary parties up to a correction unitary indexed by $a_j\in\{00,01,10,11\}$, who perform a complete tomography of the teleported particle with the Pauli measurements:
    \begin{equation}\label{eq:TomographyTeleportedStateNetworkAssisted}
        \forall k \in \{0,1,2\}, \bra{\Psi}\bigotimes_{j=1}^N\M_{a_j|\diamond}^{(j)}\otimes \rB_{k}^{j}\ket{\Psi} = \frac{1}{4^N}\bra{\psi_N}\bigotimes_{j=1}^N U_{a_j}^\dagger\sigma_{k}^{j}U_{a_j}\ket{\psi_N},
    \end{equation}
    where $\sigma_{0} \equiv \sigma_z$, $\sigma_{1} \equiv \sigma_x$, $\sigma_2 \equiv \sigma_y$, $U_{00} \equiv \idd$, $U_{01} \equiv \sigma_z$, $U_{10} \equiv \sigma_x$, $U_{11} \equiv \sigma_z\sigma_x$.
\end{enumerate}
\end{tcolorbox}
~\\
\textbf{Network assisted self-testing in the physical experiment.}
The general aim of self-testing is to show, given that the correlations of the physical experiment are the same as the one of the reference experiment, that one can extract the reference state $\ket{\psi_N}$ from the physical state $\ket\Psi$.
Several nonequivalent definitions can be introduced. In the following, we concentrate on network-assisted self-testing, which we define as:

\begin{definition}[Network-assisted Self-testing]\label{def:NetSelfTest2}
A state $\ket{\psi_N}$ can be self-tested in a network of $M\geq N$ parties if there exist quantum correlations $P(a_1...a_M|x_1...x_M)$ such that, for any $M$-partite quantum realisation of them, there exists a CPTP map $\Phi$ that extract the state $\ket{\psi_N}$ from $\ket\Psi$, up to a globally heralded conjugation: 
\begin{align}\label{eq:extractedstatetwo}
\Phi(\Psi)=\ket{\psi^\alpha_N}=\sqrt{\alpha}\ket{\psi_N}\ket{0}^{\otimes N}+\sqrt{1-\alpha}\ket{\psi_N^*}\ket{1}^{\otimes N}
\end{align}
for some unknown $\alpha\in[0,1]$. 
\end{definition}

If such CPTP map $\Phi$ exists, we say that we can extract the reference state from the physical one. 
Note that other definition of self-testing asks that $M=N$, and that the map $\Phi$ is an isometry. As discussed in the main text, since self-testing statements are built only from the observed correlations, one cannot hope to have a description of the extracted state with respect to the reference state finer than what encapsulated by eq.~\eqref{eq:extractedstatetwo}.

Note also that, from the point of view of the correlations they generate among the $N$ parties, any of the states in eq.~\eqref{eq:extractedstatetwo}, $\ket{\psi^\alpha_N}$ is at least as powerful as the reference state $\ket{\psi_N}$. In fact, any correlations obtained by performing some local measurements, for simplicity just denoted by $M_j$, on $\ket{\psi_N}$, can also be obtained by, first, each party measuring the extra qubit in $\ket{\psi^\alpha_N}$, and then performing measurement $M_j$ (or its complex conjugate $M_j^*$) when the result of the qubit measurement is 0 (or 1).


The first main result of our work is to provide a reference experiment attaining the network assisted self-testing of any $N$-partite pure state $\ket{\psi_N}$.

\begin{figure}
    \centering
    \includegraphics[scale=1.0]{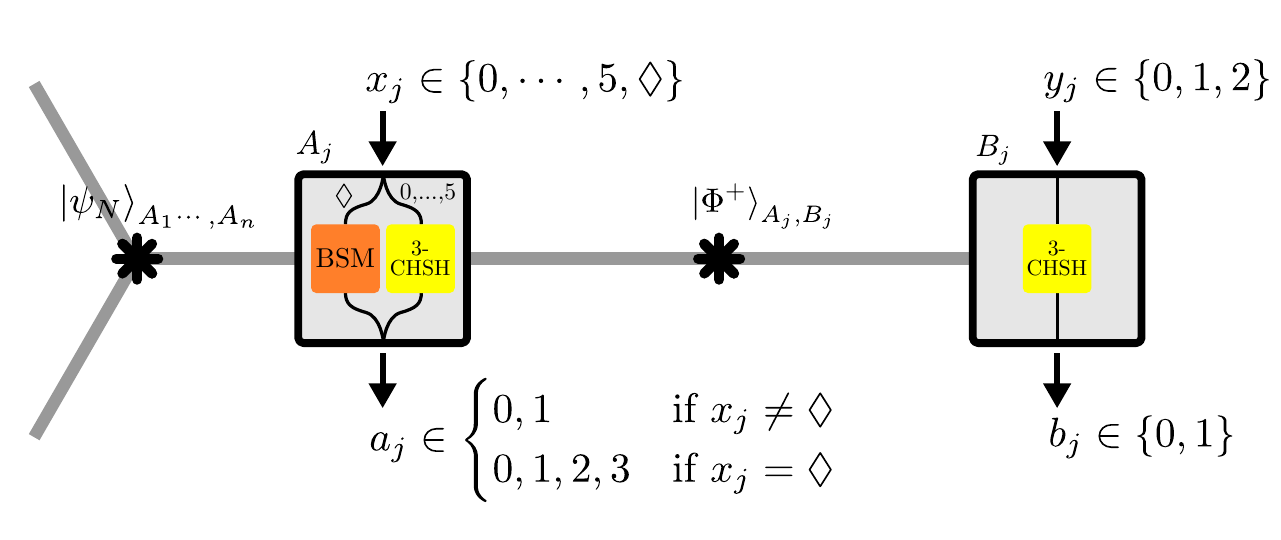}
    \caption{The reference experiment. Every main party shares with the corresponding auxiliary party a maximally entangled pair of qubits. The measurements performed by auxiliary parties can be self-tested. This effectively allows to perform tomography of the state that is teleported to the auxiliary parties when all main parties apply Bell state measurements. }
    \label{fig:ReferenceExperimentNetworkAssisted}
\end{figure}

\subsection{The proof}\label{app:ProofStrategyNetworkAssisted}


This section gives an informal proof of our construction of network assisted self-testing. The proof is rigorously developed in appendices~\ref{app:Step1NetworkAssisted},~\ref{app:Step2},~\ref{app:Step3NetworkAssisted}. 

Note first that in the reference experiment, when all main parties perform the Bell state measurement on inputs $x_j=\diamond$, they teleport the state $\ket{\psi_N}$ to the auxiliary parties, via the $N$ maximally entangled pairs of qubits $\ket{\phi_+}$ they individually share with each auxiliary party.
This teleportation process lies in the heart of our proof, which ultimately consists in self-testing the state teleported to the auxiliary parties.

Our proof consist in three steps.
We first prove Lemma~\ref{lemma:Lemma13CHSHinformal}, which allows one to self tests the tomographic measurement devices of the auxiliary parties, thanks to the maximal violation of the 3-CHSH inequality between each main/auxiliary party pair $A_j, B_j$ (only using the property eq.~\eqref{eq:3CHSHViolationNetworkAssisted} of the correlations).
Importantly, our characterisation of the tomographic measurement devices can only be done up to a complex conjugation, mapping $\sigma_y$ to $\sigma_y^*=-\sigma_y$. We solve this technical difficulty in a second step, with Lemma~\ref{lemma:Lemma2informal}, which characterises how well such an up-to-conjugation tomographic device can specify a state.
At last, a third step proves our main result through Theorem~\ref{thm:TheoremStep3informal}, which uses these characterised tomographic measures to self test the state that is teleported to the auxiliary parties when the main parties perform the teleporting Bell state measurement on input $x_j=\diamond$ (only using the property eq.~\eqref{eq:TomographyTeleportedStateNetworkAssisted} of the correlations).

\subsubsection{Step 1: Self testing the auxiliary parties tomographic devices.}
 
In the first step of our proof, we develop a subroutine to self-test the tomographically complete set of measurements performed by the auxiliary parties $B_j$ 
. This subroutine is actually the self-test of a maximally entangled pair of qubits, shared between each main and its corresponding auxiliary party. 

It is known that the maximal violation of the CHSH inequality is enough to self-test the maximally entangled pair of qubits and a pair of anticommuting observables \cite{summerswerner}. Superimposing three versions of the CHSH inequality into a 3-CHSH inequality allows us to self-test three mutually anticommuting observables \cite{APVW16,Bowles_2018PRL}. 
Such three observables can always be mapped to three Pauli operators, with a caveat that $\sigma_y$ is self-tested up to complex conjugation. More precisely, there exist two extra qubit spaces, one for each party, containing two correlated qubits. Depending on her extra qubit value, a party apply either $\sigma_y$ or $\sigma_y^*=-\sigma_y$. As these two extra qubits are correlated, this eventual conjugation of $\sigma_y$ is performed by the two parties in a correlated way. 

Hence, the observed maximal violation of the 3-CHSH inequality in the physical experiment implies that the auxiliary party measures three Pauli observables on their inputs (up to complex conjugation), independently of the input of the main parties. Note that any other construction self-testing the two-qubit maximally entangled state and the three Pauli measurements, up to complex conjugation, such as that in~\cite{McKagueMosca}, can be used for this part.
In Section~\ref{app:Step1NetworkAssisted}, we formalise this first step in the following lemma, which we write in its informal version:
\begin{tcolorbox}
\begin{lemma}[informal]\label{lemma:Lemma13CHSHinformal}
Let the physical state $\ket{\Psi}$ and dichotomic observables $\{\rA_j\}_{j=0}^5$ and $\{\rB_j\}_{j=0}^2$ reproduce the maximal violation of all the 3-$\CHSH$ inequalities between the $N$ pairs of main and auxiliary parties. 
Then, up to local ancillary systems and a local isometry, 
$$\ket{\Psi} \approx \ket{\mathrm{aux}}\bigotimes_{j=0}^{N-1}\ket{\phi^+}_{A_j'B_{j}'}.$$
Moreover, the physical experiment auxiliary party operators $B^{(j)}_0, B^{(j)}_1$ acts as the reference operators $\sigma^{(j)}_z, \sigma^{(j)}_x$ over the state, whereas $B^{(j)}_2$ acts as the reference operators $\sigma^{(j)}_y$ over the state, up to a $\sigma_z$ correction on the auxiliary system. E.g.: 
\begin{align*}
  B^{(0)}_0\cdot \ket{\Psi} &\approx \ket{\mathrm{aux}}\otimes (\sigma^{(0)}_z\cdot\ket{\phi^+})\bigotimes_{j=1}^{N-1}\ket{\phi^+}  \\
  B^{(0)}_2\cdot \ket{\Psi} &\approx \sigma_z^{(0)}\ket{\mathrm{aux}}\otimes (\sigma^{(0)}_y\cdot\ket{\phi^+})\bigotimes_{j=1}^{N-1}\ket{\phi^+}
\end{align*}
\end{lemma}
\end{tcolorbox}
\begin{proof}
For a technical version of this Lemma, and a proof, see appendix~\ref{app:Step1NetworkAssisted}.
\end{proof}

\subsubsection{Step 2: Tomography up to complex conjugation.}
The previous step allows us to self test the tomographically complete set of Pauli measurements performed by the auxiliary parties, up to a complex conjugation. We aim to use these measurements to completely characterise the state shared between the auxiliary parties when the main parties perform a teleportation Bell state measurement, on inputs $\diamond$. 
However, this requires to understand how much such a tomography-up-to-conjugation device can specify a state. This is the aim of the second step of our proof. In the following (informal) proposition, we prove that when a state is measured with tomography-up-to-conjugation devices as given in the conclusion of Lemma~\ref{lemma:Lemma13CHSHinformal}, it is characterised up to a global conjugation:

\begin{tcolorbox}
\begin{lemma}[informal]\label{lemma:Lemma2informal}
Consider a tomographically complete set of local measurements ${M'}_{a_j|x_j}^{(j)}$. Consider a second set of local measurements $M_{a_j|x_j}^{(j)}$, which are the same as ${M'}_{a_j|x_j}^{(j)}$ up to conjugation, as in Lemma~\ref{lemma:Lemma13CHSHinformal}.

Assume that the targeted state $\ket{\psi_N}$ is genuinely multipartite pure entangled. Let $\rho$ be an  unknown state such that when $\rho$ is measured with $M_{a_j|x_j}^{(j)}$, one observes the same correlations as when $\ket{\psi_N}$ is measured with ${M'}_{a_j|x_j}^{(j)}$.

Then, the purification of $\rho$ must be  
\begin{equation}
    \ket{\tilde{\Psi}} = \ket{\psi_N}\tp\ket{\xi_+} + \ket{\psi_N^*}\tp\ket{\xi_-},
\end{equation}
where $\ket{\psi_N^*}$ is the complex conjugate of $\ket{\psi_N}$, $\braket{\xi_+}{\xi_-} = 0$ and $\|\ket{\xi_+}\| + \|\ket{\xi_-}\| = 1$.
\end{lemma}
\end{tcolorbox}
\begin{proof}
For a technical version of this Lemma, and a proof, see appendix~\ref{app:Step2}.
\end{proof}

Note that to apply this Lemma, we need to assume that the reference state $\ket{\psi_N}$ is genuinely multipartite entangled, i.e. non-product. In case $\ket{\psi_N}$ is separable, one can easily adapt this proposition by applying our construction to  each of the constituents of the non-genuinely multipartite entangled state $\ket{\psi_N}$, obtaining independent potential complex conjugations for each of them.

\subsubsection{Step 3: Network assisted self testing.}
We now have all the tools to perform network assisted self testing in the physical experiment. The last step of our proof focus on the case when all main parties received the input $x_j=\diamond$. In that case, they perform a (uncharacterized) Bell state measurement, which in practice teleportS the (uncharacterized) reference state $\ket{\psi_N}$ to the auxiliary parties.
Then, the auxiliary parties perform a tomographic measurement, characterised up to complex conjugation thanks to Step~1 above. 
Step~2 of our proof allows us to show that the teleported state shared by the auxiliary parties, after the main parties performed their measurement on inputs $x_j=\diamond$, is the state $\ket{\psi_N}$, up to a global conjugation. 

This provides the following informal formulation of our first theorem, which finishes our proof:

\begin{tcolorbox}
\begin{thm}[informal]\label{thm:TheoremStep3informal}
Let the state $\ket{\Psi}$ and dichotomic observables $\{\rA_j\}_{j=0}^5$ and $\{\rB_j\}_{j=0}^2$ reproduce the correlations of the reference experiment. Then, there exists a CPTP map $\Phi$ such that, up to ancillas and junk state:
$$\Phi(\ket{\Psi})=\ket{\psi_N}\otimes\ket{0\cdots 0} + \ket{\psi_N^*}\otimes\ket{1\cdots 1} $$
\end{thm}
\end{tcolorbox}
\begin{proof}
For a technical version of this theorem, and a proof, see appendix~\ref{app:Step3NetworkAssisted}.
\end{proof}

\section{Fully network-assisted self-testing}~\label{app:GlobalStrategyFullyNetworkAssisted}

In this appendix, we explain how our result adapts when one consider the fully network assisted framework, in which we make a supplementary assumption of source independence in the physical experiment. 
Now, we will again start from a reference experiment using state $\ket{\psi_N}$ and extra resources (in our case, the same as in our network assisted protocol), and we aim to show the existence of a CPTP maps $\Phi$ which extracts $\ket{\psi_N}$. 
The existence of $\Phi$ will be proven under the assumptions that the physical experiment can be described by quantum theory with a state created \emph{by independent sources as in the reference experiment}, that no signaling is enforced, and that the observed correlations are the same as in the reference experiment.


\subsection{Notations, definitions and setup}\label{app:NotationsFullyNetworkAssisted}
~\\
\textbf{The reference experiment for fully network assisted self-testing.}
In the case of fully network assisted self-testing, our reference experiment is almost identical to the one for network assisted self-testing. The only difference is that we consider a unique auxiliary party, which takes the role of all the auxiliaries parties in the network assisted self-testing protocol, and that this auxiliary party has two additional inputs, labeled by  $\lozenge, \blacklozenge$, on which he performs a reference measurement consisting of parallel Bell state measurements.
Hence, the reference experiment involves $N+1$ parties, decomposed into $N$ \emph{main parties} labeled $A_1, ..., A_N$ and a unique \emph{auxiliary party} labeled $B$, the reference state $\ket{\psi_N}$, and $N$ copies of the $\ket{\phi^+}=\frac{1}{\sqrt{2}}(\ket{00}+\ket{11})$ state. The $N$ main parties are sharing the reference state $\ket{\psi_N}$ together, and each main party $A_j$ shares a state $\ket{\phi^+}_{A_jB}$ with the auxiliary party $B$.

The main parties $A_j$ have the same inputs and outputs as in the network assisted self-testing reference experiments, and perform the same measurements. 
The auxiliary party $B$ receives an input which is either a vector $\yb= (y_1, \cdots, y_N)$, where $y_j\in\{0,1,2\}$, or  $\lozenge$, or $\blacklozenge$: in total, the auxiliary party has $3^N + 2$ different measurement choices.

On input $\yb = (y_1, \cdots, y_N)$, $B$ outputs a vector $\bb = (b_1, \cdots, b_N)$ where $b_j$ is obtained by measuring his share of $\ket{\phi^+}_{A_jB}$ with the Pauli operator $\sigma_{y_j}$.
On input $\lozenge$, $B$ performs a reference measurement consisting of $\lfloor N/2 \rfloor$ independent Bell state measurements performed on system pairs $(1,2)$, $(3,4)$, $\cdots$, $(2\lfloor N/2 \rfloor-1,~2\lfloor N/2 \rfloor)$, outputting a vector of pairs of bits  $\bb = ((b_1,b_2), (b_{3}, b_{4}), \cdots, (b_{2\lfloor N/2 \rfloor-1}, b_{2\lfloor N/2 \rfloor}))$.
On input $\blacklozenge$, $B$ performs a reference measurement consisting of $\lfloor N/2 \rfloor$ independent Bell state measurements performed on system pairs $(N,1)$, $(2,3)$, $\cdots$, $(2\lfloor N/2 \rfloor-2,~2\lfloor N/2 \rfloor-1)$, outputting a vector of pairs of bits  $\bb = ((b_N,b_1), (b_{2}, b_{3}), \cdots, (b_{2\lfloor N/2 \rfloor-2}, b_{2\lfloor N/2 \rfloor-1}))$. 

%
~\\
\\
\textbf{The physical experiment for fully network assisted self-testing.}
The physical experiment consists in the experiment actually performed in the laboratory, seen by an observer who should be convinced that there exist a CPTP map which can extract the state $\ket{\psi_N}$ from the experimental devices.  
This physical experiment provides the same correlations as the reference experiment, detailed below. As above, the observer does not assume anything about the internal functioning of the physical devices in the physical experiment, but does assume that the states are created according to the network causal structure. She aims to \emph{deduce} the existence of a CPTP map which can extract $\ket{\psi_N}$ from the devices in the physical experiment that give rise to the observed correlations.

With a slight abuse of notation, we again denote the full state all main parties and the auxiliary party share as $\ket{\Psi}^{A_1,\cdots,A_N,B} \equiv \ket{\Psi}$. We use the same notations for $A_j$ measurement operators as in the physical experiment of the network assisted case. We write the measurement operator of $B$ on input $\yb$, for output $\bb$, as projector $N'_{\bb|\yb}$. We also introduce the observable $\rB_{y_j,\yb}^{(j)} = \sum_{\bb}(-1)^{b_j}\N_{\bb|\yb}$.
On input $y\in\{\lozenge, \blacklozenge\}$, for output $b$, we write $B$ measurement operator as projector $N'_{b|y}$.

~\\
\\
\textbf{Reference correlations in the fully network assisted self-testing experiment.}
We assume that the correlations of the physical experiment are the same as in the reference experiment, that is, they are:
\begin{tcolorbox}[title=Reference correlations in the fully network-assisted scenario,title filled]
\begin{enumerate}
    \item In between every main party $A_j$ and the auxiliary party $B$, the following 3-CHSH inequality is maximally violated:
   \begin{align}\label{eq:3CHSHViolationFullyNetworkAssisted}
       \forall \emph{$j$}, \quad& \forall\yb,\yb',\yb''\quad\emph{s.t.}\quad y_j =0,y'_j = 1,y''_j=2, \nonumber\\
       &\CHSH_j(0,1;\yb,\yb')+\CHSH_j(2,3;\yb,\yb'')+\CHSH_j(4,5;\yb',\yb'') = 6\sqrt{2},
    \end{align}
     where $\CHSH_j(x_1,x_2;\yb,\yb')=\langle A^{(j)}_{x_1}B^{(j)}_{y_j,\yb}\rangle +\langle A^{(j)}_{x_1}B^{(j)}_{y_j',\yb'}\rangle+\langle A^{(j)}_{x_2}B^{(j)}_{y_j,\yb}\rangle-\langle A^{(j)}_{x_2}B^{(j)}_{y_j',\yb'}\rangle$.
    \item When the auxiliary party performs measurements denoted with $\lozenge$ or $\blacklozenge$, the state of the corresponding auxiliary parties is projected to one of the four Bell pairs, which imposes perfect (anti)correlation between measurement observables corresponding to the Pauli observables:
    \begin{align} \label{eq:AlignFullyNetworkAssisted}
    &\bra{\Psi} \frac{\rA_{0}^{2j-1}+\rA_{1}^{2j-1}}{\sqrt{2}}\otimes\frac{\rA_{0}^{2j}+\rA_{1}^{2j}}{\sqrt{2}}\otimes\N_{b_{2j-1},b_{2j}|\lozenge}^{(2j-1, 2j)}\ket{\Psi} = \frac{1}{4}(-1)^{b_{2j}},\\
&\bra{\Psi} \frac{\rA_{0}^{2j-1}-\rA_{1}^{2j-1}}{\sqrt{2}}\otimes\frac{\rA_{0}^{2j}-\rA_{1}^{2j}}{\sqrt{2}}\otimes\N_{b_{2j-1},b_{2j}|\lozenge}^{(2j-1, 2j)}\ket{\Psi} = \frac{1}{4}(-1)^{b_{2j-1}},\\
&\bra{\Psi} \frac{\rA_{2}^{2j-1}-\rA_{3}^{2j-1}}{\sqrt{2}}\otimes\frac{\rA_{2}^{2j}-\rA_{3}^{2j}}{\sqrt{2}}\otimes\N_{b_{2j-1},b_{2j}|\lozenge}^{(2j-1, 2j)}\ket{\Psi} = \frac{1}{4}(-1)^{b_{2j-1}\oplus b_{2j}\oplus 1},\qquad \forall j \in \{1,2,\cdots,\lfloor\frac{N}{2}\rfloor\} \\
&\bra{\Psi} \frac{\rA_{0}^{2j}+\rA_{1}^{2j}}{\sqrt{2}}\otimes\frac{\rA_{0}^{2j+1}+\rA_{1}^{2j+1}}{\sqrt{2}}\otimes\N_{b_{2j},b_{2j+1}|\blacklozenge}^{(2j, 2j+1)}\ket{\Psi} = \frac{1}{4}(-1)^{b_{2j+1}},\\
&\bra{\Psi} \frac{\rA_{0}^{2j}-\rA_{1}^{2j}}{\sqrt{2}}\otimes\frac{\rA_{0}^{2j+1}-\rA_{1}^{2j+1}}{\sqrt{2}}\otimes\N_{b_{2j},b_{2j+1}|\blacklozenge}^{(2j, 2j+1)}\ket{\Psi} = \frac{1}{4}(-1)^{b_{2j}},\\
&\bra{\Psi} \frac{\rA_{2}^{2j}-\rA_{3}^{2j}}{\sqrt{2}}\otimes\frac{\rA_{2}^{2j+1}-\rA_{3}^{2j+1}}{\sqrt{2}}\otimes\N_{b_{2j},b_{2j+1}|\blacklozenge}^{(2j, 2j+1)}\ket{\Psi} = \frac{1}{4}(-1)^{b_{2j}\oplus b_{2j+1}\oplus 1}
    \end{align}
    \item When all the main parties perform the measurement $\diamond$ and the auxiliary party does not measure $\lozenge$ or $\blacklozenge$, the statistics correspond to those of the ideal experiment obtained via the teleportation of the state $\ket{\psi_N}$ to the auxiliary party, up to a correction unitary indexed by $a_j\in\{00,01,10,11\}$, who performs a complete tomography of it with the Pauli measurements
      \begin{equation}\label{eq:TomographyTeleportedStateFullyNetworkAssisted}
        \forall k \in \{0,1,2\}, \bra{\Psi}\bigotimes_{j=1}^N\M_{a_j|\diamond}^{(j)}\otimes \rB_{k}^{j}\ket{\Psi} = \frac{1}{4^N}\bra{\psi_N}\bigotimes_{j=1}^N U_{a_j}^\dagger\sigma_{k}^{j}U_{a_j}\ket{\psi_N},
    \end{equation}
    where $\sigma_{0} \equiv \sigma_z$, $\sigma_{1} \equiv \sigma_x$, $\sigma_2 \equiv \sigma_y$, $U_{00} \equiv \idd$, $U_{01} \equiv \sigma_z$, $U_{10} \equiv \sigma_x$, $U_{11} \equiv \sigma_z\sigma_x$.
\end{enumerate}
\end{tcolorbox}
~\\
\textbf{Fully network assisted self-testing in the physical experiment.}
We introduce fully network-assisted self-testing as:

\begin{definition}[Fully network-assisted self-testing]\label{def:FNetSelfTest}
A state $\ket{\psi_N}$ is fully network-assisted self-tested in a network of $M\geq N$ parties if there exist quantum correlations $P(a_1...a_M|x_1...x_M)$ such that, for any $M$-partite quantum realisation of them in a physical experiment satisfying the network structure, there exists a CPTP map $\Phi$ that extract the state 
\begin{align}\label{extractedstatetwo2}
\Phi(\Psi)=\begin{cases} \ket{\psi_N}\\
 \mathrm{~~or}\\
    \ket{\psi_N^*}
    \end{cases}
\end{align}
\end{definition}

If such CPTP map $\Phi$ exists, we say that we can extract the reference state from the physical one. 
Note that thanks to the hypothesis of network structure for the physical experiment, we don't need to introduce a potential global conjugation.

\subsection{The proof}\label{app:ProofStrategyFullyNetworkAssisted}

Our proof of fully network-assisted self-testing is adapted from the one for network-assisted self testing, hence we simply explain the main differences for the two. 

In a first step, see appendix~\ref{app:Step1FullyNetworkAssisted}, we obtain a refined version of Lemma~\ref{lemma:Lemma13CHSHinformal}, in which we take advantage of the additional reference measurements $\lozenge,\blacklozenge$ of the auxiliary party to completely characterize the tomographic measurements, up to a conjugation which now can only be global, happening with probability one. Hence, on input $\yb$, we self test the fact that the auxiliary party is either performing the ideal $N$ independent Pauli measurements, or performing a global conjugation of these ideal $N$ independent Pauli measurements.

Then, the second step is not needed anymore, as the auxiliary party tomographic setup is ideal. It might be that all the operators are complex conjugated, which only corresponds to a convention on the definition of the square root of $-1$ being $+i$ or $-i$. Note that in the setup of Appendix~\ref{app:GlobalStrategyNetworkAssisted}, the situation is different because of the presence of the extra qubits, which somehow determine whether to use the ideal experiment or its complex conjugate.

The last step is straightforward: from completely characterised tomographic devices, we can completely determine the state that is teleported to the auxiliary party after the main parties performed their measurement on inputs $x_j=\diamond$, see~\ref{app:Step3FullyNetworkAssisted}.

\section{Self-testing a tensor product of \textit{N} maximally entangled pairs of qubits and a tomographically complete set of measurements}\label{app:Step1}

In this appendix we prove the procedure for self-testing a tensor product of $N$ maximally entangled pairs of qubits in two scenarios:
\begin{itemize}
    \item in section \ref{app:Step1NetworkAssisted} each pair is shared between two spatially separated and non-communicating parties, in our scenario they correspond to one of the auxiliary and main parties. There are no assumptions about the state shared among the $2N$ parties.
    \item in section \ref{app:Step1FullyNetworkAssisted} one auxiliary party shares with $N$ main parties $N$ maximally entangled states. The network structure is assumed: the $N$ physical states are assumed to be independent. 
\end{itemize}

\subsection{Network-assisted scenario}\label{app:Step1NetworkAssisted}

In this section, we prove the lemma related to the extraction of $N$ maximally entangled pairs of qubits based on $N$ maximal violations of the 3-CHSH inequalities, as given in eqs. \eqref{eq:3CHSHViolationNetworkAssisted}.
\setcounter{lemma}{0}
\begin{tcolorbox}
\begin{lemma}[Formal version, network-assisted case]\label{lemchshNA}
Let the state $\ket{\Psi} \in \mathcal{H}^{\mathbb{A}}\otimes\mathcal{H}^{\mathbb{B}}$ and dichotomic observables $\{\rA_j\}_{j=0}^5$ and $\{\rB_j\}_{j=0}^2$ reproduce the reference correlations given in eqs. \eqref{eq:3CHSHViolationNetworkAssisted}. 
Then there exists a local isometry $V=\bigotimes_{j=0}^{N-1} V_\rA^{(j)}\otimes V_\rB^{(j)}$ such that
\begin{align}\label{josjedna}
    V\left(\ket{\Psi}\otimes\ket{0...0}_{\mathbb{A}'\mathbb{A}''{\mathbb{B}}'{\mathbb{B}}''}\right) &= \ket{\mathrm{aux}}_{\mathbb{A}\mathbb{B}\mathbb{A}''{\mathbb{B}}''}\bigotimes_{j=0}^{N-1}\ket{\phi^+}_{A_j'B_{j}'}\\ \label{MSTNpairs}
    V\left(\bigotimes_{j=1}^N\W_\rC^{(j)}\ket{\Psi}\otimes\ket{0...0}_{\mathbb{A}'\mathbb{A}''{\mathbb{B}}'{\mathbb{B}}''}\right) &= {\sigma_z^{q(w)}}^{C_j''}\ket{\mathrm{aux}}_{\mathbb{A}\mathbb{B}\mathbb{A}''{\mathbb{B}}''}\bigotimes_{j=1}^{N}{\sigma_w}^{C_j'}\ket{\phi^+}_{A_j'B_{j}'},
\end{align}
where $\W \in \{\X,\Y,\Z\}$, $\rC \in \{\rA,\rB\}$ and $q(x) = q(z) = 0$ while $q(y) = 1$.
The junk state $\ket{\mathrm{aux}}$ has the following form:
\begin{equation}\label{eqAuxA1}
\ket{\mathrm{aux}} = \sum_\iota\ket{\xi_\iota}_{\mathbb{A}\mathbb{B}}\ket{\iota}_{\mathbb{A}''}\ket{\iota}_{{\mathbb{B}}''},
\end{equation}
where $\iota \in (0,1)^N$. Through the isometry, the measurements are self-tested up to complex conjugation. For each party, measurement $\Y_\rA$ is self-tested to be either $\sigma_y$, $-\sigma_y$ or some coherent superposition of the two. 
\end{lemma}
\end{tcolorbox}

As described in the main text, every main party can receive one of seven inputs $x \in \{0,1,2,3,4,5,\diamond\}$,while every auxiliary party receives one of three inputs $y \in \{0,1,2\}$. To inputs $x \in \{0,1,\cdots, 5\}$ and all $y$-s,  the parties answer with outputs $a \in \{0,1\}$ and $b \in \{0,1\}$, respectively. To the input $\diamond$, the main parties answer with two-bit outputs $a \in \{00,01,10,11\}$. For the result in this section, the $\diamond$ input is not relevant.

Let us define for $x_j\leq 5$ the operators $\rA^{(j)}_{x_j} = \sum_{a_j}(-1)^{a_j}\M^{(j)}_{a_j|x_j}$, which are valid observables for the main parties.
For the auxiliary parties we analogously define the following operators $\rB_{y_j}^{(j)} = \sum_{b_j}(-1)^{b_j}\N_{b_j|y_j}$.
Let us consider a  pair of parties, $1$-st  main and the $1$-st auxiliary. We check the correlations between the measurement outcomes $a_1$  and $b_{1}$, i.e. the correlations between operators $\rA_{x_1}^{(1)}$ and $\rB_{y_1}^{(1)}$. The observables are measured on the whole state shared among the $2N$ parties $\ket{\Psi}$. The self-test is based on the well-known CHSH inequality.
Namely, two parties can maximally violate the CHSH inequality \cite{Clauser_1969} if and only if they share a maximally entangled pair of qubits and each measures an anticommuting pair of observables. From the observables available to the parties one can define the following three CHSH operators
\begin{align}\label{Mop1}
\mathcal{B}_{1}^{(1)} &=  \rA^{(1)}_{0}\rB^{(1)}_{0}  +  \rA^{(1)}_{1}\rB^{(1)}_{0}  +  \rA^{(1)}_{0}\rB^{(1)}_{1}  -  \rA^{(1)}_{1}\rB^{(1)}_{1}     \\ \label{Mop2}
\mathcal{B}_{2}^{(1)} &=  \rA^{(1)}_{2}\rB^{(1)}_{0}  +  \rA^{(1)}_{2}\rB^{(1)}_{2}  +  \rA^{(1)}_{3}\rB^{(1)}_{0}  -  \rA^{(1)}_{3}\rB^{(1)}_{2}     \\ \label{Mop3}
\mathcal{B}_{3}^{(1)} &=  \rA^{(1)}_{4}\rB^{(1)}_{1} +  \rA^{(1)}_{4}\rB^{(1)}_{2}  +  \rA^{(1)}_{5}\rB^{(1)}_{1}  -  \rA^{(1)}_{5}\rB^{(1)}_{2}.
\end{align}
The maximal quantum violation of all three inequalities corresponding to three CHSH operators is given by the Tsirelson bound $2\sqrt{2}$. Let us show it for the first CHSH operator by considering the sum-of-squares (SOS) decomposition of the so called shifted CHSH operator $2\sqrt{2}\idd - \mathcal{B}_{1}^{(1)}$:
\begin{equation}\label{Mchshsos}
    \sqrt{2}(2\sqrt{2}\idd - \mathcal{B}_{1}^{(1)}) = \left[ \rB^{(1)}_{0} - \frac{\rA^{(1)}_{0} + \rA^{(1)}_{1}}{\sqrt{2}}  \right]^2 + \left[ \rB^{(1)}_{1} - \frac{\rA^{(1)}_{0} - \rA^{(1)}_{1}}{\sqrt{2}}  \right]^2 
\end{equation} 
Hence, the state $\ket{\Psi}$ satisfying
$\bra{\Psi}\mathcal{B}_{1,\yb,\yb'}^{(1)} \ket{\Psi} = 2\sqrt{2}$ must also satisfy
\begin{align}\label{Mpetaa}
    \rB^{(1)}_{0}\ket{\Psi} &= \frac{\rA^{(1)}_{0} + \rA^{(1)}_{1}}{\sqrt{2}}\ket{\Psi} \\
    \rB^{(1)}_{1}\ket{\Psi} &= \frac{\rA^{(1)}_{0} - \rA^{(1)}_{1}}{\sqrt{2}}\ket{\Psi}
\end{align}
Since operators $\frac{\rA^{(1)}_{0} + \rA^{(1)}_{1}}{\sqrt{2}}$ and $\frac{\rA^{(1)}_{0} - \rA^{(1)}_{1}}{\sqrt{2}}$ anticommute by construction the same holds for $\rB^{(1)}_{0}$ and $\rB^{(1)}_{1}$ on the support of $\rho^{(1)} = \textrm{tr}_{\rA_2\rB_2\rA_3\rB_3\cdots \rA_{N-1}\rB_{N-1}\rA_{N}\rB_N}\ketbra{\Psi}{\Psi}$:
\begin{equation}
    \{\rB^{(1)}_{0},\rB^{(1)}_{1}\}\rho^{(1)} = 0
\end{equation}
By repeating the procedure with the two other CHSH inequalities when the maximal violation is obtained $\bra{\Psi}\mathcal{B}_{2}^{(1)} \ket{\Psi} = 2\sqrt{2}$ and $\bra{\Psi}\mathcal{B}_{3}^{(1)} \ket{\Psi} = 2\sqrt{2}$, the following relations can be obtained
\begin{align}\label{Mprvaa}
    \rB^{(1)}_{0}\ket{\Psi} &= \frac{\rA^{(1)}_2 + \rA^{(1)}_3}{\sqrt{2}}\ket{\Psi} \\ \label{Mdrugaa}
    \rB^{(1)}_{2}\ket{\Psi} &= \frac{\rA^{(1)}_2 - \rA^{(1)}_3}{\sqrt{2}}\ket{\Psi} \\ \label{Mtrecaa}
    \rB^{(1)}_{1}\ket{\Psi} &= \frac{\rA^{(1)}_4 + \rA^{(1)}_5}{\sqrt{2}}\ket{\Psi} \\ \label{Mcetvrtaa}
    \rB^{(1)}_{2}\ket{\Psi} &= \frac{\rA^{(1)}_4 - \rA^{(1)}_5}{\sqrt{2}}\ket{\Psi}
\end{align}
which imply the following anticommuting relations:
\begin{align}\label{Mmeasures}
    \{\rB^{(1)}_{0},\rB^{(1)}_{2}\}\rho^{(1)} &= 0\\ \label{Mzadnja}
    \{\rB^{(1)}_{1},\rB^{(1)}_{2}\}\rho^{(1)} &= 0
\end{align}
Let us now introduce some notation for the $1$-st auxiliary party: $\Z_\rB^{(1)} = \rB_{0}^{(1)}$, $\X_\rB^{(1)} = \rB_{1}^{(1)}$ and $\Y_\rB^{(1)} = \rB_{2}^{(1)}$. For $1$-st main party  let us introduce ${\Z'_\rA}^{(1)} = \frac{\rA^{(1)}_0+\rA^{(1)}_1}{\sqrt{2}}$, ${\X'_\rA}^{(1)} = \frac{\rA^{(1)}_0-\rA^{(1)}_1}{\sqrt{2}}$ and ${\Y'}_\rA^{(1)} = \frac{\rA^{(1)}_2-\rA^{(1)}_3}{\sqrt{2}}$. From equations \eqref{Mprvaa}-\eqref{Mcetvrtaa} we can conclude also that:
\begin{align}
    {\Z'_\rA}^{(1)}\ket{\Psi} &=  \Z_\rB^{(1)}\ket{\Psi} \\ \label{dodatna}
    {\X'_\rA}^{(1)}\ket{\Psi} &=  \X_\rB^{(1)}\ket{\Psi} \\
    {\Y'_\rA}^{(1)}\ket{\Psi} &=  \Y_\rB^{(1)}\ket{\Psi}
\end{align}
The operators ${\Z'_\rA}^{(1)}$, ${\X'_\rA}^{(1)}$ and ${\Y'_\rA}^{(1)}$ are not necessarily unitary, but one can define their regularized versions $\Z_\rA^{(1)}$, $\X_\rA^{(1)}$ and $\Y_\rA^{(1)}$ which are unitary by construction. The regularized version of the operator is obtain by rescaling all the eigenvalues to $\pm1$ and changing $0$ eigenvalues to $1$ without changing eigenvectors. This can be described by the following expression
\begin{equation}
\Z_\rA^{(1)} = {\Z'}_\rA^{(1)}/|{\Z'}_\rA^{(1)}|,    
\end{equation}
and analogously for $\X_\rA^{(1)}$ and $\Y_\rA^{(1)}$.

Through the following chain of relations we show that  ${\Z'
_\rA}^{(1)}\ket{\Psi} = \Z_\rA^{(1)}\ket{\Psi}$:

\begin{align*}
|| (\Z_\rA^{(1)}-{\Z'_\rA}^{(1)})\ket{\Psi}|| &= || (\idd-({\Z_\rA^{(1)}}^\dagger{\Z'_\rA}^{(1)}))\ket{\Psi}|| = || (\idd-|{\Z'_\rA}^{(1)})|)\ket{\Psi}||\\
&= || (\idd-|\Z_{\rB}^{(1)}{\Z'_\rA}^{(1)})|)\ket{\Psi}|| \leq || (\idd-\Z_{\rB}^{(1)}{\Z'_\rA}^{(1)})\ket{\Psi}|| = 0,
\end{align*}
where the first equality comes from the fact that ${\Z_\rA^{(1)}}^{\dagger}$ is unitary, the second equality just uses the definition of ${\Z}^{(1)}_{\rA}$. The third equality is equivalent to $|\Z_{\rB}^{(1)}\Z_{\rA}^{(1)}| = |\Z_{\rA}^{(1)}|$, which is correct because $\Z_{\rB}^{(1)}$ is unitary. The inequality is a consequence of $A \leq |A|$, and finally the last equality is the consequence of \eqref{dodatna}.
An equivalent analysis can be done for correlations between all the other pairs of main-auxiliary party. The conditions for self-testing is the observation of the maximal violations
\begin{align}\label{Mbij-1}
    \bra{\Psi}\mathcal{B}_{1}^{(j)}\ket{\Psi} = 2\sqrt{2},\qquad
    \bra{\Psi}\mathcal{B}_{2}^{(j)}\ket{\Psi} = 2\sqrt{2},\qquad
    \bra{\Psi}\mathcal{B}_{3}^{(j)}\ket{\Psi} = 2\sqrt{2},
\end{align}
%
for all  $j = 2, \cdots, N$. In the same way as for $j=1$ we can construct unitary operators $\X_{\rA}^{(j)}$, $\Z_{\rA}^{(j)}$, $\Y_{\rA}^{(j)}$, $\X_{\rB}^{(j)}$, $\Z_{\rB}^{(j)}$, $\Y_{\rB}^{(j)}$ such that:
\begin{align}\label{Mwuhan}
    \X_{\rA}^{(j)}\ket{\Psi} = \X_{\rB}^{(j)}\ket{\Psi},\qquad 
    &\Z_{\rA}^{(j)}\ket{\Psi} = \Z_{\rB}^{(j)}\ket{\Psi},\qquad 
    \Y_{\rA}^{(j)}\ket{\Psi} = \Y_{\rB}^{(j)}\ket{\Psi},\\
    \{\X_{\rA}^{(j)},\Z_{\rA}^{(j)}\}\rho^{(j)} =0,\qquad
    &\{\X_{\rA}^{(j)},\Y_{\rA}^{(j)}\}\rho^{(j)} =0,\qquad
    \{\Y_{\rA}^{(j)},\Z_{\rA}^{(j)}\}\rho^{(j)} =0,\\
    \{\X_{\rB}^{(j)},\Z_{\rB}^{(j)}\}\rho^{(j)} =0,\qquad
    &\{\X_{\rB}^{(j)},\Y_{\rB}^{(j)}\}\rho^{(j)} =0,\qquad
    \{\Y_{\rB}^{(j)},\Z_{\rB}^{(j)}\}\rho^{(j)} =0.
\end{align}
Furthermore, since operators for different parties act non-trivially on different Hilbert spaces, we have
\begin{align}\label{McommA}
[P_{\rA}^{(j)},Q_{\rA}^{(k)}]\ket{\Psi} &= 0\\  \label{McommB} [P_{\rB}^{(j)},Q_{\rB}^{(k)}]\ket{\Psi} &= 0
\end{align}
for every $j\neq k$ and $P,Q \in \{\X,\Z,\Y\}$.  Equipped will all these equations we can directly use the extended SWAP isometry $V_{\rA}\otimes V_{\rB} = \bigotimes_{j=1}^N\left(V^{(j)}_{\rA}\otimes V^{(j)}_{\rB}\right) $   (see Fig. \ref{fig:Mpkb} for $V^{(j)}_{\rA}\otimes V^{(j)}_{\rB}$). A single isometry $V^{(j)} = V^{(j)}_{\rA}\otimes V^{(j)}_{\rB}$ acts in the following way \cite{McKague2012},\cite{sixstate}:
\begin{align}\label{Mvj}
    V^{(j)}\left(\ket{\Psi}_{\mathbb{A}\mathbb{B}}\otimes\ket{0000}_{A_j'A_j''{B}_{j}'{B}_{j}''}\right) &= \ket{\xi}_{\mathbb{A}\mathbb{B}A_j''{B}_{j}''}\otimes\ket{\phi^+}_{A_j'B_{j}'},\\ \label{Mvjz}
    V^{(j)}\left(\Z_{\rC}^{(j)}\ket{\Psi}_{\mathbb{A}\mathbb{B}}\otimes\ket{0000}_{A_j'A_j''B_{j}'B_{j}''}\right) &= \ket{\xi}_{\mathbb{A}\mathbb{B}A_j''B_{j}''}\otimes\sigma_z^{C'_j}\ket{\phi^+}_{A_j'B_{j}'},\\ \label{Mvjx}
        V^{(j)}\left(\X_{\rC}^{(j)}\ket{\Psi}_{\mathbb{A}\mathbb{B}}\otimes\ket{0000}_{A_j'A_j''B_{j}'\rB_{j}''}\right) &= \ket{\xi}_{\mathbb{A}\mathbb{B}A_j''B_{j}''}\otimes\sigma_x^{C'_j}\ket{\phi^+}_{A_j'B_{j}'},\\ \label{Mvjy}
          V^{(j)}\left(\Y_{\rC}^{(j)}\ket{\Psi}_{\mathbb{A}\mathbb{B}}\otimes\ket{0000}_{A_j'A_j''B_{j}'B_{j}''}\right) &= \sigma_z^{C_j''}\ket{\xi}_{\mathbb{A}\mathbb{B}A_j''B_{j}''}\otimes\sigma_y^{C_j'}\ket{\phi^+}_{A_j'B_{j}'},
\end{align}
where $\rC \in \{\rA,\rB\}$, we introduced the notation $\mathbb{A} = A_1\cdots A_N$ and  $\mathbb{B} = B_1\cdots B_N$, and $\ket{\xi}$ takes the form 
\begin{align}\label{Mjunk}
\ket{\xi} =\ket{\xi_0}_{\mathbb{A}\mathbb{B}}\otimes\ket{00}_{A_j''B_{j}''} + \ket{\xi_1}_{\mathbb{A}\mathbb{B}}\otimes\ket{11}_{A_j''B_{j}''}. 
\end{align}
and
\begin{align}
    \ket{\xi_0} &= \frac{1}{2\sqrt{2}}(\idd + i\Y_\rA^{(j)}\X_{\rA}^{(j)})(\idd + \Z_\rA^{(j)})\ket{\Psi},\\
    \ket{\xi_1} &= \frac{1}{2\sqrt{2}}(\idd - i\Y_\rA^{(j)}\X_{\rA}^{(j)})(\idd + \Z_\rA^{(j)})\ket{\Psi}.
\end{align}
eqs. \eqref{Mvj}-\eqref{Mvjy} can be used to obtain:
\begin{align} \label{MvjzMeas}
    V^{(j)}\left(\frac{\idd \pm \Z_{\rC}^{(j)}}{2}\ket{\Psi}_{\mathbb{A}\mathbb{B}}\otimes\ket{0000}_{\rA_j'\rA_j''{\rB}_{j}'{B}_{j}''}\right) &= \ket{\xi}_{\mathbb{A}\mathbb{B}A_j''{B}_{j}''}\otimes\frac{\idd\pm\sigma_z^{C'_j}}{2}\ket{\phi^+}_{A_j'{B}_{j}'},\\ \label{MvjxMeas}
        V^{(j)}\left(\frac{\idd\pm\X_{\rC}^{(j)}}{2}\ket{\Psi}_{\mathbb{A}\mathbb{B}}\otimes\ket{0000}_{A_j'A_j''{B}_{j}'{B}_{j}''}\right) &= \ket{\xi}_{\mathbb{A}\mathbb{B}A_j''{B}_{j}''}\otimes\frac{\idd \pm \sigma_x^{C'_j}}{2}\ket{\phi^+}_{A_j'{B}_{1}'},\\  \label{MvjyMeas-1}
          V^{(j)}\left(\frac{\idd+\Y_{\rC}^{(j)}}{2}\ket{\Psi}_{\mathbb{A}\mathbb{B}}\otimes\ket{0000}_{A_j'A_j''{B}_{j}'{B}_{j}''}\right) &= \frac{\idd + \sigma_z^{C_j''}}{2}\ket{\xi}_{\mathbb{A}A_j''{B}_{j}''}\otimes\frac{\idd + \sigma_y^{C_j'}}{2}\ket{\phi^+}_{A_j'{B}_{j}'} + \frac{\idd - \sigma_z^{C_j''}}{2}\ket{\xi}_{\mathbb{A}A_j''{B}_{j}''}\otimes\frac{\idd - \sigma_y^{C_j'}}{2}\ket{\phi^+}_{A_j'{B}_{j}'}\\ \label{MvjyMeas-2}
          V^{(j)}\left(\frac{\idd-\Y_{\rC}^{(j)}}{2}\ket{\Psi}_{\mathbb{A}\mathbb{B}}\otimes\ket{0000}_{A_j'A_j''{B}_{j}'{B}_{j}''}\right) &= \frac{\idd + \sigma_z^{C_j''}}{2}\ket{\xi}_{\mathbb{A}A_j''{B}_{j}''}\otimes\frac{\idd - \sigma_y^{C_j'}}{2}\ket{\phi^+}_{A_j'{B}_{j}'} + \frac{\idd - \sigma_z^{C_j''}}{2}\ket{\xi}_{\mathbb{A}A_j''{B}_{j}''}\otimes\frac{\idd + \sigma_y^{C_j'}}{2}\ket{\phi^+}_{A_j'{B}_{j}'}
\end{align}
where again $\rC \in \{\rA,\rB\}$. We see that the action of observables $\Z_\rC$, $\X_{\rC}$ and $\Y_{\rC}$ on the physical state is mapped to the action of the Pauli observables $\sigma_z$, $\sigma_x$ and $\sigma_y$, respectively, on the maximally entangled pair of qubits, up to a possible complex conjugation. The isometry $V^{(j)}$ gives a nice interpretation of this complex conjugation freedom. Given the form of $\ket{\xi}$, see eq.~\eqref{Mjunk}, eq.~\eqref{MvjyMeas-1} can be read in the following way: before applying measurement $\Y_\rC$, a party measures ancilla $C_j''$ in the computational basis, and if the result is $0$ it measures $\sigma_y$, while if the result is $-1$ it measures $\sigma_y^*$. Thus, the measurement $\Y_\rC$ is in a way coherently controlled by $\sigma_z$.

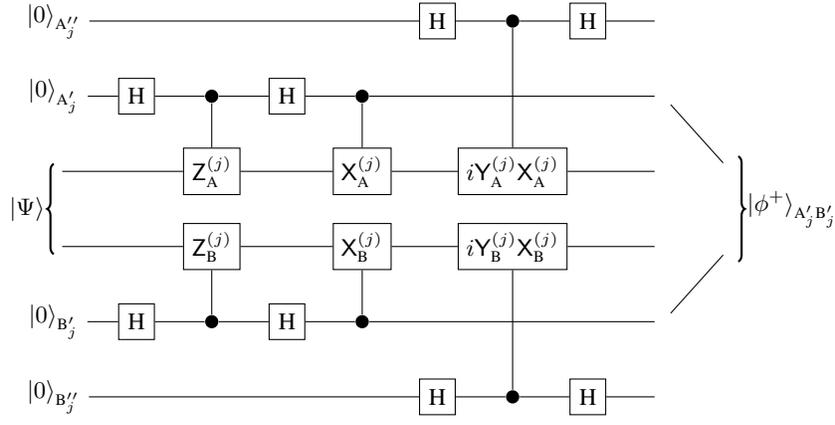
\begin{figure*}
  \centerline{
    \begin{tikzpicture}
    %
    \tikzstyle{operator} = [draw,fill=white,minimum size=1.5em] 
    \tikzstyle{phase} = [fill,shape=circle,minimum size=5pt,inner sep=0pt]
    %
    \node at (-0.1,-1) (q1) {$\ket{0}_{\rA_j'}$ };
    \node at (-0.1,0) (q11) {$\ket{0}_{\rA_j''}$ };
    \node at (-0.1,-2) (q2) { };
    \node at (-0.1,-3) (q3) { };
     \node at (-0.1,-5) (q44) {$\ket{0}_{\rB_j''}$ };
    \node at (-0.1,-4) (q4) {$\ket{0}_{\rB_j'}$ };
    \draw[decorate,decoration={brace, mirror},thick] (-0.1,-1.9) to
	node[midway,left] (bracket) {$\ket{\Psi}$}
	(-0.1,-3.1);
    %
    \node[operator] (opH11) at (1,-1) {H} edge [-] (q1);
    \node[operator] (opH21) at (1,-4) {H} edge [-] (q4);
    %
    \node[phase] (phase11) at (2,-1) {} edge [-] (opH11);
    \node[operator] (phase12) at (2,-2) {$\Z_{\rA}^{(j)}$} edge [-] (q2);
    \draw[-] (phase11) -- (phase12);
    \node[phase] (phase13) at (2,-4) {} edge [-] (opH21);
    \node[operator] (phase14) at (2,-3) {${\Z}_{\rB}^{(j)}$} edge [-] (q3);
    \draw[-] (phase13) -- (phase14);
    \node[operator] (op12) at (3,-1) {H} edge [-] (phase11);
    \node[operator] (op22) at (3,-4) {H} edge [-] (phase13);
    %
    \node[phase] (phase21) at (4,-1) {} edge [-] (op12);
    \node[operator] (phase22) at (4,-2) {$\X_{\rA}^{(j)}$} edge [-] (phase12);
    \draw[-] (phase21) -- (phase22);
    \node[phase] (phase23) at (4,-4) {} edge [-] (op22);
    \node[operator] (phase24) at (4,-3) {${\X}_{\rB}^{(j)}$} edge [-] (phase14);
    \draw[-] (phase23) -- (phase24);
 \node[operator] (opH31) at (5,0) {H} edge [-] (q11);
    \node[operator] (opH41) at (5,-5) {H} edge [-] (q44);
    \node[phase] (phase31) at (6,0) {} edge [-] (opH31);
    \node[operator] (phase32) at (6,-2) {$i\Y_{\rA}^{(j)}\X_{\rA}^{(j)}$} edge [-] (phase22);
    \draw[-] (phase31) -- (phase32);
    \node[operator] (phase33) at (6,-3) {$i{\Y}_{\rB}^{(j)}{\X}_{\rB}^{(j)}$} edge [-] (phase24);
    \node[phase] (phase34) at (6,-5) {} edge [-] (opH41);
     \draw[-] (phase33) -- (phase34);
    \node[operator] (op81) at (7,0) {H} edge [-] (phase31);
    \node[operator] (op82) at (7,-5) {H} edge [-] (phase34);
    \node (end1) at (8,-1) {} edge [-] (phase21);
    \node (end5) at (8,0) {} edge [-] (op81); 
    \node (end2) at (8,-2) {} edge [-] (phase32);
    \node (end3) at (8,-3) {} edge [-] (phase33);
    \node (end6) at (8,-5) {} edge [-] (op82);
    \node (end4) at (8,-4) {} edge [-] (phase23);
    %
 \node at (8,-1) (eend1) {};
    \node at (8.9,-2) (eend11) {};
    \draw[-] (eend1) -- (eend11);
    \node at (8,-4) (eend4) {};
    \node at (8.9,-3) (eend44) {};
    \draw[-] (eend4) -- (eend44);
    \draw[decorate,decoration={brace},thick] (9,-1.8) to 	node[right] (bracket) {$\ket{\phi^+}_{\rA_j'\rB_j'}$} 	(9,-3.2);
    %
    %
    %
    \end{tikzpicture}
  }
  \caption{The isometry for self-testing. The standard SWAP gates extracts the maximally entangled pair of qubits to the primed ancillas.}. \label{fig:Mpkb}
\end{figure*}
%



Let us now consider another isometry $V^{(k)} = V_A^{(k)}\tp V_B^{(k)}$ with $k \neq j$ and check how it acts together with $V_A^{(j)}\tp V_B^{(j)}$. Thanks to eqs. \eqref{McommA} and \eqref{McommB} we obtain:
\begin{align*}
     &V^{(k)}\left(V^{(j)}\left(\ket{\Psi}_{\mathbb{A}\mathbb{B}}\otimes\ket{0000}_{A_j'A_j''{B}_{j}'{B}_{j}''}\right) \otimes\ket{0000}_{A_k'A_k''{B}_{k}'{B}_{k}''}\right) = \\ &= V^{(k)}\left(\ket{\xi}_{\mathbb{A}\mathbb{B}A_j''{B}_{j}''}\otimes\ket{0000}_{A_k'A_k''{B}_{k}'{B}_{k}''}\right) \otimes\ket{\phi^+}_{A_j'{B}_{j}'}\\
     &= \left[V^{(k)}\left(\ket{\xi_0}_{\mathbb{A}\mathbb{B}}\otimes\ket{0000}_{A_k'A_k''{B}_{k}'{B}_{k}''}\right) \otimes\ket{00}_{\rA_j''\rB_{j}''} + V^{(k)}\left(\ket{\xi_1}_{\mathbb{A}\mathbb{B}}\ket{0000}_{A_k'A_k''{B}_{k}'{B}_{k}''}\right)\otimes\ket{11}_{\rA_j''\rB_{j}''}\right]\otimes\ket{\phi^+}_{A_j'{B}_{j}'}\\
     &= \ket{\xi'}_{\mathbb{A}\mathbb{B}A_j''A_k''B_j''B_k''}\otimes \ket{\phi^+}_{A_j'{B}_{j}'} \otimes \ket{\phi^+}_{A_k'{B}_{k}'},
\end{align*}
where
\begin{align}\label{Mjunk2}
\ket{\xi'} =\ket{\xi_{00}}_{\mathbb{A}\mathbb{B}}\otimes\ket{0000}_{A_j''B_{j}''A_k''B_{k}''} +\ket{\xi_{01}}_{\mathbb{A}\mathbb{B}}\otimes\ket{0011}_{A_j''B_{j}''A_k''B_{k}''} +  \ket{\xi_{10}}_{\mathbb{A}\mathbb{B}}\otimes\ket{1100}_{A_j''B_{j}''A_k''B_{k}''} + \ket{\xi_{11}}_{\mathbb{A}\mathbb{B}}\otimes\ket{1111}_{A_j''B_{j}''A_k''B_{k}''}. 
\end{align}
and
\begin{align}
    \ket{\xi_{00}} &= \frac{1}{2^3}(\idd + i\Y_\rA^{(j)}\X_{\rA}^{(j)})(\idd + \Z_\rA^{(j)})(\idd + i\Y_\rA^{(k)}\X_{\rA}^{(k)})(\idd + \Z_\rA^{(k)})\ket{\Psi},\\
    \ket{\xi_{01}} &= \frac{1}{2^3}(\idd + i\Y_\rA^{(j)}\X_{\rA}^{(j)})(\idd + \Z_\rA^{(j)})(\idd - i\Y_\rA^{(j)}\X_{\rA}^{(j)})(\idd + \Z_\rA^{(j)})\ket{\Psi},\\
    \ket{\xi_{10}} &= \frac{1}{2^3}(\idd - i\Y_\rA^{(j)}\X_{\rA}^{(j)})(\idd + \Z_\rA^{(j)})(\idd + i\Y_\rA^{(k)}\X_{\rA}^{(k)})(\idd + \Z_\rA^{(k)})\ket{\Psi},\\
    \ket{\xi_{11}} &= \frac{1}{2^3}(\idd - i\Y_\rA^{(j)}\X_{\rA}^{(j)})(\idd + \Z_\rA^{(j)})(\idd - i\Y_\rA^{(j)}\X_{\rA}^{(j)})(\idd + \Z_\rA^{(j)})\ket{\Psi}.
\end{align}
If we repeat the procedure for all pairs of parties, we obtain the final local isometry $V = \bigotimes_{j=0}^{N-1} V_\rA^{(j)}\otimes V_\rB^{(j)}$  such that
\begin{align}\label{josjednaa}
    V\left(\ket{\Psi}\otimes\ket{0...0}_{\mathbb{A}'\mathbb{A}''{\mathbb{B}}'{\mathbb{B}}''}\right) &= \ket{\mathrm{aux}}_{\mathbb{A}B\mathbb{A}''{\mathbb{B}}''}\bigotimes_{j=0}^{N-1}\ket{\phi^+}_{A_j'B_{j}'}\\ \label{MSTNpairss}
    V\left(\bigotimes_{j=1}^N\W_\rC^{(j)}\ket{\Psi}\otimes\ket{0...0}_{\mathbb{A}'\mathbb{A}''{\mathbb{B}}'{\mathbb{B}}''}\right) &= {\sigma_z^{q(w)}}^{C_j''}\ket{\mathrm{aux}}_{\mathbb{A}B\mathbb{A}''{\mathbb{B}}''}\bigotimes_{j=1}^{N}{\sigma_w}^{C_j'}\ket{\phi^+}_{A_j'B_{j}'},
\end{align}
where $\W \in \{\X,\Y,\Z\}$, $\rC \in \{\rA,\rB\}$ and $q(x) = q(z) = 0$ while $q(y) = 1$.
The junk state $\ket{\mathrm{aux}}$ has the following form:
\begin{equation}\label{eqAuxA11}
    \ket{\mathrm{aux}} = \sum_\iota\ket{\xi_\iota}_{\mathbb{A}B}\ket{\iota}_{\mathbb{A}''}\ket{\iota}_{{\mathbb{B}}''},
\end{equation}
where $\iota \in (0,1)^N$. Through the isometry, the measurements are self-tested up to complex conjugation. For each party measurement $\Y_\rA$ is self-tested to be either $\sigma_y$, $-\sigma_y$ or some coherent superposition of the two. 

To conclude, in this section we proved that if the three CHSH inequalities are maximally violated by $N$ pairs of main-auxiliary parties, $N$ maximally entangled pairs can be extracted. Moreover, the maximal violations allow us to self-test Pauli measurements $\sigma_x$ and $\sigma_z$, while $\sigma_y$ is self-tested up to complex conjugation, as shown in \eqref{josjedna}, and where junk state, given explicitly in \eqref{eqAuxA1} has $N$ flags, which determine whether the auxiliary party applies $\sigma_y$ or $\sigma_y^*$. In the model given here, the auxiliary party measures the flag qubits in the computational basis (hence $\sigma_z$ acting on $\ket{\mathrm{aux}}$ in \eqref{josjedna}), and if the measurement result is $0$ it applies $\sigma_y$, while if the result is $1$ it applies $\sigma_y^*$.


\subsection{Fully network-assisted scenario} \label{app:Step1FullyNetworkAssisted}



In this section, we exploit the fact that the sources present in the physical experiment are mutually independent. This allows us to join the $N$ auxiliary parties into one. We prove the following lemma, which is a refinement of the previous one with the important difference that the physical measurements are now shown to be equivalent either to the reference measurements or to their complex conjugate.

\begin{tcolorbox}
\setcounter{lemma}{0}
\begin{lemma}[Formal version, fully-network-assisted case]\label{lemmachshFNA}
Let the state $\ket{\Psi}_{\mathbb{A}\mathbb{B}} \in \mathcal{H}^{\mathbb{A}}\otimes\mathcal{H}^{\mathbb{B}}$, dichotomic observables $\{\rA_j\}_{j=0}^5$ and $\{\rB_{j,\yb}\}_{j=0}^2$ and measurements $\{\N_{b_j|\lozenge}^{2j-1,2j}, \N_{b_j|\blacklozenge}^{2j,2j+1}\}_{j=1}^{\lfloor N/2\rfloor}$ reproduce the reference correlations given in eqs. \eqref{eq:3CHSHViolationFullyNetworkAssisted} and \eqref{eq:AlignFullyNetworkAssisted}. 
Then there exists a local isometry $V=\bigotimes_{j=0}^{N-1} V_\rA^{(j)}\otimes V_\rB^{(j)}$ such that
\begin{align}\label{josjednab}
    V\left(\ket{\Psi}_{\mathbb{A}\mathbb{B}}\otimes\ket{0...0}_{\mathbb{A}'\mathbb{A}''{\mathbb{B}}'{\mathbb{B}}''}\right) &= \ket{\mathrm{aux}}_{\mathbb{A}\mathbb{B}\mathbb{A}''{\mathbb{B}}''}\bigotimes_{j=0}^{N-1}\ket{\phi^+}_{A_j'B_{j}'}\\ \label{MSTNpairsb}
   V\left(\bigotimes_{j=1}^N\W_\rB^{(j)}\ket{\Psi}_{\mathbb{A}\mathbb{B}}\otimes\ket{0\cdots 0}_{\mathbb{A}'\mathbb{A}''{\mathbb{B}}'{\mathbb{B}}''}\right) &= \begin{cases} \ket{\xi_0}_{\mathbb{A}\mathbb{B}}\ket{0\cdots0}_{\mathbb{A}''{\mathbb{B}}''}\bigotimes_{j=1}^{N}\sigma_w^{B_j'}\ket{\phi^+}_{A_j'B_{j}'},\\
    \ket{\xi_1}_{\mathbb{A}\mathbb{B}}\ket{1\cdots1}_{\mathbb{A}''{\mathbb{B}}''}\bigotimes_{j=1}^{N}{\sigma^*_w}^{B_j'}\ket{\phi^+}_{A_j'B_{j}'},
    \end{cases}
\end{align}
where $\W \in \{\X,\Y,\Z\}$, $\rC \in \{\rA,\rB\}$.
The junk state $\ket{\mathrm{aux}}$ has the following form:
\begin{equation}\label{eqAuxA12}
\ket{\mathrm{aux}} = \ket{\xi_0}_{\mathbb{A}\mathbb{B}}\ket{0\cdots0}_{\mathbb{A}''{\mathbb{B}}''} +  \ket{\xi_1}_{\mathbb{A}\mathbb{B}}\ket{1\cdots1}_{\mathbb{A}''{\mathbb{B}}''}.
\end{equation}
Through the isometry, the measurements are self-tested up to complex conjugation. For each party, measurement $\Y_\rA$ is self-tested to be either $\sigma_y$ or  $-\sigma_y$. 
\end{lemma}
\end{tcolorbox}

Because of the quantum network geometry, we made the assumption that the physical state has the form
\begin{equation}\label{networkstate}
    \ket{\Psi}_{\mathbb{A}\mathbb{B}} = \left(\bigotimes_{j=1}^N\ket{\Psi_j}_{A_jB_j}\right)\otimes\ket{\tilde{\Psi}}_{\tilde{A}_1\cdots\tilde{A}_N} \equiv \left(\bigotimes_{j=1}^N\ket{\Psi_j}_{A_jB_j}\right)\otimes\ket{\tilde{\Psi}}_{\tilde{\mathbb{A}}} 
\end{equation}
In the following, we omit the label ${\mathbb{A}\mathbb{B}}$ from the physical state. As described in the main text, another difference in the network architecture is that $N$ auxiliary parties are merged into one. In what follows, we present the proof that $N$ maximal violations of the extended CHSH inequality each violated by one main and the auxiliary party self-test the tensor product of $N$ maximally entangled pairs of qubits. The proof is similar to the one presented in section~\ref{app:Step1NetworkAssisted}, adapted to the fact that now there are no $N$ auxiliary parties, but only one.

In the first part of the proof we focus on the self-testing of $N$ maximally entangled qubit pairs and for the moment disregard the inputs $\diamond$, $\lozenge$ and $\blacklozenge$.
Let us now define the operator $\rA^{(j)}_{x_j} = \sum_{a_j}(-1)^{a_j}\M^{(j)}_{a_j|x_j}$. For the auxiliary party we define the following operators $\rB^{(j)}_{y_j,\mathbf{y}} = \sum_{\mathbf{b}}(-1)^{b_j}\N_{\mathbf{b}|\mathbf{y}}$.
All the operators we defined are valid measurement observables.
Let us consider a  pair of parties, $1$-st  main and the auxiliary. We check the correlations between the measurement outcomes $a_1$  and $b_{1}$, i.e. the correlations between operators $\rA_{x_1}^{(1)}$ and $\rB_{y_1,\mathbf{y}}^{(1)}$. As operator $\rB_{y_1,\mathbf{y}}^{(1)}$ acts on the whole Hilbert space of the auxiliary party, we have to consider the whole state shared among the $N+1$ parties $\ket{\Psi} = \left(\bigotimes_{j=1}^N\ket{\Psi_j}_{A_jB_j}\right)\otimes\ket{\tilde{\Psi}}_{\tilde{\mathbb{A}}} $. As in the previous proof, the self-test is based on the extended CHSH inequality. For $i\geq 2$, let $y_i \in \{0,1,2\}$ and consider $\yb=(0,y_2, ..., y_N)$, $\yb'=(1,y'_2, ..., y'_N)$ and $\yb''=(2,y''_2, ..., y''_N)$.
From the observables available to the parties one can define the following three CHSH operators
\begin{align}\label{op1}
\mathcal{B}_{1,\yb,\yb'}^{(1)} &=  \rA^{(1)}_{0}\rB^{(1)}_{0,\mathbf{y}}  +  \rA^{(1)}_{1}\rB^{(1)}_{0,\mathbf{y}}  +  \rA^{(1)}_{0}\rB^{(1)}_{1,\mathbf{y}'}  -  \rA^{(1)}_{1}\rB^{(1)}_{1,\mathbf{y}'}     \\ \label{op2}
\mathcal{B}_{2,\yb,\yb''}^{(1)} &=  \rA^{(1)}_{2}\rB^{(1)}_{0,\mathbf{y}}  +  \rA^{(1)}_{2}\rB^{(1)}_{2,\mathbf{y}''}  +  \rA^{(1)}_{3}\rB^{(1)}_{0,\mathbf{y}}  -  \rA^{(1)}_{3}\rB^{(1)}_{2,\mathbf{y}''}     \\ \label{op3}
\mathcal{B}_{3,\yb',\yb''}^{(1)} &=  \rA^{(1)}_{4}\rB^{(1)}_{1,\mathbf{y}'} +  \rA^{(1)}_{4}\rB^{(1)}_{2,\mathbf{y}''}  +  \rA^{(1)}_{5}\rB^{(1)}_{1,\mathbf{y}'}  -  \rA^{(1)}_{5}\rB^{(1)}_{2,\mathbf{y}''}.
\end{align}
All three CHSH operators attain the Tsirelson bound $2\sqrt{2}$. The SOS decomposition of the shifted CHSH operator $2\sqrt{2}\idd - \mathcal{B}_{1,\yb,\yb'}^{(j)}$ reads:
\begin{equation}\label{chshsos}
    \sqrt{2}(2\sqrt{2}\idd - \mathcal{B}_{1,\yb,\yb'}^{(1)}) = \left[ \rB^{(1)}_{0,\yb} - \frac{\rA^{(1)}_{0} + \rA^{(1)}_{1}}{\sqrt{2}}  \right]^2 + \left[ \rB^{(1)}_{1,\yb'} - \frac{\rA^{(1)}_{0} - \rA^{(1)}_{1}}{\sqrt{2}}  \right]^2 
\end{equation}
Hence, the state $\ket{\Psi}$ satisfying
$\bra{\Psi}\mathcal{B}_{1,\yb,\yb'}^{(1)} \ket{\Psi} = 2\sqrt{2}$ must also satisfy
\begin{align}\label{petaa}
    \rB^{(1)}_{0,\yb}\ket{\Psi} &= \frac{\rA^{(1)}_{0} + \rA^{(1)}_{1}}{\sqrt{2}}\ket{\Psi} \\ \label{sestaa}
    \rB^{(1)}_{1,\yb'}\ket{\Psi} &= \frac{\rA^{(1)}_{0} - \rA^{(1)}_{1}}{\sqrt{2}}\ket{\Psi}
\end{align}
The equations
\begin{equation}
\frac{\rA^{(1)}_{0} \pm \rA^{(1)}_{1}}{\sqrt{2}}\ket{\Psi} = \left(\frac{\rA^{(1)}_{0} \pm \rA^{(1)}_{1}}{\sqrt{2}}\ket{\Psi_1}_{A_1B_1}\right)\bigotimes_{j=2}^N\ket{\Psi_j}_{A_jB_j}\otimes\ket{\tilde{\Psi}}_{\tilde{\mathbb{A}}}
\end{equation}
imply
\begin{align}
    \rB^{(1)}_{0,\yb}\ket{\Psi} &= \left(\rB^{(1)}_{0,\yb}\ket{\Psi_1}_{A_1B_1}\right)\bigotimes_{j=2}^N\ket{\Psi_j}_{A_jB_j}\otimes\ket{\tilde{\Psi}}_{\tilde{\mathbb{A}}}\\
    \rB^{(1)}_{1,\yb'}\ket{\Psi} &= \left(\rB^{(1)}_{1,\yb'}\ket{\Psi_1}_{A_1B_1}\right)\bigotimes_{j=2}^N\ket{\Psi_j}_{A_jB_j}\otimes\ket{\tilde{\Psi}}_{\tilde{\mathbb{A}}}
\end{align}
Operators  $\rB^{(1)}_{0,\yb'}$ and $\rB^{(1)}_{1,\yb'}$ act nontrivially only on $\ket{\Psi_1}_{A_1B_1}$ because they act on the full state $\ket{\Psi}$ in the same way as $(\rA_0^{(1)}\pm \rA_1^{(1)})/\sqrt{2}$, which cannot affect the states $\ket{\Psi_j}$ for $j\neq 1$.

Since the operators $\frac{\rA^{(1)}_{0} + \rA^{(1)}_{1}}{\sqrt{2}}$ and $\frac{\rA^{(1)}_{0} - \rA^{(1)}_{1}}{\sqrt{2}}$ anticommute by construction, the same holds for $\rB^{(1)}_{0,\yb}$ and $\rB^{(1)}_{1,\yb'}$ on the support of $\rho^{(1)}_B = \textrm{tr}_{A_1}\ketbra{\Psi_1}{\Psi_1}$:
\begin{equation}
    \{\rB^{(1)}_{0,\yb},\rB^{(1)}_{1,\yb'}\}\rho_B^{(1)} = 0
\end{equation}
By repeating the procedure with the two other CHSH inequalities that attain the maximal violation $\bra{\Psi}\mathcal{B}_{2,\yb,\yb''}^{(1)} \ket{\Psi} = 2\sqrt{2}$ and $\bra{\Psi}\mathcal{B}_{3,\yb',\yb''}^{(1)} \ket{\Psi} = 2\sqrt{2}$, the following relations can be obtained
\begin{align}\label{prvaa}
    \rB^{(1)}_{0,\yb}\ket{\Psi} &= \frac{\rA^{(1)}_2 + \rA^{(1)}_3}{\sqrt{2}}\ket{\Psi} \\ \label{drugaa}
    \rB^{(1)}_{2,\yb''}\ket{\Psi} &= \frac{\rA^{(1)}_2 - \rA^{(1)}_3}{\sqrt{2}}\ket{\Psi} \\ \label{trecaa}
    \rB^{(1)}_{1,\yb'}\ket{\Psi} &= \frac{\rA^{(1)}_4 + \rA^{(1)}_5}{\sqrt{2}}\ket{\Psi} \\ \label{cetvrtaa}
    \rB^{(1)}_{2,\yb''}\ket{\Psi} &= \frac{\rA^{(1)}_4 - \rA^{(1)}_5}{\sqrt{2}}\ket{\Psi}
\end{align}
which imply the following anticommuting relations:
\begin{align}\label{measures}
    \{\rB^{(1)}_{0,\yb},\rB^{(1)}_{2,\yb''}\}\rho_B^{(1)} &= 0\\ \label{zadnja}
    \{\rB^{(1)}_{1,\yb'},\rB^{(1)}_{2,\yb''}\}\rho_B^{(1)} &= 0,
\end{align}
and also
\begin{align}\label{onemore}
    \rB^{(1)}_{2,\yb''}\ket{\Psi} &= \left(\rB^{(1)}_{2,\yb''}\ket{\Psi_1}_{A_1B_2}\right)\bigotimes_{j=2}^N\ket{\Psi_j}_{A_jB_j}\otimes\ket{\tilde{\Psi}}_{\tilde{\mathbb{A}}} .
\end{align}

If equalities $\bra{\Psi}\mathcal{B}_{1,\yb,\yb'}^{(1)} \ket{\Psi} = 2\sqrt{2}$,  $\bra{\Psi}\mathcal{B}_{2,\yb,\yb''}^{(1)} \ket{\Psi} = 2\sqrt{2}$ and $\bra{\Psi}\mathcal{B}_{3,\yb',\yb''}^{(1)} \ket{\Psi} = 2\sqrt{2}$ are satisfied for all $\yb$, $\yb'$ and $\yb''$,  eqs. \eqref{petaa}-\eqref{onemore} hold for all such $\yb$, $\yb'$ and $\yb''$.  
Let us now introduce some notation for the auxiliary party: $\Z_\rB^{(1)} = \rB_{0,\mathbf{0}}^{(1)}$, $\X_\rB^{(1)} = \rB_{1,\mathbf{1}}^{(1)}$ and $\Y_\rB^{(1)} = \rB_{2,\mathbf{2}}^{(1)}$, where $\mathbf{0} = (0,\cdots, 0)$, $\mathbf{1} = (1,\cdots, 1)$ and $\mathbf{2} = (2,\cdots, 2)$. For the $1$-st main party  let us introduce ${\Z'}_\rA^{(1)} = \frac{\rA^{(1)}_0+\rA^{(1)}_1}{\sqrt{2}}$, ${\X'}_\rA^{(1)} = \frac{\rA^{(1)}_0-\rA^{(1)}_1}{\sqrt{2}}$ and ${\Y'}_\rA^{(1)} = \frac{\rA^{(1)}_2-\rA^{(1)}_3}{\sqrt{2}}$. From equations \eqref{petaa}, \eqref{sestaa},\eqref{prvaa}- \eqref{cetvrtaa} we can conclude also that:
\begin{align}
    {\Z'}_\rA^{(1)}\ket{\Psi} &=  \frac{\rA^{(1)}_2 + \rA^{(1)}_3}{\sqrt{2}}\ket{\Psi} \\
    {\X'}_\rA^{(1)}\ket{\Psi} &=  \frac{\rA^{(1)}_4 + \rA^{(1)}_5}{\sqrt{2}}\ket{\Psi} \\
    {\Y'}_\rA^{(1)}\ket{\Psi} &=  \frac{\rA^{(1)}_4 - \rA^{(1)}_5}{\sqrt{2}}\ket{\Psi} .
\end{align}
Similarly, as eqs. \eqref{petaa}-\eqref{cetvrtaa} hold for all $\yb_0$, $\yb'$ and $\yb''$, we can conclude:
\begin{align}
    \Z_{\rB}^{(1)}\ket{\Psi} &= \rB_{0,\yb}^{(1)}\ket{\Psi},\qquad \forall \yb|y_1 = 0, \\
    \X_{\rB}^{(1)}\ket{\Psi} &= \rB_{1,\yb'}^{(1)}\ket{\Psi},\qquad \forall \yb'|y'_1 = 1,\\
    \Y_{\rB}^{(1)}\ket{\Psi} &= \rB_{2,\yb''}^{(1)}\ket{\Psi},\qquad \forall \yb''|y''_1 = 2.\\
\end{align}
The operators ${\Z'}_\rA^{(1)}$, ${\X'}_\rA^{(1)}$ and ${\Y'}_\rA^{(1)}$ are not necessarily unitary, but we can repeat the regularization procedure from the previous proof to show that there exist unitary operators $\Z_\rA^{(1)}$, $\X_\rA^{(1)}$ and $\Y_\rA^{(1)}$ such that:
\begin{align*}
\Z_\rA^{(1)}\ket{\Psi} = {\Z'_\rA}^{(1)}\ket{\Psi}, \qquad \X_\rA^{(1)}\ket{\Psi} = {\X'_\rA}^{(1)}\ket{\Psi}, \qquad \Y_\rA^{(1)}\ket{\Psi} = {\Y'_\rA}^{(1)}\ket{\Psi}.
\end{align*}
An equivalent analysis can be done for correlations between all other main parties an the auxiliary party. Observing maximal violations:
\begin{align}\label{bij-1}
    \bra{\Psi}\mathcal{B}_{1,\mathbf{y},\mathbf{y}'}^{(j)}\ket{\Psi} &= 2\sqrt{2},\\
    \bra{\Psi}\mathcal{B}_{2,\mathbf{y},\mathbf{y}''}^{(j)}\ket{\Psi} &= 2\sqrt{2},\\ \label{bij-3}
    \bra{\Psi}\mathcal{B}_{3,\mathbf{y}',\mathbf{y}''}^{(j)}\ket{\Psi} &= 2\sqrt{2},
\end{align}
%
for all $\mathbf{y}$, $\mathbf{y}'$, $\mathbf{y}''$ and $j = 2, \cdots, N$. In the same way as for $j=1$ we can construct unitary operators $\X_{\rA}^{(j)}$, $\Z_{\rA}^{(j)}$, $\Y_{\rA}^{(j)}$, $\X_{\rB}^{(j)}$, $\Z_{\rB}^{(j)}$, $\Y_{\rB}^{(j)}$ such that:
\begin{align}\label{wuhan}
    \X_{\rA}^{(j)}\ket{\Psi} = \X_{\rB}^{(j)}\ket{\Psi},\qquad 
    &\Z_{\rA}^{(j)}\ket{\Psi} = \Z_{\rB}^{(j)}\ket{\Psi},\qquad
    \Y_{\rA}^{(j)}\ket{\Psi} = \Y_{\rB}^{(j)}\ket{\Psi},\\ \label{dodatnaa}
    \{\X_{\rA}^{(j)},\Z_{\rA}^{(j)}\} =0,\qquad
    &\{\X_{\rA}^{(j)},\Y_{\rA}^{(j)}\} =0,\qquad
    \{\Y_{\rA}^{(j)},\Z_{\rA}^{(j)}\} =0,\\
    \{\X_{\rB}^{(j)},\Z_{\rB}^{(j)}\}\rho_B^{(j)} =0,\qquad
    &\{\X_{\rB}^{(j)},\Y_{\rB}^{(j)}\}\rho_B^{(j)} =0,\qquad
    \{\Y_{\rB}^{(j)},\Z_{\rB}^{(j)}\}\rho_B^{(j)} =0,\\ \label{Bcommutes} P_\rB^{(j)}\ket{\Psi} = &\left(P_\rB^{(j)}\ket{\Psi_j}_{A_jB_j}\right)\bigotimes_{k\neq j}\ket{\Psi_k}_{A_kB_k}\otimes\ket{\tilde{\Psi}}_{\tilde{\mathbb{A}}}\qquad P \in \{\Z,\X,\Y\} \end{align}
eqs. \eqref{Bcommutes} imply:
\begin{equation}\label{Bcommutesall}
   [P_{\rB}^{(j)},Q_{\rB}^{(k)}]\ket{\Psi} = 0, 
\end{equation}
for all $j \neq k$ and $P,Q \in \{\X,\Z,\Y\}$. The operators of different main parties commute since they act non-trivially on different Hilbert spaces. Equipped with all these equations we can directly use the extended SWAP isometry $V_{\rA}\otimes V_{\rB} =  \left(\bigotimes_{j=1}^NV^{(j)}_{\rA}\right)\otimes\left(\prod_{j=1}^N V^{(j)}_{\rB}\right) $; see Fig. \ref{fig:Mpkb} for $V^{(j)} = V^{(j)}_\rA\otimes V_\rB^{(j)}$.
Due to \eqref{Bcommutesall} and the fact that different isometries $V^{(j)}$ all mutually commute we can get the decomposition:
\begin{equation}\label{decompose}
    V_{\rA}\otimes V_{\rB}\left(\ket{\Psi} \otimes \bigotimes_{j=1}^N\ket{0000}_{A'_jB'_jA''_jB''_j}\right) = \bigotimes_{j=1}^NV^{(j)}\left(\ket{\Psi_j}_{A_jB_j} \otimes \ket{0000}_{A'_jB'_jA''_jB''_j}\right)
\end{equation}

A single isometry $V^{(j)}$ acts in the following way:
\begin{align}\label{vj}
    V^{(j)}\left(\ket{\Psi_j}_{A_jB_j}\otimes\ket{0000}_{A_j'A_j''{B}_{j}'{B}_{j}''}\right) &= \ket{\xi_j}_{A_jB_jA_j''{B}_{j}''}\otimes\ket{\phi^+}_{A_j'B_{j}'},\\ \label{vjz}
    V^{(j)}\left(\Z_{\rC}^{(j)}\ket{\Psi_j}_{A_jB_j}\otimes\ket{0000}_{A_j'A_j''B_{j}'B_{j}''}\right) &= \ket{\xi_j}_{A_jB_jA_j''B_{j}''}\otimes{\sigma_z}^{C'_j}\ket{\phi^+}_{A_j'B_{j}'},\\ \label{vjx}
        V^{(j)}\left(\X_{\rC}^{(j)}\ket{\Psi}_{A_jB_j}\otimes\ket{0000}_{A_j'A_j''B_{j}'B_{j}''}\right) &= \ket{\xi_j}_{A_jB_jA_j''B_{j}''}\otimes{\sigma_x}^{C'_j}\ket{\phi^+}_{A_j'B_{j}'},\\ \label{vjy}
         V^{(j)}\left(\Y_{\rC}^{(j)}\ket{\Psi_j}_{A_jB_j}\otimes\ket{0000}_{A_j'A_j''B_{j}'B_{j}''}\right) &= {\sigma_z}^{C_j''}\ket{\xi_j}_{{A}_jB_jA_j''B_{j}''}\otimes{\sigma_y}^{C_j'}\ket{\phi^+}_{A_j'B_{j}'},
\end{align}
where $\rC \in \{\rA,\rB\}$ and  $\ket{\xi_j}$ takes the form 
\begin{align}\label{junk}
\ket{\xi_j} =\ket{\xi^j_0}_{{A}_jB_j}\otimes\ket{00}_{A_j''B_{j}''} + \ket{\xi^j_1}_{{A}_jB_j}\otimes\ket{11}_{A_j''B_{j}''}. 
\end{align}
and
\begin{align}
    \ket{\xi^j_0}_{A_jB_j} &= \frac{1}{2\sqrt{2}}(\idd + i\Y_\rA^{(j)}\X_{\rA}^{(j)})(\idd + \Z_\rA^{(j)})\ket{\Psi_j}_{A_jB_j},\\
    \ket{\xi^j_1}_{A_jB_j} &= \frac{1}{2\sqrt{2}}(\idd - i\Y_\rA^{(j)}\X_{\rA}^{(j)})(\idd + \Z_\rA^{(j)})\ket{\Psi_j}_{A_jB_j}.
\end{align}
eqs. \eqref{vj}-\eqref{vjy} can be used to obtain:
\begin{align} \label{vjzMeas}
    V^{(j)}\left(\frac{\idd \pm \Z_{\rC}^{(j)}}{2}\ket{\Psi_j}_{A_jB_j}\otimes\ket{0000}_{A_j'A_j''{B}_{j}'{B}_{j}''}\right) &= \ket{\xi_j}_{{A}_jB_jA_j''{B}_{j}''}\otimes\frac{\idd\pm\sigma_z^{C'_j}}{2}\ket{\phi^+}_{A_j'{B}_{j}'},\\ \label{vjxMeas}
        V^{(j)}\left(\frac{\idd\pm\X_{\rC}^{(j)}}{2}\ket{\Psi_j}_{A_jB_j}\otimes\ket{0000}_{A_j'A_j''{B}_{j}'{B}_{j}''}\right) &= \ket{\xi_j}_{{A}_jB_jA_j''{B}_{j}''}\otimes\frac{\idd \pm \sigma_x^{C'_j}}{2}\ket{\phi^+}_{A_j'{B}_{j}'},\\ \label{vjyMeas-}
          V^{(j)}\left(\frac{\idd\pm\Y_{\rC}^{(j)}}{2}\ket{\Psi_j}_{A_jB_j}\otimes\ket{0000}_{A_j'A_j''{B}_{j}'{B}_{j}''}\right) &= \frac{\idd + {\sigma_z}^{C_j''}}{2}\ket{\xi_j}_{{A}_jB_jA_j''{B}_{j}''}\otimes\frac{\idd \pm {\sigma_y}^{C_j'}}{2}\ket{\phi^+}_{A_j'{B}_{j}'} + \\  &\qquad\qquad +\frac{\idd - {\sigma_z}^{C_j''}}{2}\ket{\xi_j}_{{A}_jB_jA_j''{B}_{j}''}\otimes\frac{\idd \mp {\sigma_y}^{C_j'}}{2}\ket{\phi^+}_{A_j'{B}_{j}'},
\end{align}
where again $\rC \in \{\rA,\rB\}$. These eqs. correspond to the self-testing of measurements. As we see, performing measurements $\Z$, $\X$ and $\Y$ on the physical state is equivalent to measuring $\sigma_z$, $\sigma_x$ and $\sigma_y$ (up to complex conjugation) on the extracted reference state. 
Finally, eqs. \eqref{decompose}-\eqref{vjy} imply:
\begin{align}\label{foound}
    V\left(\ket{\Psi}_{\mathbb{A}\mathbb{B}}\otimes\ket{0...0}_{\mathbb{A}'\mathbb{A}''{\mathbb{B}}'{\mathbb{B}}''}\right) &= \ket{\mathrm{aux}}_{\mathbb{A}\mathbb{B}\mathbb{A}''{\mathbb{B}}''}\bigotimes_{j=0}^{N-1}\ket{\phi^+}_{A_j'B_{j}'}\\ \label{STNpairs}
    V\left(\bigotimes_{j=1}^N\W_\rC^{(j)}\ket{\Psi}_{\mathbb{A}\mathbb{B}}\otimes\ket{0...0}_{\mathbb{A}'\mathbb{A}''{\mathbb{B}}'{\mathbb{B}}''}\right) &= {\sigma_z^{q(w)}}^{C_j''}\ket{\mathrm{aux}}_{\mathbb{A}\mathbb{B}\mathbb{A}''{\mathbb{B}}''}\bigotimes_{j=1}^{N}{\sigma_w}^{C_j'}\ket{\phi^+}_{A_j'B_{j}'},
\end{align}
where $\W \in \{\X,\Y,\Z\}$, $\rC \in \{\rA,\rB\}$ and $q(x) = q(z) = 0$ while $q(y) = 1$, $\mathbb{A}' = A'_1,\cdots, A_N'$,$\mathbb{}'' = A''_1,\cdots, A_N''$, $\mathbb{B}' = B'_1,\cdots, B_N'$,$\mathbb{B}'' = B''_1,\cdots, B_N''$, .
The junk state $\ket{\mathrm{aux}}$ has the following form:
\begin{equation}
\ket{\mathrm{aux}} = \sum_\iota\ket{\xi_\iota}_{\mathbb{A}\mathbb{B}}\ket{\iota}_{\mathbb{A}''}\ket{\iota}_{{\mathbb{B}}''},
\end{equation}
where $\iota \in (0,1)^N$, and
\begin{equation}
    \ket{\xi_\iota}_{\mathbb{A}\mathbb{B}} = \bigotimes_{j = 1}^N\ket{\xi_{\iota(j)}}_{A_jB_j}
\end{equation}
where $\iota(j)$ is the $j$-th element of $\iota$.

Up to now we have not taken into account the correlations of the measurements denoted with $\lozenge$ and $\blacklozenge$. The isometry maps the expression $\otimes_{j=1}^N\Y^{(j)}_{\rA}\ket{\Psi}$ to a coherent superposition of $2^N$ expressions, with different elements of the superposition being labeled with a particular string $\iota$. Whether the extracted observable corresponding to $\Y^{(j)}_{\rA}$ is $\sigma_y$ or $\sigma_y^*=-\sigma_y$  in an element of the superposition depends on the corresponding value $\iota(j)$: if $\iota(j) = 0$ there is no complex conjugation (\emph{i.e.} $\sigma_y$ is extracted), while $\iota(j) = 1$ implies that the complex conjugation is applied (\emph{i.e.} $\sigma_y^*$ is extracted). The complex conjugations of the extracted measurements corresponding to different parties are independent from each other. The correlations of the measurements $\lozenge$ and $\blacklozenge$ will remove this independence, and only two possibilities will remain: either only the elements of the superposition corresponding to extracted $\sigma_y$ will remain, or only those corresponding to extracted $\sigma_y^*$.  For this we reproduce the proof from \cite{Bowles_2018}.
Remember that in the ideal case the measurements denoted with $\lozenge$ and $\blacklozenge$ are the Bell state measurements: when the input is $\lozenge$ the auxiliary party should measure in the Bell basis the following pairs of qubits $(1,2)$, $(3,4)$, $\cdots,(N-1,N)$, while the input $\blacklozenge$ "tells" the auxiliary party to measure in the Bell basis qubit-pairs $(2,3)$, $(4,5)$, $\cdots, (N,1)$. When the auxiliary party receives the outcome corresponding to the projector $\proj{\Phi_{{+}}}$ applied on systems $1$ and $2$, the two corresponding main parties should have the perfect correlation between physical observables corresponding to reference observables $\sigma_{z}$ or $\sigma_{x}$ and perfectly anticorrelated  physical observables corresponding to reference observable $\sigma_{{y}}$. Notice that complex conjugation of $\sigma_{{y}}$ is equivalent to swapping the outcomes since $\sigma_{{y}}^* = -\sigma_{{y}}$. But if the main parties do this to their observables corresponding to $\sigma_{{y}}$ independently from each other, the correlations compatible with Bell state measurement outcomes will be lost. Note that the main parties are never asked to measure exactly Pauli observables but we can still work with correlations of their observables whose linear combinations give exactly Pauli observables. 

For $m=1,\cdots,n$ let us define the following operators
\begin{align}
\Ss_{m,b^*}=\sum_{\bb:b_m=b^*}\N^{(j)}_{\bb|\lozenge} , \quad \T_{m,b^*}=\sum_{\bb:b_m=b^*}\N^{(j)}_{\bb|\blacklozenge},
\end{align}
which are the projectors onto the the subspace corresponding to $b_m=b^*$ for the two measurements.

\def\arraystretch{1.5}
\begin{table}\label{bsmtab1}
\begin{tabular}{ |c | c | c | c | c |} 
\hline
 & $\quad \idd \quad$ &$\Z_{\rA}^{(2m-1)} \Z_{\rA}^{(2m)}$ & $\X^{(2m-1)}_{\rA}\X^{(2m)}_{\rA}$ & $\Y^{(2m-1)}_{\rA}\Y^{(2m)}_{\rA}$ \\ 
\hline
$\Ss_{m,00}$&$\frac{1}{4}$ &$\frac{1}{4}$ & $\frac{1}{4}$ & $-\frac{1}{4}$ \\ 
\hline
$\Ss_{m,01}$& $\frac{1}{4}$ &$\frac{1}{4}$ & $-\frac{1}{4}$ & $\frac{1}{4}$ \\ 
\hline
$\Ss_{m,10}$&$\frac{1}{4}$ & $-\frac{1}{4}$ & $\frac{1}{4}$  & $\frac{1}{4}$ \\ 
\hline
$\Ss_{m,11}$ & $\frac{1}{4}$ &$-\frac{1}{4}$ & $-\frac{1}{4}$  & $-\frac{1}{4}$ \\ 
\hline
\end{tabular}\qquad\qquad
\begin{tabular}{ |c | c | c | c | c |} 
\hline
 & $\quad \idd \quad$ &$\Z^{(2m)}_{\rA} \Z^{(2m+1)}_{\rA}$ & $\X^{(2m)}_{\rA}\X^{(2m+1)}_{\rA}$ & $\Y^{(2m)}_{\rA}\Y^{(2m+1)}_{\rA}$ \\ [0.6ex]
\hline
$\T_{m,00}$&$\frac{1}{4}$ &$\frac{1}{4}$ & $\frac{1}{4}$ & $-\frac{1}{4}$ \\ 
\hline
$\T_{m,01}$& $\frac{1}{4}$ &$\frac{1}{4}$ & $-\frac{1}{4}$ & $\frac{1}{4}$ \\ 
\hline
$\T_{m,10}$&$\frac{1}{4}$ & $-\frac{1}{4}$ & $\frac{1}{4}$  & $\frac{1}{4}$ \\ 
\hline
$\T_{m,11}$ & $\frac{1}{4}$ &$-\frac{1}{4}$ & $-\frac{1}{4}$  & $-\frac{1}{4}$ \\ 
\hline
\end{tabular}
\caption{\label{tablecors} Elements of the table give correlation $\bra{\Psi}C\otimes R\ket{\Psi}$ where $C$ is the operator labelling the column and $R$ the operator labelling the row.}
\end{table}
In what follows we explore the constraints on the physical experiment imposed by the correlations given in Table \ref{tablecors}. The norm of the states $\Ss_{m,b}\ket{\Psi}$ and $\T_{m,b}\ket{\Psi}$ is equal to $\frac{1}{2}$. Hence, we can write
\begin{equation}\label{sim}
\Ss_{m,00}\ket{\Psi} \sim \frac{1}{4}\left(\ket{\Psi} +  \Z_{\rA}^{(2m-1)}\Z_{\rA}^{(2m)}\ket{\Psi} + \X_{\rA}^{(2m-1)}\X_{\rA}^{(2m)}\ket{\Psi} - \Y_{\rA}^{(2m1)}\Y_{\rA}^{(2m)}\ket{\Psi}\right).
\end{equation}
Note that the set of states $\{\ket{\Psi}, \Z_{\rA}^{(2m-1)}\Z_{\rA}^{(2m)}\ket{\Psi}, \X_{\rA}^{(2m-1)}\X_{\rA}^{(2m)}\ket{\Psi},\Y_{\rA}^{(2m-1)}\Y_{\rA}^{(2m)}\ket{\Psi}\}$ is orthonormal and it can form a basis for $\mathcal{H}^{A_{2m-1}}\tp \mathcal{H}^{A_{2m}}\tp\mathcal{H}^{B_{2m-1}}\tp \mathcal{H}^{B_{2m}}$. Moreover $\Ss_{m,0}\ket{\Psi}$  has the same norm as the expression from the right hand side of $\sim$  in eq. (\ref{sim})  which implies
\begin{equation}\label{M0}
\Ss_{m,00}\ket{\Psi} = \frac{1}{4}\left(\ket{\Psi} +  \Z_{\rA}^{(2m-1)}\Z_{\rA}^{(2m)}\ket{\Psi} + \X_{\rA}^{(2m-1)}\X_{\rA}^{(2m)}\ket{\Psi} - \Y_{\rA}^{(2m-1)}\Y_{\rA}^{(2m)}\ket{\Psi}\right).
\end{equation}
In the same way we can obtain the following equations:
\begin{eqnarray}\label{M1}
\Ss_{m,01}\ket{\Psi} &=& \frac{1}{4}\left(\ket{\Psi} +  \Z_{\rA}^{(2m-1)}\Z_{\rA}^{(2m)}\ket{\Psi} - \X_{\rA}^{(2m-1)}\X_{\rA}^{(2m)}\ket{\Psi} + \Y_{\rA}^{(2m-1)}\Y_{\rA}^{(2m)}\ket{\Psi}\right),\\ \label{M2}
\Ss_{m,10}\ket{\Psi} &=& \frac{1}{4}\left(\ket{\Psi} -  \Z_{\rA}^{(2m-1)}\Z_{\rA}^{(2m)}\ket{\Psi} + \X_{\rA}^{(2m-1)}\X_{\rA}^{(2m)}\ket{\Psi} + \Y_{\rA}^{(2m-1)}\Y_{\rA}^{(2m)}\ket{\Psi}\right),\\ \label{M3}
\Ss_{m,11}\ket{\Psi} &=& \frac{1}{4}\left(\ket{\Psi} -  \Z_{\rA}^{(2m-1)}\Z_{\rA}^{(2m)}\ket{\Psi} - \X_{\rA}^{(2m-1)}\X_{\rA}^{(2m)}\ket{\Psi} - \Y_{\rA}^{(2m-1)}\Y_{\rA}^{(2m)}\ket{\Psi}\right),\\ \label{N0}
\T_{m,00}\ket{\Psi} &=& \frac{1}{4}\left(\ket{\Psi} +  \Z_{\rA}^{(2m-1)}\Z_{\rA}^{(2m)}\ket{\Psi} + \X_{\rA}^{(2m-1)}\X_{\rA}^{(2m)}\ket{\Psi} - \Y_{\rA}^{(2m-1)}\Y_{\rA}^{(2m)}\ket{\Psi}\right),\\  \label{N1}
\T_{m,01}\ket{\Psi} &=& \frac{1}{4}\left(\ket{\Psi} +  \Z_{\rA}^{(2m-1)}\Z_{\rA}^{(2m)}\ket{\Psi} - \X_{\rA}^{(2m-1)}\X_{\rA}^{(2m)}\ket{\Psi} + \Y_{\rA}^{(2m-1)}\Y_{\rA}^{(2m)}\ket{\Psi}\right),\\  \label{N2}
\T_{m,10}\ket{\Psi} &=& \frac{1}{4}\left(\ket{\Psi} -  \Z_{\rA}^{(2m-1)}\Z_{\rA}^{(2m)}\ket{\Psi} + \X_{\rA}^{(2m-1)}\X_{\rA}^{(2m)}\ket{\Psi} + \Y_{\rA}^{(2m-1)}\Y_{\rA}^{(2m)}\ket{\Psi}\right),\\  \label{N3}
\T_{m,11}\ket{\Psi} &=& \frac{1}{4}\left(\ket{\Psi} -  \Z_{\rA}^{(2m-1)}\Z_{\rA}^{(2m)}\ket{\Psi} - \X_{\rA}^{(2m-1)}\X_{\rA}^{(2m)}\ket{\Psi} - \Y_{\rA}^{(2m-1)}\Y_{\rA}^{(2m)}\ket{\Psi}\right).
\end{eqnarray}

Equations (\ref{M0}-\ref{M3}) can be used to derive
\begin{subequations}
\begin{eqnarray}\label{Zz}
\Z_{\rA}^{(2m-1)}\Z_{\rA}^{(2m)}\ket{\Psi} &=& \left(\Ss_{m,00} + \Ss_{m,01} - \Ss_{m,10} - \Ss_{m,11}\right)\ket{\Psi}, \\ \label{Xx}
\X_{\rA}^{(2m-1)}\X_{\rA}^{(2m)}\ket{\Psi} &=& \left(\Ss_{m,00} - \Ss_{m,01} + \Ss_{m,10} - \Ss_{m,11}\right)\ket{\Psi}, \\ \label{Yy}
\Y_{\rA}^{(2m-1)}\Y_{\rA}^{(2m)}\ket{\Psi} &=& \left(-\Ss_{m,00} + \Ss_{m,01} + \Ss_{m,10} - \Ss_{m,11}\right)\ket{\Psi}.
\end{eqnarray}
\end{subequations}
The set $\{\Ss_{m,b}\}_{m,b}$ is orthogonal and all its elements commute with all the operators from $\{\Z_{\rA}^{(j)},\X_{\rA}^{(j)}\}_{j}$, so the last set of equations implies
\begin{eqnarray}\nonumber
\X_{\rA}^{(2m-1)}\X_{\rA}^{(2m)}\Z_{\rA}^{(2m-1)}\Z_{\rA}^{(2m)}\ket{\Psi} &=& \X_{\rA}^{(2m-1)}\X_{\rA}^{(2m)}\left(\Ss_{m,00} + \Ss_{m,01} - \Ss_{m,10} - \Ss_{m,11}\right)\ket{\Psi}\\ \nonumber
&=& \left(\Ss_{m,00} + \Ss_{m,01} - \Ss_{m,10} - \Ss_{m,11}\right)\left(\Ss_{m,00} - \Ss_{m,01} + \Ss_{m,10} - \Ss_{m,11}\right)\ket{\Psi}\\ \nonumber
&=& \left(\Ss_{m,00}- \Ss_{m,01} - \Ss_{m,10} + \Ss_{m,11} \right)\ket{\Psi}\\ \label{XZY}
&=& -\Y_{\rA}^{(2m-1)}\Y_{\rA}^{(2m)}\ket{\Psi}.
\end{eqnarray}
Similarly, from equations (\ref{N0}-\ref{N3}) we get
\begin{equation}\label{XZY+1}
\X_{\rA}^{(2m-1)}\X_{\rA}^{(2m)}\Z_{\rA}^{(2m-1)}\Z_{\rA}^{(2m)}\ket{\Psi} = -\Y_{\rA}^{(2m-1)}\Y_{\rA}^{(2m)}{\Psi}.
\end{equation}
Let us check how eq. (\ref{XZY}) for $m=1$ affects vector $\ket{\xi_\iota} = \tp_{j=1}^n(\idd + (-1)^{\iota(j)}i\Y_{\rA}^{(j)}\X_{\rA}^{(j)})(\idd + \Z_{\rA}^{(j)})\ket{\Psi}$:
\begin{equation*}
\ket{\xi_\iota} = L_{\textrm{rest}}\tp\left(\idd + (-1)^{\iota(1)}i\Y_{\rA}^{(1)}\X_{\rA}^{(1)}\right)\left(\idd + \Z_{\rA}^{(1)}\right)\tp\left(\idd + (-1)^{\iota(2)}i\Y_{\rA}^{(2)}\X_{\rA}^{(2)}\right)\left(\idd + \Z_{\rA}^{(2)}\right) \ket{\Psi}
\end{equation*}
where $L_{\textrm{rest}} = \tp_{j=3}^N (\idd + (-1)^{\iota(j)}i\Y_{\rA}^{(j)}\X_{\rA}^{(j)})(\idd + \Z_{\rA}^{(j)})$. Without loss of generality assume $\iota(1) \neq \iota(2)$. In the following expression for the sake of  brevity we omit $L_{\textrm{rest}}$. Then $\ket{\xi_\iota}$ reads
\begin{eqnarray*}
\ket{\Psi} \pm i\Y_{\rA}^{(2)}\X_{\rA}^{(2)}{\Psi} &+& \Z_{\rA}^{(2)}\ket{\Psi} \pm i\Y_{\rA}^{(2)}\X_{\rA}^{(2)}\Z_{\rA}^{(2)}\ket{\Psi}
\mp i\Y_{\rA}^{(2)}\X_{\rA}^{(2)}\ket{\Psi} + \Y_{\rA}^{(1)}\X_{\rA}^{(1)}\Y_{\rA}^{(2)}\X_{\rA}^{(2)}\ket{\Psi} \mp i\Y_{\rA}^{(1)}\X_{\rA}^{(1)}\Z_{\rA}^{(2)}\ket{\Psi} + \\ + \Y_{\rA}^{(1)}\X_{\rA}^{(1)}\Y_{\rA}^{(2)}\X_{\rA}^{(2)}\Z_{\rA}^{(2)}\ket{\Psi} &+& 
\Z_{\rA}^{(1)}\ket{\Psi} \pm i\Z_{\rA}^{(1)}\Y_{\rA}^{(2)}\X_{\rA}^{(2)}\ket{\Psi} + \Z_{\rA}^{(1)}\Z_{\rA}^{(2)}\ket{\Psi} \pm  i\Z_{\rA}^{(1)}\Y_{\rA}^{(2)}\X_{\rA}^{(2)}\Z_{\rA}^{(2)}\ket{\Psi}
\mp i\Y_{\rA}^{(1)}\X_{\rA}^{(1)}\Z_{\rA}^{(1)}\ket{\Psi} + \\ + \Y_{v}^{(1)}\X_{\rA}^{(1)}\Z_{\rA}^{(1)}\Y_{\rA}^{(2)}\X_{\rA}^{(2)}\ket{\Psi} &+& \mp i\Y_{\rA}^{(1)}\X_{\rA}^{(1)}\Z_{\rA}^{(1)}\Z_{\rA}^{(2)}\ket{\Psi} + \Y_{\rA}^{(1)}\X_{\rA}^{(1)}\Z_{\rA}^{(1)}\Y_{\rA}^{(2)}\X_{\rA}^{(2)}\Z_{\rA}^{(2)}\ket{\Psi}.
\end{eqnarray*}
Let us rearrange eq. (\ref{XZY}) for the case $m=1$. It can be as a sum of eight terms each being equal to $0$ as shown below.
\begin{align}\label{lancel}
\begin{split}
\ket{\Psi} + \Y_{\rA}^{(1)}\X_{\rA}^{(1)}\Z_{\rA}^{(1)}\Y_{\rA}^{(2)}\X_{\rA}^{(2)}\Z_{\rA}^{(2)}\ket{\Psi} = 0, &\qquad  \Y_{\rA}^{(2)}\X_{\rA}^{(2)}\ket{\Psi} + \Y_{\rA}^{(1)}\X_{\rA}^{(1)}\Z_{\rA}^{(1)}\Z_{\rA}^{(2)}\ket{\Psi} = 0, \\
\Z_{\rA}^{(2)}\ket{\Psi} + \Y_{\rA}^{(1)}\X_{\rA}^{(1)}\Z_{\rA}^{(1)}\Y_{\rA}^{(2)}\X_{\rA}^{(2)}\ket{\Psi} = 0, &\qquad  \Y_{\rA}^{(2)}\X_{\rA}^{(2)}\Z_{\rA}^{(2)}\ket{\Psi} + \Y_{\rA}^{(1)}\X_{\rA}^{(1)}\Z_{\rA}^{(1)}\ket{\Psi} = 0, \\ 
\Y_{\rA}^{(1)}\X_{\rA}^{(1)}\ket{\Psi} + \Z_{\rA}^{(1)}\Y_{\rA}^{(2)}\X_{\rA}^{(2)}\Z_{\rA}^{(2)}\ket{\Psi} = 0, &\qquad  \Y_{\rA}^{(1)}\X_{\rA}^{(1)}\Y_{\rA}^{(2)}\X_{\rA}^{(2)}\ket{\Psi} + \Z_{\rA}^{(1)}\Z_{\rA}^{(2)}\ket{\Psi} = 0, \\
\Y_{\rA}^{(1)}\X_{\rA}^{(1)}\Z_{\rA}^{(2)}\ket{\Psi} + \Z_{\rA}^{(1)}\Y_{\rA}^{(2)}\X_{\rA}^{(2)}\ket{\Psi} = 0,&\qquad  \Y_{\rA}^{(1)}\X_{\rA}^{(1)}\Y_{\rA}^{(2)}\X_{\rA}^{(2)}\Z_{\rA}^{(2)}\ket{\Psi} + \Z_{\rA}^{(1)}\ket{\Psi} = 0.
\end{split}
\end{align}
To obtain these equations we used eq. (\ref{XZY}), commutation relations between the operators of different main parties, expressions \eqref{dodatnaa},  and the fact that operators $R_{\rA}^{(k)}$ for $R \in \{X,Y,Z\}$ are unitary and Hermitian.\\ 

Assuming $\iota(1) \neq \iota(2)$ implies $\ket{\xi_\iota} = 0$. Analogously,  eq. (\ref{XZY}) can be used to prove that $\ket{\iota} = 0$ if there exists $m$ such that $\iota(2m-1) \neq \iota(2m)$.  Similarly, starting from eq. (\ref{XZY+1}) one can show that $\ket{\xi_{\iota}} = 0$ if there exists a value of $m$ such that $\iota(2m) \neq \iota(2m+1)$. The only two states $\ket{\iota}$ which satisfy $\iota(2m-1) = \iota(2m) = \iota(2m+1)$ are $\ket{\iota} = \ket{0\dots 0}$ and $\ket{\iota} = \ket{1\dots 1}$. This means that
\begin{equation}\label{last}
\ket{\textrm{aux}} = \ket{\xi_{0}}_{\mathbb{A}\mathbb{B}}\tp\ket{0\dots 0}_{\mathbb{A}''\mathbb{B}''} + \ket{\xi_{1}}_{\mathbb{A}\mathbb{B}}\tp\ket{1\dots 1}_{\mathbb{A}\mathbb{B}''},
\end{equation}
With this form of the junk state $\ket{\textrm{aux}}$ we see that the right hand side of eq. \eqref{foound} is an entangled state across partition $A_1B_1|A_2B_2|\cdots|A_NB_N$, while the left hand side is a product state across the same bipartition. This implies that either $\ket{\xi_0} = 0$ or $\ket{\xi_1} = 0$. With this insight the final self-testing expression is:

\begin{align}\label{foound1}
    V\left(\ket{\Psi}_{\mathbb{A}\mathbb{B}}\otimes\ket{0\cdots 0}_{\mathbb{A}'\mathbb{A}''{\mathbb{B}}'{\mathbb{B}}''}\right) &= \begin{cases}
    \ket{\xi_0}_{\mathbb{A}\mathbb{B}}\ket{0\cdots0}_{\mathbb{A}''{\mathbb{B}}''}\bigotimes_{j=0}^{N-1}\ket{\phi^+}_{A_j'B_{j}'}\\
    \ket{\xi_1}_{\mathbb{A}\mathbb{B}}\ket{1\cdots1}_{\mathbb{A}''{\mathbb{B}}''}\bigotimes_{j=0}^{N-1}\ket{\phi^+}_{A_j'B_{j}'}
    \end{cases}
    \\ \label{STNpairs1}
    V\left(\bigotimes_{j=1}^N\W_\rB^{(j)}\ket{\Psi}_{\mathbb{A}\mathbb{B}}\otimes\ket{0\cdots 0}_{\mathbb{A}'\mathbb{A}''{\mathbb{B}}'{\mathbb{B}}''}\right) &= \begin{cases} \ket{\xi_0}_{\mathbb{A}\mathbb{B}}\ket{0\cdots0}_{\mathbb{A}''{\mathbb{B}}''}\bigotimes_{j=1}^{N}\sigma_w^{B_j'}\ket{\phi^+}_{A_j'B_{j}'},\\
    \ket{\xi_1}_{\mathbb{A}\mathbb{B}}\ket{1\cdots1}_{\mathbb{A}''{\mathbb{B}}''}\bigotimes_{j=1}^{N}{\sigma^*_w}^{B_j'}\ket{\phi^+}_{A_j'B_{j}'},
    \end{cases}
\end{align}
To conclude, in this appendix we showed that $N$ maximally entangled pairs of qubits can be extracted if the auxiliary party maximally violates three CHSH inequalities with each of the main parties. The same isometry that extracts $N$ maximally entangled pairs extracts either the reference measurements, or their complex conjugates, but it is impossible to distinguish between the two cases.

\section{A useful proposition}\label{app:Step2}

In this section, we prove an important proposition dealing with quantum state tomography with partially characterised measurements. In the previous sections we saw that measurements can be self-tested only up to complex conjugation. Consider measurement correlations obtained by applying a tomographically complete set of measurements on a pure entangled state $\ket{\tilde{\psi}}$. The proposition tells that if the same correlations are obtained when all the measurements from the tomographically complete set are characterized up to complex conjugation, the underlying state is in the most general case a coherently controlled superposition of $\ket{\tilde{\psi}}$ and $\ket{\tilde{\psi}^*}$. 
This proposition is necessary for obtaining the final self-testing result. Before stating the proposition, let us introduce some notation. 
First, given an (unknown) reference state $\ket{\tilde{\psi}}$ and measurements $\{\{\M'_{a_i|x_i}\}_{a_i}\}_{i,x_i}$, let us define the reference behaviour:
\begin{equation}
    \tilde{p}(\ab|\xb) = \bra{\tilde{\psi}}\M'_{a_1|x_1}\tp \cdots \tp \M'_{a_N|x_N}\ket{\tilde{\psi}}.
\end{equation}
If the (known) set of measurements $\{\{\M'_{a_i|x_i}\}_{a_i}\}_{i,x_i}$  is tomographically complete, knowing the full probability distribution $\{\tilde{p}(\ab|\xb)\}$ allows us to retrieve the exact form of the reference state $\ket{\tilde{\psi}}$. For the fixed set of measurements $\{\{\M'_{a_i|x_i}\}_{a_i}\}_{i,x_i}$ define the set of quantum realisable behaviours 
\begin{equation}
P_Q\left(\{\M'_{a_i|x_i}\}\right) = \bigg\{\{\tilde{p}(\ab|\xb)\}\bigg|\exists \rho \geq 0, \textrm{tr}\rho = 1, \tilde{p}(\ab|\xb) = \textrm{tr}\left[\left(\M'_{a_1|x_1}\tp \cdots \tp \M'_{a_N|x_N}\right)\rho\right] \bigg\}.
\end{equation}
Note that $P_Q(\{\M'_{a_i|x_i}\})$ is convex and its extremal points are the behaviours obtained by measuring pure states. Before stating the theorem, let us recall that the pure state $\ket{\tilde{\psi}} \in \mathcal{H}^{\rA_1}\otimes\cdots\otimes\mathcal{H}^{A_N}$ is genuinely multipartite entangled~\cite{Seevinck2001} if it is not product across any bipartition.  

\begin{tcolorbox}
\begin{prop}[Formal version]\label{prop}
Consider an unknown state $\rho^{\mathbb{A}'\mathbb{A}\mathbb{A}''}$  and  local measurements of the form 
\begin{equation}
    \M_{a_i|x_i} = {\M'}_{a_i|x_i}^{\rA'_i}\tp{\hat{\M}_{+}}^{{\rA}_i{\rA}''_i} + {{\M'}^*}_{a_i|x_i}^{\rA_i}\tp{\hat{\M}_{-}}^{{\rA}_i{\rA}''_i},
\end{equation}
where the set $\{\M'_{a_i|x_i}\}$ is fully characterized tomographically complete and $\M_+ + \M_- = \idd$, that give rise to the behaviour $\{\tilde{p}(\ab|\xb)\}$. If there exists a genuinely multipartite entangled state $\ket{\tilde{\psi}}$ such that
\begin{equation}\label{escofier}
    \tilde{p}(\ab|\xb) = \bra{\tilde{\psi}}\M'_{a_1|x_1}\tp \cdots \tp \M'_{a_N|x_N}\ket{\tilde{\psi}}, \qquad \forall\ab,\xb
\end{equation}
the purification of $\rho^{\mathbb{A}'\mathbb{A}\mathbb{A}''}$ must be  
\begin{equation}
    \ket{\tilde{\Psi}}^{\mathbb{A}'\mathbb{A}\mathbb{A}''\rP} = \ket{\tilde{\psi}}^{\mathbb{A}'}\tp\ket{\xi_+}^{\mathbb{A}\mathbb{A}''\rP} + \ket{\tilde{\psi}^*}^{\mathbb{A}'}\tp\ket{\xi_-}^{\mathbb{A}\mathbb{A}''\rP},
\end{equation}
where $\braket{\xi_+}{\xi_-} = 0$ and $\|\ket{\xi_+}\| + \|\ket{\xi_-}\| = 1$.
\end{prop}
\end{tcolorbox}
\begin{proof}

We start the proof by providing a lemma that has been proven in various places, including \cite{Hildebrand} and \cite{Johnston}. In the following, we fix a computational basis. For a vector $\ket{r}$, we define its complex conjugate vector  $\ket{r^*}$ obtained by complex conjugating all coordinates of $\ket{r}$ in the computational basis.
\begin{lem}\label{lemmaSM}
Let $\ket{\phi}$ be a pure quantum state with Schmidt decomposition $\ket{\phi} = \sum_{i=1}^m\lambda_i\ket{s_i}\ket{r_i}$. Then the spectrum of $\proj{\phi}^{T_B}$, where $(\cdot)^{T_B}$ denotes the partial transposition, is:
\begin{align}
    \lambda_i^2,& \qquad \textrm{for}\quad 1 \leq i \leq m, \qquad \textrm{corresponding} \quad \textrm{to} \quad \textrm{the}\quad\textrm{eigenvector}\qquad \ket{s_i}\ket{r_i^*} \\
    \pm\lambda_i\lambda_{i'},& \qquad \textrm{for}\quad 1 \leq i \neq i' \leq m, \qquad \textrm{corresponding} \quad \textrm{to} \quad \textrm{the}\quad \textrm{eigenvector}\qquad (\ket{s_i}\ket{r_{i'}^*} \pm \ket{s_{i'}}\ket{r_{i}^*})/\sqrt{2}.
\end{align}
\end{lem}
\begin{proof}
The partial transpose of $\proj{\phi}$ is 
\begin{equation}
    \proj{\phi}^{T_B} = \sum_{k,k'= 1}^m \lambda_k\lambda_{k'}\ketbra{s_k}{s_{k'}}\ketbra{r_{k'}^*}{r_{k}^*}
\end{equation}
This operator clearly satisfies
\begin{align}
    \proj{\phi}^{T_B}\ket{s_i}\ket{r_i^*} &= \lambda_i^2\ket{s_i}\ket{r_i^*}\\
    \proj{\phi}^{T_B}\frac{\ket{s_i}\ket{r_{i'}^*} \pm \ket{s_{i'}}\ket{r_{i}^*}}{\sqrt{2}} &= \pm\lambda_i\lambda_{i'}\frac{\ket{s_i}\ket{r_{i'}^*} \pm \ket{s_{i'}}\ket{r_{i}^*}}{\sqrt{2}}
\end{align}
This gives in total $m + 2\frac{m(m-1)}{2} = m^2$ orthonormal eigenvectors, hence completes the proof.
\end{proof}

This then gives the following corollary:
\begin{cor}\label{usefulcor}
Let $\rho$ be a density operator acting on the Hilbert space $\mathcal{H}_{A}\otimes\mathcal{H}_{B}$ and $(\cdot)^{T_{B}}$ be the partial transposition operation on the space $\mathcal{H}_{B}$. Then the spectrum of $\rho^{T_{B}}$ has values contained in $[-\frac{1}{2},1]$. Furthermore, if there exists a state $\ket{\xi}\in\mathcal{H}_{A}\otimes\mathcal{H}_{B}$ such that $\bra{\xi}\rho^{T_{B}}\ket{\xi}=1$, then $\rho$ is a separable state of the form
\begin{equation}
\rho=\sum_{j}p_{j}\ket{\phi^{A}_{j}}\bra{\phi^{A}_{j}}\otimes\ket{\phi^{B}_{j}}\bra{\phi^{B}_{j}},
\end{equation}
where $\ket{\phi^{A}_{j}}\in\mathcal{H}_{A}$ and $\ket{\phi^{B}_{j}}\in\mathcal{H}_{B}$.
\end{cor}
\begin{proof}
For a symmetric operator $\sigma$, let $\lambda_{\textrm{max}}(\sigma)$ be its maximal eigenvalue. Consider the spectral decomposition of $\rho=\sum_{i}p_{i}\ket{\phi_{i}}\bra{\phi_{i}}$.
We have:
\begin{eqnarray}\label{eqeig}
\lambda_{\textrm{max}}(\rho^{T_{B}})&\leq&\sum_{i}p_{i}\lambda_{\textrm{max}}(\ket{\phi_{i}}\bra{\phi_{i}}^{T_{B}})\nonumber\\
&\leq&1.
\end{eqnarray}

It follows from Lemma \ref{lemmaSM} that equality in \ref{eqeig} is achieved if and only if there is a single term in the Schmidt decomposition of each state $\ket{\phi_{i}}$. Equivalently, this implies that $\ket{\phi_{i}}$ is a product state of the form $\ket{\phi_{i}^{A}}\ket{\phi_{i}^{B}}$, and thus $\rho$ is separable with respect to the partition into spaces $\mathcal{H}_{A}$ and $\mathcal{H}_{B}$. Therefore, if there exists a state $\ket{\xi}\in\mathcal{H}_{A}\otimes\mathcal{H}_{B}$ such that $\bra{\xi}\rho^{T_{B}}\ket{\xi}=1$, then the maximum eigenvalue of $\rho^{T_{B}}$ must be $1$, and $\rho$ is also separable in this way.

It remains to be shown that the smallest eigenvalue of $\rho^{T_{B}}$, $\lambda_{\textrm{min}}$, is greater than or equal to $-\frac{1}{2}$.
From Lemma \ref{lemmaSM}, there exists $\lambda_{i}$ and $\lambda_{j}$ such that $\lambda_{\textrm{min}}=-\lambda_{i}\lambda_{j}$. Now note that $\lambda_{\textrm{min}}^{2}=\lambda_{i}^{2}\lambda_{j}^{2}$ and $\sum_{k}\lambda_{k}^{2}=1$ (by definition), hence $\lambda_{j}^{2}\leq 1-\lambda_{i}^{2}$.
Thus, $\lambda_{min}^{2}\leq \lambda_{i}^{2}(1-\lambda_{i}^{2})\leq \frac{1}{4}$. Therefore, $\lambda_{\textrm{min}}\geq -\frac{1}{2}$, concluding the proof.
\end{proof}

Let us introduce the notation $\hat{a}_i \in \{+,-\}$ for all $i$, $\hat{\ab} = (\hat{a}_1,\cdots, \hat{a}_N)$, ${\M'}_{a_i|x_i,+} = {\M'}_{a_i|x_i}$ and ${\M'}_{a_i|x_i,-} = {\M'}^*_{a_i|x_i}$. The behaviour $\tilde{p}(\ab|\xb)$ can be written in the following way
\begin{align}\label{born}
    \tilde{p}(\ab|\xb) &= \sum_{\hat{\ab}}\textrm{tr}\left[\bigotimes_{i=1}^N\left({\M'}^{\rA'_i}_{a_i|x_i,\hat{a}_i}\tp\hat{\M}^{{\rA}_i{\rA}''_i}_{\hat{a}_i}\right)\rho\right]\\
    &= \sum_{\hat{\ab}}\alpha_{\hat{\ab}}\textrm{tr}\left[\bigotimes_{i=1}^N{\M'}^{\rA'_i}_{a_i|x_i,\hat{a}_i}\rho_{\hat{\ab}}\right]\\
    &= \sum_{\hat{\ab}}\alpha_{\hat{\ab}}\textrm{tr}\left[\bigotimes_{i=1}^N{\M'}^{\rA'_i}_{a_i|x_i}\rho_{\hat{\ab}}^{T_{\hat{\ab}}}\right]\\
    &=\textrm{tr}\left[\bigotimes_{i=1}^N{\M'}^{\rA'_i}_{a_i|x_i}\left(\sum_{\hat{\ab}}\alpha_{\hat{\ab}}\rho_{\hat{\ab}}^{T_{\hat{\ab}}}\right)\right]\\
    &=\textrm{tr}\left[\bigotimes_{i=1}^N{\M'}^{\rA'_i}_{a_i|x_i}\tilde{\rho}\right]
\end{align}
where, 
\begin{align}
\alpha_{\hat{\ab}} &= \textrm{tr}\left[\bigotimes_{i=1}^N\idd^{\rA'_i} \tp\hat{\M}_{\hat{a}_i}^{{\rA}_i{\rA}''_i}\rho\right] \\
\rho_{\hat{\ab}} &= \frac{1}{\alpha_{\hat{\ab}}}\textrm{tr}_{{\mathbb{A}}\mathbb{A}''}\left[\bigotimes_{i=1}^N\idd^{\rA'_i} \tp\hat{\M}^{{\rA}_i{\rA}''_i}_{\hat{a}_i}\rho\right]\\
\tilde{\rho}&=\sum_{\hat{\ab}}\alpha_{\hat{\ab}}\rho_{\hat{\ab}}^{T_{\hat{\ab}}}
\label{tintin}
\end{align}
and 
$(\cdot)^{T_{\hat{\ab}}}$ denotes the transpose of an operator with respect to all the parties $A_i$ such that $\hat{a}_i = -$. This transposition can be partial or full, depending on the string $\hat{\ab}$: if  $\hat{a}_i = -$ for all $i$ then the whole operator is transposed; if $\hat{a}_i = +$ for all $i$ then no transpose is applied to any of the parties spaces; otherwise, the transposition is partial. For every $\hat{\ab}$ for which the value $\alpha_{\hat{\ab}}$ is non-zero, $\rho_{\hat{\ab}}$ \added{is} a valid quantum state. If $\alpha_{\hat{\ab}} = 0$ we introduce the convention $\rho_{\hat{\ab}} = 0$. Furthermore $\sum_{\hat{\ab}}\alpha_{\hat{\ab}} = 1$. Given eq. \eqref{escofier} and the fact that the set $\{\M'_{a_i|x_i}\}$ is tomographically complete, it must be that
\begin{equation}\label{laure}
    \proj{\tilde{\psi}} = \sum_{\hat{\ab}}\alpha_{\hat{\ab}}\rho_{\hat{\ab}}^{T_{\hat{\ab}}}.
\end{equation}
Let $\Pi_{\tilde{\psi}}:=\tilde{\ket{\psi}}\tilde{\bra{\psi}}$,  $\beta_{\hat{\ab}}:=\tr{\Pi_{\tilde{\psi}}\rho_{\hat{\ab}}^{T_{\hat{\textbf{a}}}}}$ and $\omega_{\hat{\ab}}:=\rho_{\hat{\ab}}^{T_{\hat{\textbf{a}}}}-\beta_{\hat{\ab}}\tilde{\ket{\psi}}\tilde{\bra{\psi}}$. 
Equivalently
\begin{equation}
\rho_{\hat{\ab}}^{T_{\hat{\textbf{a}}}}=\beta_{\hat{\ab}}\tilde{\ket{\psi}}\tilde{\bra{\psi}}+\omega_{\hat{\ab}}.
\end{equation}

By \eqref{laure} we have 
\begin{eqnarray}
\tilde{\ket{\psi}}\tilde{\bra{\psi}}&=&\Pi_{\tilde{\psi}}\tilde{\ket{\psi}}\tilde{\bra{\psi}}\Pi_{\tilde{\psi}}\\
&=&\sum_{\hat{\ab}}\alpha_{\hat{\ab}}\Pi_{\tilde{\psi}}\rho_{\hat{\ab}}^{T_{\hat{\ab}}}\Pi_{\tilde{\psi}}\\
&=&\sum_{\hat{\ab}}\alpha_{\hat{\ab}}\beta_{\hat{\ab}}\tilde{\ket{\psi}}\tilde{\bra{\psi}}.
\end{eqnarray}
This then implies that 
\begin{equation}\label{normalise}
\sum_{\hat{\ab}}\alpha_{\hat{\ab}}\beta_{\hat{\ab}}=1.
\end{equation}
We have $\beta_{\hat{\ab}}\leq\lambda_{\max}(\rho_{\hat{\ab}}^{T_{\hat{\ab}}})$ and $\sum \alpha_{\hat{\ab}} =1$, hence Corollary \ref{usefulcor} and \eqref{normalise} imply that $\beta_{\hat{\ab}}=1$ for all $\hat{\ab}$ such that $\alpha_{\hat{\ab}}\neq 0$. 

Consider a given $\hat{\ab}$ such that $\alpha_{\hat{\ab}}\neq 0$.
We have just proven that 
\begin{equation}
\rho_{\hat{\ab}}^{T_{\hat{\textbf{a}}}}=\tilde{\ket{\psi}}\tilde{\bra{\psi}}+\omega_{\hat{\ab}},
\end{equation}
In the case where $\hat{\ab}=(+,...,+)$ or $\hat{\ab}=(-, ..., -)$, $\rho_{\hat{\ab}}^{T_{\hat{\textbf{a}}}}$ is a density matrix, hence $\omega_{\hat{\ab}}=0$.
Let us now show that this is also the case for the remaining values of $\hat{\ab}$.

Consider the following partition of the two sets of parties $\Omega_{+}:\{i|\hat{a}_{i}=+\}$ and $\Omega_{-}:\{j|\hat{a}_{i}=-\}$. Operator
$\rho_{\hat{\ab}}^{T_{\hat{\ab}}}$ is the partial transpose of $\rho_{\hat{\ab}}$ over space $\Omega_{-}$. 
As $\beta_{\hat{\ab}}=\tilde{\bra{\psi}}\rho_{\hat{\ab}}^{T_{\hat{\ab}}}\tilde{\ket{\psi}}=1$, Corollary \ref{usefulcor} implies that $\rho_{\hat{\ab}}$ is separable across the partition into parties $\Omega_{+}$ and $\Omega_{-}$. In other words
\begin{equation}
\rho_{\hat{\ab}}=\sum_{i}p_{i}\bigotimes_{j\in\Omega_{+}}\ket{\phi^{j}_{i}}\bra{\phi^{j}_{i}}\bigotimes_{k\in\Omega_{-}}\ket{\phi^{k}_{i}}\bra{\phi^{k}_{i}}.
\end{equation}
In particular, for all $\hat{\ab}$, $\rho_{\hat{\ab}}^{T_{\hat{\textbf{a}}}}$ is a density operator, hence as before we have $\rho_{\hat{\ab}}^{T_{\hat{\textbf{a}}}}=\tilde{\ket{\psi}}\tilde{\bra{\psi}}$. 
Now we show that $\alpha_{\hat{\ab}}=0$ unless $\hat{\ab}=(+,...,+)$ or $\hat{\ab}=(-, ..., -)$. Assume it is not the case. The state $\tilde{\ket{\psi}}$ is genuinely multipartite entangled. \emph{i.e.} entangled across all bipartitions of the $N$ parties, and thus whenever $\hat{\ab} \not\in \{(+,\cdots, +),(-,\cdots,-)\}$, $\rho_{\hat{\ab}}=\tilde{\ket{\psi}}\tilde{\bra{\psi}}^{T_{\hat{\textbf{a}}}}$ is not positive due to the PPT criterion for pure states. Hence to avoid the contradiction we conclude that  $\alpha_{\hat{\ab}}=0$ unless $\hat{\ab}=(+,...,+)$ or $\hat{\ab}=(-, ..., -)$. 

Therefore from the expression in \ref{born}, we derive the state $\tilde{\rho}$ to be 
\begin{equation}
\tilde{\rho}=\alpha\rho_{+}+(1-\alpha)\rho_{-}^{*},
\end{equation}
where $\rho_{+}:=\proj{\tilde{\psi}}$ and $\rho_{-}:=\proj{\tilde{\psi}^{*}}$.

Thus the purification of the shared state $\tilde{\rho}$ is
\begin{equation}
    \ket{\tilde{\Psi}}^{\mathbb{A}'\mathbb{A}\mathbb{A}''\rP} = \ket{\tilde{\psi}}^{\mathbb{A}'}\tp\ket{\xi_+}^{\mathbb{A}\mathbb{A}''\rP} + \ket{\tilde{\psi}^*}^{\mathbb{A}'}\tp\ket{\xi_-}^{\mathbb{A}\mathbb{A}''\rP},
\end{equation}
where $\M_+^{\rA_i\rA''_i}\ket{\tilde{\Psi}}^{\mathbb{A}'\mathbb{A}\mathbb{A}''\rP} = \ket{\tilde{\psi}}^{\mathbb{A}'}\tp\ket{\xi_+}^{\mathbb{A}\mathbb{A}''\rP}$ and  $\M_-^{\rA'_i}\ket{\tilde{\Psi}}^{\mathbb{A}'\mathbb{\rA}\mathbb{\rA}''\rP} = \ket{\tilde{\psi}^*}^{\mathbb{A}'}\tp\ket{\xi_-}^{\mathbb{A}\mathbb{A}''\rP}$. Since we do not bound the dimension of the spaces in which the measurement operators $\{\M_+,\M_-\}$ act, they can be considered to be projective, which implies $\braket{\tilde{\psi}}{\tilde{\psi}^*}\braket{\xi_+}{\xi_-} = 0$. Moreover,  $\M_++\M_- = \idd$ implies  $\|\ket{\xi_+}\| + \|\ket{\xi_-}\| = 1$.


\end{proof}

\section{Self-testing }\label{app:Step3}

\subsection{Network assisted self-testing}\label{app:Step3NetworkAssisted}
In this section, we prove Theorem 1 from the main text. First we give a rigorous statement of the theorem. 
\setcounter{thm}{0}
\begin{tcolorbox}
\begin{thm}[Formal version, network-assisted case]
Let the state $\ket{\Psi} \in \mathcal{H}^{\mathbb{A}}\otimes\mathcal{H}^{\mathbb{B}}$, dichotomic observables $\{\rA_j\}_{j=0}^5$ and $\{\rB_j\}_{j=0}^2$ and measurements $\left\{\N_{b_j|\lozenge}^{2j-1,2j}, \N_{b_j|\blacklozenge}^{2j,2j+1}\right\}_{j=1}^{\lfloor N/2\rfloor}$, $\left\{\M_{a_j|\diamond}\right\}_{j=1}^N$ reproduce all reference correlations given in eqs. \eqref{eq:3CHSHViolationNetworkAssisted} and \eqref{eq:TomographyTeleportedStateNetworkAssisted}. Then there exists a CPTP map $\Phi$ such that:
\begin{align}
    \Phi(\ket{\Psi}\otimes\ket{0,\cdots,0}^{\mathbb{B}'\mathbb{B}''}) &= \sum_{a_j}\bigotimes U_{a_j}^{B_j'}\left[\bigotimes_{j=1}^N\M_{a_j|\diamond}\otimes\idd^{\mathbb{B}\mathbb{B}'\mathbb{B}''}V_{\rB}(\ket{\Psi}\otimes \ket{0,\cdots,0}^{\mathbb{B}'\mathbb{B}''})\right]\\
    &= (\ket{\psi_N}^{\mathbb{B}'}\otimes\ket{0\cdots 0}^{\mathbb{B}''} + \ket{\psi_N^*}^{\mathbb{B}'}\otimes\ket{1\cdots 1}^{\mathbb{B}''})\otimes \ket{\textrm{aux}}^{\mathbb{A}\mathbb{B}},
\end{align}
where $U_{a_j}$ are the correcting unitaries defined as $U_{00} = \idd$, $U_{01} = \sigma_z$, $U_{10} = \sigma_x$ and $U_{11} = \sigma_x\sigma_z$, and $V_B$ is the local isometry from Lemma \ref{lemchshNA}.
The junk state $\ket{\mathrm{aux}}$ has the following form:
\begin{equation}\label{eqAuxapp}
\ket{\mathrm{aux}} = \sum_\iota\ket{\xi_\iota}_{\mathbb{A}B}\ket{\iota}_{\mathbb{A}''}\ket{\iota}_{{\mathbb{B}}''},
\end{equation}
where $\iota \in (0,1)^N$.
\end{thm}
\end{tcolorbox}

Lemma \ref{lemchshNA} provides us with self-testing of $N$ maximally entangled pairs of qubits based on the reproduction of correlations given in Eq. \eqref{eq:3CHSHViolationNetworkAssisted}.
Let us now consider the case when every main party gets input $\diamond$, while all the auxiliary parties gets some $N$-trit input. Given that auxiliary parties obtain outputs $b_1, \cdots, b_N$  the full state satisfies the following expression (cf. \eqref{MSTNpairs})
\begin{align}\label{vm}
    V\left(\bigotimes_{j=1}^{N}\N_{b_j|y_j}\ket{\Psi}\otimes\ket{0\cdots 0}_{\mathbb{A}'\mathbb{A}''\mathbb{B}'\mathbb{B}''}\right) =  \sum_{\iota}\ket{\xi_{\iota}}_{\mathbb{A}\mathbb{B}}\otimes\ket{\iota\iota}_{\mathbb{A}''{\mathbb{B}}''}\bigotimes_{j=1}^{N} {{\N'}_{\tilde{b}_j|y_j }}\ket{\phi^+}_{A_j'B'_{j}},
\end{align}
%
where $\tilde{b}_j = b_j \oplus 1$ if $\iota_j = 1$ and $y_j = 2$, and $\tilde{b}_j = b_j$ otherwise. The equation is the same as eq. \eqref{MSTNpairs}, written in terms of the measurement operators, rather than observables. Using $\tilde{b}_j$ instead of $b_j$ allows us to have a more compact equation. Namely, $\tilde{b}_j$ is different than $b_j$ only when the input is $y_j = 2$ and only in the elements of the sum corresponding to $\iota$ such that $\iota_j = 1$. In other words, in the elements of the sum labelled by $\iota$ such that $\iota_j = 1$, measurement $\N_{b_j|2}$ is mapped to $\N'_{1-b_j|2}$. We have also introduced the following notation for the reference measurements
\begin{align}
    \N'_{b|0} \equiv \frac{\idd + (-1)^b\sigma_z}{2}, \qquad \N'_{b|1} \equiv \frac{\idd + (-1)^b\sigma_x}{2}, \qquad \N'_{b|2} \equiv \frac{\idd + (-1)^b\sigma_y}{2},
\end{align}
The expression \eqref{vm} together with  \eqref{josjedna}
implies that 
\begin{align}
    V_\rB^{(j)}\left(\N_{b_j|y_j}\otimes\idd^{B_j'B_j''}\right){V_\rB^{(j)}}^\dagger =  \idd^{{B_j}}\otimes\proj{0}^{B_j''}\otimes {\N'}_{{b}_j|y_j }^{B_j'} + \idd^{{B_j}}\otimes\proj{1}^{B_j''}\otimes {\N'}_{\tilde{b}_j|y_j }^{B_j'},
\end{align}
holds on the support of $\rho_B = \tr_\mathbb{A}[\proj{\Psi}]\otimes\proj{0\cdots0}_{\mathbb{B}'\mathbb{B}''}$.
Equivalently, for a tensor product of physical measurement operators, the following relation holds on the support of $\rho_B$
\begin{align}
    V_B\left(\bigotimes_{j=1}^N\N_{b_j|y_j}\otimes\idd^{\mathbb{B}'\mathbb{B}''}\right)V_B^\dagger =  \sum_\iota\idd^{\mathbb{B}}\otimes\proj{\iota}^{\mathbb{B}''} \bigotimes_{j=1}^N{N'}_{\tilde{b}_j|y_j }^{B_j'}.
\end{align}
Observe now the following expression
\begin{align}
    p(\bb|\ab,\diamond,y) &= \frac{1}{p(\ab|\diamond)}\textrm{tr}\left[\left(\bigotimes_{j=1}^N\M_{a_j|\diamond}\otimes\N_{\bb|\yb}\right)\proj{\Psi}\right]\\
    &= \frac{1}{p(\ab|\diamond)}\textrm{tr}\left[\sum_\iota\idd^{\mathbb{B}}\otimes\proj{\iota}^{\mathbb{B}''} \bigotimes_{j=1}^N{\N'}_{\tilde{b}_j|y_j }^{B_j'} \textrm{tr}_{\mathbb{A}}\left(\bigotimes_{j=1}^N\M_{a_j|\diamond}\otimes\idd^{\mathbb{B}\mathbb{B}'\mathbb{B}''}V_\rB\left(\proj{\Psi}\otimes \proj{0,\cdots,0}^{\mathbb{B}'\mathbb{B}''}\right)\right)\right]\\
    &= \textrm{tr}\left[\sum_\iota\proj{\iota}^{\mathbb{B}''} \bigotimes_{j=1}^N{\N'}_{\tilde{b}_j|y_j }^{B_j'}\tilde{\Psi}^{\mathbb{B}'\mathbb{B}'' }_{\ab}\right],
\end{align}
where we introduced notation:
\begin{equation}
    \tilde{\Psi}_{\ab}^{\mathbb{B}'\mathbb{B}'' } = \frac{1}{p(\ab|\diamond)}\textrm{tr}_{\mathbb{A}\mathbb{B}}\left[\bigotimes_{j=1}^N\M_{a_j|\diamond}\otimes\idd^{\mathbb{B}\mathbb{B}'\mathbb{B}''}V_\rB\left(\proj{\Psi}\otimes \proj{0\cdots0}^{\mathbb{B}'\mathbb{B}''}\right)\right],
\end{equation}
and $\ab = a_1,\cdots, a_N$ and $\N_{\bb|\yb} = \bigotimes_{j=1}\N_{b_j|y_j}^{(j)}$.

Assume now that the probabilities $p(\bb|\ab,\diamond,\yb)$ satisfy the condition
\begin{equation}\label{condition}
    p(\bb|\ab,\diamond,\yb) = \textrm{tr}\left[U_{\ab}\proj{\psi_N}U_{\ab}^\dagger\bigotimes_{j=0}^{N-1}\N'_{b_j|y_j}\right]
\end{equation}
and $p(\ab|\yb) = 1/2^N$, for all $\ab$ and $\yb$, \emph{i.e.} the reference correlations given in \eqref{eq:TomographyTeleportedStateNetworkAssisted}. The correcting unitaries are defined as $U_{\ab} = \bigotimes_{a_j}U_{a_j}$ where:
\begin{equation}
    U_{a_j} = \sigma_z^{a_{j,1}}\sigma_x^{a_{j,2}},
\end{equation}
and $a_{j,1}$ and $a_{j,2}$ are two bits defining the two-bit outcome $a_j$. Thus, Proposition \ref{prop} implies that for every set of outcomes $\ab$ the state $\ket{\tilde{\Psi}_\ab}^{\mathbb{B}'\mathbb{B}'' }$ satisfies
\begin{equation}
    \ket{\tilde{\Psi}_\ab}^{\mathbb{B}'\mathbb{B}''} = U_\ab^{\mathbb{B}'}\otimes\idd^{\mathbb{B}''}(\ket{\psi_N}^{\mathbb{B}'}\otimes\ket{0\cdots 0}^{\mathbb{B}''} + \ket{\psi_N^*}^{\mathbb{B}'}\otimes\ket{1\cdots 1}^{\mathbb{B}''}).
\end{equation}
This result is all what we need to find the final isometry, inspired by the isometry used in \cite{MDIST}. The isometry $\Phi$ defined in the following way
\begin{align}
    \Phi(\ket{\Psi}\otimes\ket{0,\cdots,0}^{\mathbb{B}'\mathbb{B}''}) &= \sum_{a_j} \left[\bigotimes_{j=1}^N\M_{a_j|\diamond}\otimes U_{a_j}^{B_j'}\otimes\idd^{\mathbb{B}\mathbb{B}''}V_{\rB}(\ket{\Psi}\otimes \ket{0,\cdots,0}^{\mathbb{B}'\mathbb{B}''})\right]\\
    &= (\ket{\psi_N}^{\mathbb{B}'}\otimes\ket{0\cdots 0}^{\mathbb{B}''} + \ket{\psi_N^*}^{\mathbb{B}'}\otimes\ket{1\cdots 1}^{\mathbb{B}''})\otimes \ket{\textrm{aux}}^{\mathbb{A}\mathbb{B}},
\end{align}
gives us exactly the result we needed. 

\subsection{Self-testing in the fully network-assisted scenario}\label{app:Step3FullyNetworkAssisted}

In this section, we prove Theorem 2 from the main text, which we state here in a its rigorous form:
\setcounter{thm}{0}
\begin{tcolorbox}
\begin{thm}[Formal version, fully-network-assisted case]
Let the state $\ket{\Psi} \in \mathcal{H}^{\mathbb{A}}\otimes\mathcal{H}^{\mathbb{B}}$, dichotomic observables $\{\rA_j\}_{j=0}^5$ and $\{\rB_j\}_{j=0}^2$ and measurements $\left\{\N_{b_j|\lozenge}^{2j-1,2j}, \N_{b_j|\blacklozenge}^{2j,2j+1}\right\}_{j=1}^{\lfloor N/2\rfloor}$, $\left\{\M_{a_j|\diamond}\right\}_{j=1}^N$ reproduce all reference correlations given in eqs. \eqref{eq:3CHSHViolationFullyNetworkAssisted}, \eqref{eq:AlignFullyNetworkAssisted} and \eqref{eq:TomographyTeleportedStateFullyNetworkAssisted}. Then there exists a CPTP map $\Phi$ such that:
\begin{align}
    \Phi\left(\frac{1}{p(\bb|\yb)}\tr_{\mathbb{B}\mathbb{B}'\mathbb{B}''}\left[    V\left(\otimes_{j=1}^N\N_{b_j|y_j}^{(j)}\ket{\Psi}_{\mathbb{A}\mathbb{B}}\otimes\ket{0\cdots 0}_{\mathbb{A}'\mathbb{A}''{\mathbb{B}}'{\mathbb{B}}''}\right)\right]\right) = \begin{cases} \xi^{\mathbb{A}\mathbb{A}''\mathbb{A}'}\otimes\ket{\psi_N}^{\bar{\mathbb{A}}},\\
    \xi^{\mathbb{A}\mathbb{A}''\mathbb{A}'}\otimes\ket{\psi_N^*}^{\bar{\mathbb{A}}},
    \end{cases}
\end{align}
where $V$ is the local isometry figuring in \eqref{josjednab}.
\end{thm}
\end{tcolorbox}


As announced, exploiting the network constraints, a multipartite state can be self-tested up to complex conjugation. That is, contrary to the previous case, the self-testing isometry does not map a physical state to a controlled coherent superposition of the reference state and the complex conjugate of the reference state. Rather, it maps the state either to the reference state or to its complex conjugate. We follow here the result of Lemma \ref{lemmachshFNA}. Consider the expression given therein (cf. \eqref{MSTNpairsb}) 
\begin{align}\label{malotru}
    V\left(\bigotimes_{j=1}^N\N_{b_j|y_j}^{(j)}\ket{\Psi}_{\mathbb{A}\mathbb{B}}\otimes\ket{0\cdots 0}_{\mathbb{A}'\mathbb{A}''{\mathbb{B}}'{\mathbb{B}}''}\right) &= \begin{cases} \ket{\xi_0}_{\mathbb{A}\mathbb{B}}\ket{0\cdots0}_{\mathbb{A}''{\mathbb{B}}''}\bigotimes_{j=1}^{N}\N'_{b_j|y_j}\ket{\phi^+}_{A_j'B_{j}'},\\
    \ket{\xi_1}_{\mathbb{A}\mathbb{B}}\ket{1\cdots1}_{\mathbb{A}''{\mathbb{B}}''}\bigotimes_{j=1}^{N}{\N'}_{b_j|y_j}^{*}\ket{\phi^+}_{A_j'B_{j}'},
    \end{cases}
\end{align}
To avoid long expressions with the two cases, we concentrate on the first, since the derivation is the same, up to complex conjugation. Observe the following expression:
\begin{align}
    p(\ab|\bb,\diamond,\yb) &= \frac{1}{p(\bb|\yb)}\textrm{tr}\left[\left(\bigotimes_{j=1}^N\M_{a_j|\diamond}\otimes\N_{b_j|y_j}\right)\ket{\Psi}\right]\\
    &= \frac{1}{p(\bb|\yb)}\textrm{tr}\left[V_\rA^{\dagger}\left(\bigotimes_{j=1}^N\M_{a_j|\diamond}\otimes\idd^{\mathbb{A}'\mathbb{A}''}\right)\ket{\xi_0}_{\mathbb{A}\mathbb{B}}\ket{0\cdots0}_{\mathbb{A}''{\mathbb{B}}''}\bigotimes_{j=1}^{N}\N'_{b_j|y_j}\ket{\phi^+}_{A_j'B_{j}'}\right]\\
    &= \textrm{tr}\left[V_\rA^{\dagger}\left(\bigotimes_{j=1}^N\M_{a_j|\diamond}\otimes\idd^{\mathbb{A}'\mathbb{A}''}\right)\xi^0_{\mathbb{A}}\otimes\proj{0\cdots0}_{\mathbb{A}''}\bigotimes_{j=1}^{N}\nu_{b_j|y_j}^{A_j'}\right],
\end{align}
where $\xi^0_{\mathbb{A}} = \tr_{\mathbb{B}}[\proj{\xi_0}_{\mathbb{A}\mathbb{B}}]$ and $\nu_{b_j|y_j}^{A_j'} = \tr_{B_j}[\N'_{b_j|y_j}\proj{\phi^+}_{A_j'B_{j}'}]$. Introducing the effective measurement operator
\begin{equation}\label{effmeasop}
    \tilde{\M}_{\ab}^{\mathbb{A}'} = \tr_{\mathbb{A}\mathbb{A}''}\left[V_\rA^{\dagger}\left(\bigotimes_{j=1}^N\M_{a_j|\diamond}\otimes\idd^{\mathbb{A}'\mathbb{A}''}\right)\xi^0_{\mathbb{A}}\otimes\proj{0\cdots0}_{\mathbb{A}''}\otimes\idd^{\mathbb{A}'}\right],
\end{equation}
we obtain
\begin{equation}
    p(\ab|\bb,\diamond,\yb) = \tr\left[\bigotimes_{j=1}^N\nu_{b_j|y_j}^{A_j'} \tilde{\M}_{\ab}^{\mathbb{A}'}\right]
\end{equation}

Since the set $\{\otimes_{j=1}^N\nu_{b_j|y_j}^{A_j'}\}_{\bb,\yb}$ is tomographically complete, the set of probabilities $\{p(\ab|\bb,\diamond,\yb)\}_{\ab,\bb,\yb}$ allows us to exactly reconstruct the effective measurement operators $\{\tilde{\M}_{\ab}\}_{\ab}$. If the operators have the form:
\begin{equation}\label{mdiselftest}
    \tilde{\M}_{\ab} = \frac{1}{d^N}U_{\ab}^\dagger\proj{\psi_N}U_{\ab}U_{\ab}.
\end{equation}
The correcting unitaries are defined as  $U_{\ab} = \bigotimes_{a_j}U_{a_j}$ where:
\begin{equation}
    U_{a_j} = \sigma_z^{a_{j,1}}\sigma_x^{a_{j,2}},
\end{equation}
and $a_{j,1}$ and $a_{j,2}$ are two bits defining the two-bit outcome $a_j$. With the result \eqref{mdiselftest}, we can directly use the isometry from \cite{MDIST} to extract the reference state $\ket{\psi_N}$. One branch of the isometry $\Phi_j$ is given on Fig. \ref{isomulti}. The isometry $\Phi = \otimes_{j=1}^N\Phi_j$ is applied to the state $\frac{1}{p(\bb|\yb)}\tr_{\mathbb{B}\mathbb{B}'\mathbb{B}''}[    V(\otimes_{j=1}^N\N_{b_j|y_j}^{(j)}\ket{\Psi}_{\mathbb{A}\mathbb{B}}\otimes\ket{0\cdots 0}_{\mathbb{A}'\mathbb{A}''{\mathbb{B}}'{\mathbb{B}}''})]$. Using eqs. \eqref{effmeasop} and \eqref{mdiselftest} we obtain:
\begin{equation}
    \Phi\left(\frac{1}{p(\bb|\yb)}\tr_{\mathbb{B}\mathbb{B}'\mathbb{B}''}\left[    V\left(\otimes_{j=1}^N\N_{b_j|y_j}^{(j)}\ket{\Psi}_{\mathbb{A}\mathbb{B}}\otimes\ket{0\cdots 0}_{\mathbb{A}'\mathbb{A}''{\mathbb{B}}'{\mathbb{B}}''}\right)\right]\right) = \xi^{\mathbb{A}\mathbb{A}''\mathbb{A}'}\otimes\ket{\psi_N}^{\bar{\mathbb{A}}},  
\end{equation}
where $\bar{\mathbb{A}} = \bar{A}_1\cdots \bar{A}_N$. Taking into account the second case from \eqref{malotru} we obtain the final result:
\begin{align}
    \Phi\left(\frac{1}{p(\bb|\yb)}\tr_{\mathbb{B}\mathbb{B}'\mathbb{B}''}\left[    V\left(\otimes_{j=1}^N\N_{b_j|y_j}^{(j)}\ket{\Psi}_{\mathbb{A}\mathbb{B}}\otimes\ket{0\cdots 0}_{\mathbb{A}'\mathbb{A}''{\mathbb{B}}'{\mathbb{B}}''}\right)\right]\right) = \begin{cases} \xi^{\mathbb{A}\mathbb{A}''\mathbb{A}'}\otimes\ket{\psi_N}^{\bar{\mathbb{A}}},\\
    \xi^{\mathbb{A}\mathbb{A}''\mathbb{A}'}\otimes\ket{\psi_N^*}^{\bar{\mathbb{A}}}
    \end{cases}
\end{align}

\begin{figure}
  \centerline{
    \begin{tikzpicture}[thick]
    %
    \tikzstyle{operator} = [draw,fill=white,minimum size=1.5em] 
    \tikzstyle{operator2} = [draw,fill=white,minimum height=1.8cm]
    \tikzstyle{phase} = [fill,shape=circle,minimum size=5pt,inner sep=0pt]
    \tikzstyle{circlewc} = [draw,minimum width=0.3cm]
    %
    \node at (0,0) (q1) {$\ket{0}^{\tilde{\rA}_j}$};
    \node at (0,-1) (q2) {$\ket{0}^{{\rA}_j'}$};
    \node at (0,-2.5) (q3) {$\ket{\Psi}^{\rA_0\rA_1 \cdots \rA_{N-1}}$};
    \node at (0, -2) (qex1) {};
    \node at (3.40, -1) (qex3) {};
    \node at (3.40, -2) (qex4) {};
    \node at (2.80, -1) (qex5) {};
    \node at (2.80, -2) (qex6) {};
    %
    \node[operator] (op11) at (1,0) {$H$} edge [-] (q1);
    %
    \node[circlewc] (phase11) at (2,-1) {$\sigma_x$} edge [-] (q2);
    \node[phase] (circlewc12) at (2,0) {} edge [-] (op11);
    \draw[-] (phase11) -- (circlewc12);
    \draw[-] (phase11) -- (qex5);
    \draw[-] (qex1)-- (qex6);
    %
    \node[operator] (op31) at (3,0) {$U_{a_j}$} edge [-] (circlewc12);
    \node[operator2] (operator32) at (3, -1.5) {$\M_{a_j|\diamond}$}; 
[-] (circlewc42);
    \draw[-] (op31) -- (operator32);
    %
    \node (end1) at (5,0) {} edge [-] (op31);
    \node (end2) at (5,-1) {} edge [-] (qex3);
    \node (end3) at (5,-2) {} edge [-] (qex4);
    \end{tikzpicture}
  }
  \caption{
   A circuit representing one branch of the isometry $\Phi_j$. $H$ is the Hadamard gate, defined as $H\ket{k} = (\ket{0}+(-1)^k\ket{1})/\sqrt{2}$. It takes as an input the state $\ket{\Psi}^{\rA_1 \cdots \rA_N}$ and each party performs a unitary operation $U_{a_j}$ conditioned on the outcome of the measurement $\M_{a_j}$. 
 }
  \label{isomulti}
\end{figure}
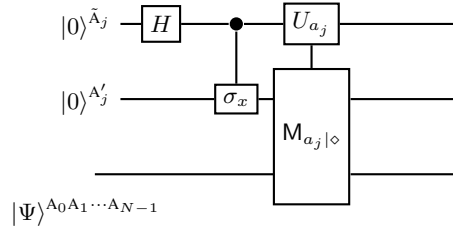

\section{Extension to qudit states}\label{app:ExtensionQudits}

In previous appendices we gave the procedures for network assisted and fully network assisted self-testing of all qubit pure entangled states. The procedures can readily be generalized to self-testing of all pure entangled states whose local dimensions are larger than two. The generalization procedure relies on the method for device-independent witnessing of all entangled states presented in \cite{Bowles_2018PRL}. In this appendix we give the outline of the proof in the network-assisted scenario. 

The reference experiment for self-testing a qudit $N$-partite state $\ket{\psi_d}$ involves $N$ main parties sharing the reference state, and $N$ auxiliary parties, each sharing a maximally entangled pair of qudits $\ket{\phi_+} = \frac{1}{\sqrt{d}}\sum_{i=0}^{d-1}\ket{ii}$ with one of the main parties. 
If the main parties perform qudit Bell state measurements on their shares of $\ket{\psi_d}$ and $\ket{\phi_+}$, they can teleport the reference state to the auxiliary parties. On their side, the auxiliary parties can apply a tomographically complete set of measurements on the teleported state and learn its form.

In a similar way like in the qubit case, all steps of the procedure have to be self-tested. The first step is self-testing of maximally entangled pair of qudits, and a tomographically complete set of measurements performed by an auxiliary party. The procedure for this is given in \cite{Bowles_2018}. Here we just recapitulate that this can be done by first encoding $\ket{\psi_d}$ into a larger Hilbert space $\otimes_{i=1}^N\mathcal{H}_i$, where $\mathcal{H}_i = \mathbb{C}^k$ and $k = 2^{\ceil{\log_2{d}}}$. Then, for self-testing of every $\ket{\psi_d}$ and a tomographically complete set of measurements we use the maximal violation of $\ceil{\log_2{d}}$ $3$-CHSH inequalities in parallel. The details of self-testing procedure  are described in \cite{Bowles_2018}. This procedure characterises a tomographically complete set of measurements applied by the auxiliary parties, up to complex conjugation. Then, using Proposition \ref{prop} it is possible to self-test the state teleported to the auxiliary parties if the physical correlations correspond to those obtained when performing the tomographically complete set of measurements on the teleported reference state.


\end{document}